\documentclass[11pt]{amsart}

\pdfoutput = 1

\usepackage[english]{babel}
\usepackage[autostyle]{csquotes}
\usepackage{amsmath,amssymb,amsthm}
\usepackage{amsfonts}
\usepackage{amsxtra}

\usepackage{array,dcolumn}
\usepackage{graphicx}
\usepackage[hyperfootnotes=true,
        pdffitwindow=true,
        plainpages=false,
        pdfpagelabels=true,
        pdfpagemode=UseOutlines,
        pdfpagelayout=SinglePage,
                hyperindex,]{hyperref}

\usepackage{lmodern}
\usepackage[T1]{fontenc}
\usepackage[utf8]{inputenc}
\usepackage[activate={true,nocompatibility},spacing,kerning]{microtype}
\usepackage{bm}
\usepackage{bbm}
\usepackage{todonotes}
\usepackage[sc,osf]{mathpazo}
\microtypecontext{spacing=nonfrench}

  \calclayout

\newcommand{\mathsym}[1]{{}}
\newcommand{\unicode}[1]{{}}

\theoremstyle{plain}
\newtheorem{theorem}{Theorem}

\newtheorem{proposition}[theorem]{Proposition}

\theoremstyle{definition}

\theoremstyle{remark}
\newtheorem{remark}[theorem]{Remark}

\newcommand{\R}{\mathbb R}
\newcommand{\C}{\mathbb C}

\numberwithin{equation}{section}

\numberwithin{theorem}{section}

\numberwithin{figure}{section}

\begin{document}


\title[]{Progress on the study of the Ginibre ensembles I: {G{\SMALL in}UE}}

\author{Sung-Soo Byun}
\address{Center for Mathematical Challenges, Korea Institute for Advanced Study, 85 Hoegiro, Dongdaemun-gu, Seoul 02455, Republic of Korea}
\email{sungsoobyun@kias.re.kr}

\author{Peter J. Forrester}
\address{School of Mathematics and Statistics, 
University of Melbourne, Victoria 3010, Australia}
\email{pjforr@unimelb.edu.au}

\date{}


\begin{abstract}
The Ginibre unitary ensemble (GinUE) consists of $N \times N$ random matrices with independent complex standard Gaussian entries. This was introduced in 1965 by Ginbre, who showed that the eigenvalues form a determinantal point process with an explicit correlation kernel, and after scaling they are supported on the unit disk with constant density. For some time now it has been appreciated that GinUE has a fundamental place within random matrix theory, both for its applications and for the richness of its theory. Here we review the progress on a number of themes relating to the study of GinUE. These are eigenvalue probability density functions and correlation functions, fluctuation formulas, sum rules and asymptotic behaviours of correlation functions, and normal matrix models. We discuss too applications in quantum many body physics and quantum chaos, and give an account of some statistical properties of the eigenvectors.
\end{abstract}


\maketitle

\tableofcontents

\section{Introduction}\label{S1}
In a fundamental paper on random matrix theory from the early 1960's, Dyson \cite{Dy62c} isolated three
ensembles of Hermitian matrices, and three ensembles of unitary matrices. This was done by seeking the 
minimal requirement of a quantum Hamiltonian $H$, respectively evolution operator $U$, to exhibit a time reversal
symmetry or not, and then imposing a probability measure. First, for a time reversal operator $T$ --- defined in general by the requirement that it be anti-unitary  --- it was first shown that there are two possibilities, either
$T^2=\mathbb I$ or $T^2= - \mathbb I$, with the latter requiring that the Hilbert space be even dimensional.
For $T$ to commute with $H$ or $U$ it was then shown that in
the case $T^2=\mathbb I$ both $H$ and $U$ should be invariant under the transpose operation. In the other possible
case, that $T^2= - \mathbb I$, an invariance under the
so-called quaternion dual $M \mapsto Z_{2N} M^T Z_{2N}^{-1}$ was deduced. Here $Z_{2N}$ is the $2N \times
2N$ anti-symmetric tridiagonal matrix with entries all $-1$ in the leading upper triangular diagonal, and all $1$ in
the leading lower triangular diagonal, and moreover in this case it was shown that $T$ has the realisation $T = Z_{2N} K$ where $K$
corresponds to complex conjugation and $2N$ is the dimension of the Hilbert space.
 In the case of a quantum Hamiltonian $H$, it was then shown that $[H,T]=0$
with $T^2=\mathbb I$ implies a basis can be chosen so that the elements are real, while with $T^2=-\mathbb I$ it implies that
the elements can be chosen to have a $2 \times 2$ block structure
\begin{equation}\label{1.1}
\begin{bmatrix} z & w \\ - \bar{w} & \bar{z} \end{bmatrix}.
\end{equation}
The $2 \times 2$ matrix (\ref{1.1}) can be identified with a member of the (real) quaternion number field. Hence for quantum
Hamiltonians, Dyson was lead to the requirement that matrices in his sought ensemble theory should have
real entries, or have a $2 \times 2$ block structure corresponding to the quaternion number field in the presence of time reversal
symmetry, or to be complex without a time reversal symmetry. This requirement is in addition to the matrices being Hermitian and
thus having real eigenvalues and a matrix of eigenvectors which can be chosen to be unitary.

The work  \cite{Dy62c} specifies the eigenvalue probability density function (PDF) for an ensemble of
quantum Hamiltonians $H$ modelled as random matrices, and chosen from a Gaussian distribution on
the elements proportional to 
\begin{equation}\label{1.1a}
\exp(-\beta \, {\rm Tr} \, H^2/2). 
\end{equation}
The scaling factor $\beta$ is chosen for
convenience, and takes on the value of the number of independent parts of the corresponding number
field --- thus $\beta = 1$ for real entries (time reversal symmetry with $T^2 = \mathbb I$), $\beta = 2$ for
complex entries (no time reversal symmetry), and $\beta = 4$ for quaternion entries 
(time reversal symmetry with $T^2 = - \mathbb I$). With this specification, $\beta$ is referred to as
the Dyson index. A detail is that Hermitian matrices commuting with the quaternion dual must have
doubly degenerate eigenvalues (Kramer's degeneracy), with the convention in (\ref{1.1a}) that the
trace operation relates to the independent eigenvalues only. The result of Dyson, known earlier
in the case $\beta = 1$ by Wigner (see the Introduction section of the book edited by Porter 
\cite{Po65} for references and moreover reprints of the original works) is that the 
eigenvalue PDF is given by  \cite[Eq.~(146)]{Dy62c} 
\begin{equation}\label{1.1b}
\prod_{l=1}^N e^{- \beta \lambda_l^2/2} \prod_{1 \le j < k \le N} | \lambda_k - \lambda_j |^\beta,
\end{equation}
up to proportionality. If instead of the distribution on elements being chosen as (\ref{1.1a}),
a weighting 
\begin{equation}\label{1.1c}
\exp ( - \beta \, {\rm Tr} \, V(H)/2 )
\end{equation}
 for some real valued function $V(\lambda)$ is chosen instead, the modification of (\ref{1.1a})
 is that it now reads
 \begin{equation}\label{1.1d}
\prod_{l=1}^N e^{- \beta V(\lambda_l)/2} \prod_{1 \le j < k \le N} | \lambda_k - \lambda_j |^\beta.
\end{equation}
Here it is being assumed that $ V(\lambda)$ decays sufficiently fast at infinity for (\ref{1.1d})
to be normalisable. A weighting of the form (\ref{1.1c}) is said to specify an invariant ensemble,
since it is invariant under conjugation by a unitary matrix $H \mapsto U H U^\dagger$
(and where too the elements of $U$ are restricted to be real ($\beta = 1$) and
quaternion ($\beta = 4$) so that $H$ remains in the same ensemble). A weighting of the form
(\ref{1.1a}) is said to specify a Gaussian ensemble.

In 1965 Ginibre \cite{Gi65}, motivated by mathematical curiosity \cite[\S 2.2, quoting correspondence with
Ginibre]{AK07}, initiated a study of non-Hermitian Gaussian ensembles with either real,
complex or quaternion entries. Replacing (\ref{1.1a}) is the joint distribution on elements of the corresponding
matrices, now to be denoted $G = [g_{ij}]_{i,j=1}^N$, proportional to
 \begin{equation}\label{1.1e}
\exp(-\beta \, {\rm Tr} \, G^\dagger  G  /2) =  \prod_{i,j=1}^N \exp(- \beta |g_{ij}|^2/2).
\end{equation} 
Note that the second form in this expression shows that the real and imaginary parts (there $\beta$ such parts; e.g.~in the quaternion case $\beta=4$ there is one real and three imaginary parts) are all
independent, identically distributed Gaussians. 
The concern of Ginibre was with the functional form of the eigenvalue PDF, and the implied eigenvalue statistics. The most obvious difference with the
Hermitian case is that the eigenvalues are now in general complex. Also significant is the fact that the
eigenvectors no longer form an orthonormal basis. Notwithstanding these differences, it was found
in the complex case (referred to as the complex Ginibre ensemble, or alternatively as GinUE, where
in the latter the U stands for unitary refers to the bi-unitary invariance of (\ref{1.1e}) being unchanged by the
mapping $G \mapsto UGV$ for $U,V$ unitary matrices) that the eigenvalue PDF is proportional to
 \begin{equation}\label{1.1f}
\prod_{l=1}^N e^{-  | z_l|^2 } \prod_{1 \le j < k \le N} | z_k - z_j |^2,
\end{equation} 
which is in direct correspondence with the Hermitian result (\ref{1.1b}) with $\beta = 2$, obtained
essentially by replacing $\lambda_j$ by $z_j$. 

The eigenvalue
PDF in the case of real and quaternion entries does not follow this correspondence. First, in both these cases the
eigenvalues come in complex conjugate pairs. Appreciating this point, the functional form of the
eigenvalue PDF as found by Ginibre in the quaternion case can be obtained from (\ref{1.1f})
(not  (\ref{1.1b}) with $\beta = 4$) by first replacing $N$ by $2N$, then identifying $z_{j+N}$ as
$\bar{z}_j$ and ignoring terms which involve only $\{ \bar{z}_j \}$ (which are thought of as part
of the image system \cite{Fo16}) to obtain 
 \begin{equation}\label{1.1g}
\prod_{l=1}^N e^{-  2| z_l|^2 } | z_l - \bar{z}_l |^2  \prod_{1 \le j < k \le N} | z_k - z_j |^2 | z_k - \bar{z}_j|^2, \quad {\rm Im} \, z_l > 0.
\end{equation} 
The real case is still more complicated. First, the eigenvalue PDF is not absolutely continuous. Rather, it decomposes
into sectors depending on the number of real eigenvalues. Its precise functional form was not obtained
in Ginibre's original work, with a further 25 years or so elapsing before this was achieved in a publication
by Lehmann and Sommers \cite{LS91}.

It would seem that the first occurrence of a Ginibre ensemble
in applications (specially the real Ginibre ensemble of GinOE, where here the ``O'' stands for orthogonal
and refers to the bi-orthogonal invariance of the matrices) arose in the 1972 work of May \cite{Ma72a} on the stability
of  complex ecological webs. Upon linearising about a fixed point, and the modelling of the fluctuations away
from an attractor by a real Ginibre matrix $G$, May was led to the first order linear differential equation system for the perturbed
populations --- an $n \times 1$ column vector $\mathbf x$ --- specified by
 \begin{equation}\label{1.1h}
 {d \over dt} \mathbf x = (- \mathbb I + \alpha G)  \mathbf x,
 \end{equation} 
 where $\alpha$ is a scalar parameter.
 The stability is then determined by the maximum of the real part of the spectrum of $G$, the precise determination of which has only recently become available in the literature \cite{Be10,CESX22}.

 At the beginning of the 1980's the interpretation of (\ref{1.1f}), written in the Boltzmann
 factor form
\begin{equation}\label{1.1i} 
e^{-\beta U(z_1,\dots, z_N)}, \qquad U = {1 \over 2} \sum_{j=1}^N |z_j|^2 - \sum_{1 \le j < k \le N} \log | z_k - z_j |, \quad  \beta = 2, 
\end{equation} 
as a model of charged particles, repelling pairwise via a logarithmic potential, and attracted to the origin
in the plane via a harmonic potential, gained attention \cite{AJ81,Ca81}. This viewpoint was
already prominent in the works of Dyson in the context of (\ref{1.1b}) (see too the even earlier work of
Wigner \cite{Wi57a} as reprinted in \cite{Po65}), and was noted for (\ref{1.1g}) in Ginibre's original
article \cite{Gi65}.
In contrast to (\ref{1.1b}), in (\ref{1.1i}) the domain is two-dimensional, and the
logarithmic potential is the solution of the corresponding Poisson equation (in one spacial dimension,
the solution of the  Poisson equation is proportional to $|x|$), so (\ref{1.1i}) corresponds to
a type of Coulomb gas. Broad aspects of the latter have been the subject of the recent reviews
\cite{Le22}, \cite{Ch21a}.  Also in the early 1980's, for all positive
values of $\beta/2$ odd,
(\ref{1.1i}) gained attention as the absolute value squared of Laughlin's trial wave function for the
fractional quantum Hall effect \cite{La83}.

Fast forward 40 years, and there are now a multitude of applications which require knowledge of
properties of Ginibre matrices, in particular their eigenvalues and eigenvectors. 
This is due in no small part due to a resurgence of interest in non-Hermitian quantum mechanics
\cite{AGU20}.
Moreover, the progression of time as seen a much deeper understanding of the mathematical
structures associated with the Ginibre matrices, and the theoretical progress has been considerable.
It is the purpose of this article to review a number of these advances, both in the theory and the
applications. Due to space considerations, attention will be focused here on the complex
case of the GinUE --- a subsequent review article is planned relating specifically too GinOE and GinSE.
To the era up to the year 2010, accounts of the progress with emphasis
similar to the present article can be found in \cite[Ch.~15, with proofs]{Fo10}, \cite[some proofs sketched]{KS11}, with the latter overlapping mainly with \S \ref{S2}.
To make the presentation self contained, this material is also part of the present review, albeit with
some reordering and additional context. And when practical from the viewpoint of the space required, proofs of a number of the results are presented.

There are four main themes to the review. These form sections two through to five: eigenvalue PDFs and correlation functions, fluctuation formulas, sum rules and asymptotic behaviours, and normal matrix models. There is also a sixth section entitled further theory and applications. Here the topics considered are the analogy between GinUE and the quantum many body system for free Fermions in the plane subject to a perpendicular magnetic field, the relevance of GinUE statistics to studies in quantum chaos, and statistical properties of the eigenvectors of GinUE matrices.

\section{Eigenvalue PDFs and correlations}\label{S2}
\subsection{Eigenvalue PDF}
In the original paper of Ginibre \cite{Gi65}, the diagonalisation formula $G = V \Lambda V^{-1}$,
where $\Lambda$ is the diagonal matrix of eigenvalues, and $V$ is the matrix of corresponding 
eigenvectors which are unique up to normalisation, was used as the starting point to derive
the eigenvalue PDF (\ref{1.1f}).  The matrix $V$ was then further decomposed $V = U T D$, where
$U$ is unitary, $T$ is upper triangular with all diagonal elements equal to $1$, and
$D$ is a diagonal matrix with real positive elements. As a consequence
\begin{equation}\label{2.1} 
{\rm Tr} \, G^\dagger G = {\rm Tr} \, \bar{\Lambda} B \Lambda B^{-1}, \qquad B = T^\dagger T
\end{equation} 
This is independent of $U$ and $D$, and $B^{-1}$ is a simpler structure than $V^{-1}$, since it
has determinant unity. It is necessary to integrate out the variables of $B$, which was done in $N$
steps with each one consisting of integrating out over the last remaining row and column.

A more versatile (equally applicable to the GinOE, for example) method of derivation of (\ref{1.1f}) has
since been found. It is due to Dyson, and first appeared in published form in
\cite[Appendix 35]{Me91}. Here, instead of using the diagonalisation formula for $G$, the starting
point is the Schur decomposition
\begin{equation}\label{GS}
G = U Z U^\dagger.
\end{equation} 
Here $U$ is a unitary matrix, unique up to the phase of each column, and $Z$ is an
upper triangular matrix with elements on the diagonal equal to the eigenvalues of $G$.

\begin{proposition}\label{P1.1}
For GinUE matrices, specified by the distribution on elements proportional to (\ref{1.1e}) in the case
$\beta = 2$ (complex elements), the eigenvalue PDF is equal to $1/C_{N}$ times (\ref{1.1f}), where upon relaxing the ordering constraint on the eigenvalues implied by (\ref{GS}),
the normalisation constant $C_N$ is specified by
\begin{equation}\label{2.1a} 
C_N =  \pi^N \prod_{j=1}^N j!.
\end{equation} 
\end{proposition}

\begin{proof}
We have from (\ref{GS}) that 
\begin{equation}\label{2.1b} 
{\rm Tr} \, G^\dagger G = \sum_{j=1}^N | z_j |^2 + \sum_{1\le j < k \le N} |Z_{jk}|^2,
\end{equation} 
where here $\{z_j \}$ denotes the diagonal elements of $Z$ (which are the eigenvalues of
$G$), and $\{Z_{jk}\}$ denotes the upper triangular elements. Note the simplification relative
to (\ref{2.1}).

After this brisk start, there is still quite a challenge to compute the Jacobian corresponding to
(\ref{GS}). The strategy of \cite[Appendix 35]{Me91} is explained in more detail in
\cite{HKPV08}, and repeated in \cite[Proof  of Proposition 15.1.1]{Fo10}. Here one begins
by computing the matrix of differentials $U^\dagger dG \, U$, with the Jacobian corresponding
to the (absolute value of) factor which results from the corresponding wedge product.
For the latter task, proceeding in the order of the indices $(j,k)$, with $j$ decreasing from $N$
to $1$, and $k$ increasing from $1$ to $N$, gives  the factor 
\begin{equation}\label{2.1cZ} 
\prod_{j<k} | z_j - z_k |^2. 
\end{equation} 
However
the product of differentials so obtained is not immediately recognisable in the factorised form
\begin{equation}\label{2.1c} 
\wedge_j d z_j^{\rm r} dz_j^{\rm i} \wedge (U^\dagger d U) \wedge_{j<k} d Z_{jk}^{\rm r} d Z_{jk}^{\rm i},
\end{equation} 
where the superscripts indicate the real and imaginary part. Further arguing involving a count of the
number of independent real variables associated with $U$, which implies some apparent differentials
contribute zero to the wedge product,  is required to make the simplification to this
form. With (\ref{2.1c}) established, the integration over $\{Z_{jk}\}$ in (\ref{2.1b}) is immediate.

To deduce the normalisation (\ref{2.1a}) from this calculation requires first that the normalisation
of (\ref{2.1b}) be included throughout the calculation, and second knowledge of the integration
formula (see e.g.~\cite[Eq,~(4.4)]{DF17})
\begin{equation}\label{2.1d} 
\int_U (U^\dagger dU) = {\rm vol} \, \Big ( U(N)/(U(1))^N \Big )  = 2^{N (N - 1)/2} \prod_{l=1}^{N - 1}
{\pi^l \over \Gamma(l+1)}.
\end{equation}
Finally, an extra factor of $N!$ is required in $C_N$ to account for relaxing an ordering of the eigenvalues.
\end{proof}

\subsection{Correlation functions}
With the joint eigenvalue PDF denoted $p_N(z_1,\dots,z_N)$, the $k$-point correlation function
$\rho_{(k),N}(z_1,\dots,z_k)$ is specified by
\begin{equation}\label{2.1e} 
\rho_{(k),N}(z_1,\dots,z_k) = N (N-1) \cdots (N - k + 1) \int_{\mathbb C} d^2 z_{k+1} \cdots  \int_{\mathbb C} d^2 z_{N} \,
p_N(z_1,\dots,z_N),
\end{equation}
where, with $z:= x + i y$, $d^2 z := dx dy$.
In the simplest case $k=1$ this corresponds to the eigenvalue density.
Ginibre \cite{Gi65} showed that
\begin{equation}\label{2.1f} 
\rho_{(k),N}(z_1,\dots,z_k) =  \det \Big [ K_N(z_j, z_l) \Big ]_{j,l=1}^k,
\end{equation}
for a particular function $K_N(w,z)$, referred to as the correlation kernel. The
structure (\ref{2.1f}) makes the eigenvalues of GinUE an example of a determinantal
point process \cite{Bo11}.

\begin{proposition}\label{P2.2}
The kernel function in (\ref{2.1f}) is specified by
\begin{equation}\label{2.1g} 
K_N(w,z) = {1 \over \pi}  e^{- ( |w|^2 + |z|^2)/2}  \sum_{j=1}^N  {  (w \bar{z})^{j-1} \over (j-1)! } =
{1 \over \pi}  e^{- ( |w|^2 + |z|^2)/2} e^{w \bar{z}} {\Gamma(N; w \bar{z}) \over \Gamma(N)},
\end{equation}
where $\Gamma(j;x) = \int_x^\infty t^{j-1} e^{-t} \, dt $ denotes the (upper) incomplete gamma function.
\end{proposition}

\begin{proof}
To deduce the determinantal structure (\ref{2.1f}) with $K_N$ specified by the first equality in
(\ref{2.1g})
(the second equality follows from the first by the identity
$
{\Gamma(N;x) \over \Gamma(N)}  = e^{-x} \sum_{j=1}^N {x^{j-1} \over (j-1)!}
$), the first step is to rewrite the product in 
(\ref{1.1e}) according to
\begin{equation}\label{2.1hm}
\prod_{1 \le j < k \le N} | z_k - z_j |^2 = \prod_{1 \le j < k \le N}  (z_k - z_j ) ( \overline{z}_k -  \overline{z}_j),
\end{equation}
then to rewrite each of the product of differences on the RHS as a Vandermonde determinant,
\begin{equation}\label{2.1h} 
 \prod_{1 \le j < k \le N}  (z_k - z_j ) = \det  [ z_j^{k-1} ]_{j,k=1}^N;
\end{equation} 
see e.g.~\cite[Exercises 1.9 Q.1]{Fo10} for a derivation. Multiplying (\ref{2.1h}) by its conjugate
and taking the transpose of the matrix on the RHS (which leaves the determinant unchanged)
shows
\begin{equation}\label{2.1i} 
{1 \over C_N} \prod_{l=1}^N e^{- | z_l |^2} \prod_{1 \le j < k \le N} | z_k - z_j |^2 =  \det \Big [ K_N(z_j, z_k) \Big ]_{j,k=1}^N.
\end{equation} 

The significance of the form (\ref{2.1i}) for purposes of computing the
 integrations as required by (\ref{2.1e}) are the reproducing and normalisation properties of $K_N$,
 \begin{equation}\label{2.1j} 
 \int_{\mathbb C} K_N(w_1,z)  K_N(z, w_2) \, d^2 z = K_N(w_1, w_2), \qquad
 \int_{\mathbb C} K_N(z,z)  \, d^2 z =  N. 
 \end{equation}
 Using these properties, a cofactor expansion along the bottom row can be used to show \cite{Dy70}
  \begin{equation}\label{2.1k} 
  \int_{\mathbb C} \det [ K_N(z_j, z_k) ]_{j,k=1}^m \, d^2 z_m = (-(m-1)+N)  \det [ K_N(z_j, z_k) ]_{j,k=1}^{m-1};
 \end{equation} 
 see also \cite[Proof of Proposition 5.1.2]{Fo10}. Applying this inductively gives (\ref{2.1f}).
\end{proof}

According to (\ref{2.1f}) with $k=1$ and (\ref{2.1g}), the eigenvalue density is given by the rotationally
invariant functional form
 \begin{equation}\label{2.1l} 
\rho_{(1),N}(z) =  {1 \over \pi}  {\Gamma(N; |z|^2 ) \over \Gamma(N)}.
 \end{equation}
Ginibre \cite{Gi65} identified a sharp transition for $|z| \approx \sqrt{N}$ from the constant value
${1 \over \pi}$ for $|z|$ less than this critical value, to a value approaching zero for   $|z|$ greater
than this critical value. An equivalent statement is the limit law
 \begin{equation}\label{2.2a}
 \lim_{N \to \infty} \rho_{(1),N}(\sqrt{N}z) = \begin{cases} \displaystyle {1 \over \pi}, & |z| <1, \\
 0, & |z|>1. \end{cases}
 \end{equation} 
 This was the first example of what now is termed the circular law, which specifies the 
(global) scaled limiting eigenvalue density for a wide class of non-Hermitian random matrices,
with identically and independently distributed elements to be constant inside a particular circle
in the complex plane, and zero outside \cite{BC12}.

\begin{remark} $ $ \\
1.~In the sense of probability theory, (\ref{2.2a}) is a statement in the mean, due to the ensemble average. The strong version of the circular law establishes that for large $N$ the eigenvalues of a single GinUE matrix obey the circular law  almost surely \cite{BC12}. \\
2.~Define the numerical range of an $N \times N$ matrix $X$ by $W(X) = \{ \overline{X \mathbf u} \cdot \mathbf u \: | \: ||\mathbf u || = 1 \}$. 
By the variational characterisation of eigenvalues, for $X$ Hermitian $W(X)$ must be contained in the interval of the real line $[\lambda_{\rm min}, \lambda_{\rm max}]$, and in fact is equal to this interval. It is proved in \cite{LCGV14} that for $X$ a global scaled Ginibre matrix, and thus with eigenvalue density obeying the circular law (\ref{2.2a}), $W(X)$ converges to the centred disk in the complex plane of radius $\sqrt{2}$.
\end{remark}

The correlation kernel (\ref{2.1g}), and thus the correlation functions (\ref{2.1f}), admit distinct scaling
limits depending on the centring of the variables being in the bulk region, where the density is constant,
or the edge region, where the density begins to decrease to zero. The first was specified in Ginibre's
original paper \cite{Gi65}, whereas the latter was not made explicit until some time later \cite{FH98}.

\begin{proposition}\label{P2.3}
Let $K_N(w,z)$ be specified by (\ref{2.1g}). We have
 \begin{align}
K_\infty^{\rm b}(w,z) & := \lim_{N \to \infty} K_N(w,z) = {1 \over \pi}  e^{- ( |w|^2 + |z|^2)/2} e^{w \bar{z}}, \label{2.2b}\\
K_\infty^{\rm e}(z_1, z_2) & := \lim_{N \to \infty} K_N(-i \sqrt{N} + z_1, - i \sqrt{N} + z_2 )   = 
e^{ - (|z_1|^2 +|z_2|^2)/2} e^{z_1 \bar{z}_2} h\Big ( {1 \over 2}(-i z_1 + i \bar{z}_2) \Big ),  \label{2.2c}
\end{align}
where $
h(z) = {1 \over 2 \pi} \Big ( 1 + {\rm erf}(\sqrt{2} z) \Big )$.
\end{proposition}

\begin{proof}
The limit (\ref{2.2b}) is immediate from (\ref{2.1g}), and the fact that for fixed $w \bar{z}$,
$$
\lim_{N \to \infty}  {\Gamma(N; w \bar{z}) \over \Gamma(N)} = 1.
$$
The derivation of (\ref{2.2c}) relies on (the first term of) the asymptotic expansion \cite{NO19}
\begin{equation}\label{2.2e}
  {\Gamma(N; N + \tau \sqrt{N} ) \over \Gamma(N)} = {1 \over 2} (1 - {\rm erf} (\tau/ \sqrt{2})) +
  {1 \over 3 \sqrt{2 \pi N}} e^{-\tau^2/2}(\tau^2 - 1) + {\rm O} \Big ( {1 \over N} \Big ).
\end{equation}
\end{proof}  

With the leading support of the eigenvalues the disk $|z| < \sqrt{N}$, it is natural to consider as a point in the bulk any $z_0 = s \sqrt{N} + i t \sqrt{N})$ for some $|s|, |t| < 1$. Making use of the fact that $\Gamma(N;u)/\Gamma(N) \to 1$ for $|u|/\sqrt{N} \to c$, $c< 1$ as $N \to \infty$, as is consistent with (\ref{2.2e}), it follows that
$$
\lim_{N \to \infty} K_N(z_0 + w, z_0 + z) = K_{\infty}^{\rm b}(w,z)
$$
independent of $z_0$; see also \cite[Appendix C]{BS09}. Similarly, for any $|\nu|=1$,
\begin{equation}\label{2.20a}
\lim_{N \to \infty} K_N(\nu (\sqrt{N}+ z_1), \nu (\sqrt{N}+ z_2) )   = 
e^{ - (|z_1|^2 +|z_2|^2)/2} e^{z_1 \bar{z}_2} h\Big ( {1 \over 2}(- z_1 - \bar{z}_2) \Big )
\end{equation}
as calculated in \cite[Appendix C with $s_k = \nu z_k$]{BS09}.

Generally (\ref{2.1f}) gives for the appropriately scaled two-point correlation
$$
\rho_{(2),\infty}(z_1,z_2)
 = \rho_{(1),\infty}(z_1)
\rho_{(1),\infty}(z_2)-K_\infty(z_1,z_2)K_\infty(z_2,z_1).
$$
For large separation of $z_1$ and $z_2$ the leading order of the RHS is given by the first term which is the product of the densities. The second term $-K_\infty(z_1,z_2)K_\infty(z_2,z_1)$ must decay sufficiently rapidly for it to be square integrable, since the limiting form of the first integration formula in (\ref{2.1j}) remains valid,
\begin{equation}\label{2.2r} 
 \int_{\mathbb C} K_\infty(w_1,z)  K_\infty(z, w_2) \, d^2 z = K_\infty(w_1, w_2).
 \end{equation}
 This can be verified from the results of Proposition \ref{P2.3} by the evaluation of appropriate Gaussian integrals. To separate off the  product of densities, one defines the truncated (or connected) two-point correlation
\begin{equation}\label{2.2s}
\rho_{(2),\infty}^T(z_1,z_2):=  \rho_{(2),\infty}(z_1,z_2)-
\rho_{(1),\infty}(z_1)
\rho_{(1),\infty}(z_2)= -K_\infty(z_1,z_2)K_\infty(z_2,z_1).
\end{equation}

In particular, with bulk scaling, we read off from this and (\ref{2.2b}) that
\begin{equation}\label{2.2t}
\rho_{(2),\infty}^{{\rm b},T}(z_1,z_2)= -
{1 \over \pi^2}e^{-|z_1 - z_2|^2},
\end{equation}
which thus exhibits a Gaussian decay. With edge scaling, (\ref{2.2s}) and (\ref{2.2c}) give
\begin{equation}\label{2.2t1}
\rho_{(2),\infty}^{{\rm e},T}(z_1,z_2)= -
e^{ - (x_1 - x_2)^2- (y_1 - y_2)^2 } \bigg | h\Big ( {1 \over 2}(y_1 + y_2 - i (x_1 - x_2)) \Big ) \bigg |^2.
\end{equation}
Use of the asymptotic expansion of the error function \cite[Eq.~(7.12.1)]{DLMF} gives to leading order
\begin{equation}\label{2.2t2}
\rho_{(2),\infty}^{{\rm e},T}(z_1,z_2)
\mathop{\sim}\limits_{|z_1-z_2| \to \infty}
- {1 \over 2 \pi^3}
{e^{ - 2 y_1^2 - 2 y_2^2 }
\over 
(y_1 + y_2)^2 + (x_1-x_2)^2}. 
\end{equation}
While this decays in all directions, parallel to the boundary of the leading order density (i.e.~in the $x$-direction) we see that the decay is algebraic, as an inverse square.

\begin{remark}\label{R2.4} $ $ \\
1.~It follows from (\ref{2.2e}) that the edge scaling of the eigenvalue density, in the coordinates of
(\ref{2.2c}) with $z=x+iy$, has the large $N$ expansion
\begin{equation}\label{2.2f}
\rho_{(1),N}^{\rm e}(y) = {1 \over 2 \pi} \Big ( 1 + {\rm erf}(\sqrt{2} y) \Big ) +
 {1 \over 3 \pi \sqrt{2 \pi N}} e^{-2 y^2}( y^2 - 1) + {\rm O} \Big ( {1 \over N} \Big ),
\end{equation}
and thus in particular
\begin{equation}\label{2.2f+}
\rho_{(1),\infty}^{\rm e}(y) = {1 \over 2 \pi} \Big ( 1 + {\rm erf}(\sqrt{2} y) \Big ),
\end{equation}
where this latter expression is consistent with (\ref{2.2c}) upon setting $z_1 = z_2 = z$.
There is interest in the functional form of the $1/\sqrt{N}$ correction term in our discussion of \S \ref{S4.2} below; see too \cite{LR16,ACC22}.
Integrating (\ref{2.2f+}) over $y \in (-\infty, 0]$ and
multiplying by $2 \pi \sqrt{N}$ (the length of the bounding circle of the leading order support)
shows that to leading order the expected number of eigenvalues with modulus greater that $\sqrt{N}$
is $\sqrt{N}/(2 \pi)$ \cite{Gi65}.  Note too that setting $y=0$ in (\ref{2.2f+}) gives $\rho_{(1),\infty}^{\rm e}(0)= {1 \over 2 \pi}$, or equivalently $\lim_{N \to \infty} \rho_{(1),N}(\sqrt{N}) = {1 \over 2 \pi}$, which is exactly ${1 \over 2}$ of the limiting value inside the unit circle as given by (\ref{2.2a}). \\
2.~The two-dimensional classical Coulomb system interpretation (\ref{1.1i}) of the eigenvalue PDF
(\ref{1.1f}) allows for (\ref{2.2a}) to be anticipated. For this, one scales $z_j \mapsto \sqrt{N} z_j$
and introduces a mean field energy functional
\begin{multline}\label{PTm}
{N \over 2} \sum_{j=1}^N | z_j|^2 - \sum_{1 \le j < k \le N} \log | z_k - z_j | \\
\sim {N \over 2} \Big ( \int_\Omega \rho_{(1)}(z) |z|^2 \, d^2z -
\int_\Omega d^2w \, \rho_{(1)}(w) \int_\Omega d^2z \, \rho_{(1)}(z) \log | z - w| \Big ),
\end{multline}
with the hypothesis that $ \rho_{(1)}(z) $ is chosen so that this functional is minimised
and furthermore integrates over $\Omega$ to unity.
Characterising the minimisation property by the vanishing of the functional upon
variation with respect to $ \rho_{(1)}(z) $ gives
\begin{equation}\label{PT}
| z|^2 - 2 \int_\Omega \rho_{(1)}(w) \log | z - w|  \, d^2w  = C, 
\end{equation}
where $C$ is a constant, valid for $z \in \Omega$. 
Applying the Laplacian operation $\nabla_z^2$ to this, using the standard fact
$$
- \nabla_z^2 \log | z - w| = - 2 \pi \delta (z - w),
$$
it follows
$$
 \rho_{(1)}(w)  = {1 \over \pi} \chi_{|w| < 1}.
 $$
 Here the restriction to $|w| <  1$ is implied by the rotational invariance, together with the minimisation and normalisation requirements. The notation
 $\chi_A$ denotes the indicator function of the condition $A$,
 taking on the value $1$ when $A$ is true, and $0$ otherwise.
 The support $\Omega$ of 
 $ \rho_{(1)}(w)$ --- which here is the unit disk --- in such a Coulomb gas picture is typically referred to as the
 droplet (see e.g.~\cite{ABWZ02,HM13}). Moreover, this potential theoretic
 reasoning can rigorously be justified; see e.g.~\cite[\S 3.1]{Ch21a} and
 references therein, as well as the discussion and references in the paragraph including (\ref{10.1b2}) of \S \ref{S4.2b} below. \\
 3.~Let $X_0$ have finite rank, and $X$ be a GinUE matrix, scaled so that the leading support of the eigenvalues is the unit disk. Assume too that the eigenvalues of $X_0$ are inside of the unit disk and near the boundary. In the limit $N \to \infty$ it has recently been shown that the edge correlation functions centred on the eigenvalues of $X_0$ form a determinantal point process with kernel involving generalisations of the error function \cite{LZ22}.
\end{remark}  

\subsection{Elliptic GinUE} \label{S2.3}
In 1991 Lehmann and Sommers introduced a one parameter
generalisation of the non-Hermitian complex Gaussian matrices specifying GinUE.
In this generalisation, varying the parameter allows for the Hermitian ensemble of
Gaussian matrices known as the GUE (Gaussian unitary ensemble) to be obtained.
An Hermitian matrix $H$ from the GUE can be constructed from a scaled non-Hermitian
GinUE matrix $\tilde{G} = (1/\sqrt{2}) G$ according to
\begin{equation}\label{Jam}
H = {1 \over 2} \Big ( \tilde{G} + \tilde{G}^\dagger \Big ).
\end{equation}

\begin{proposition}
Let the parameters $0 < \tau, v < 1$  be related by $\tau = (1 - v^2)/(1 + v^2)$.
For $H_1, H_2$ elements of the GUE, define
\begin{equation}\label{Ja}
J =  \sqrt{1 +  \tau} ( H_1 + i v H_2  ).
\end{equation}
The eigenvalue PDF of the ensemble of matrices $\{J\}$ is given by
\begin{equation}\label{2.3}
\exp \Big ( - {1 \over 1 - \tau^2} 
\sum_{j=1}^N \Big ( |z_j|^2 - {\tau \over 2}
(z_j^2 + \bar{z}_j^{2} )\Big ) \Big ) \prod_{1 \le j < k \le N}
| z_k - z_j|^2.
\end{equation}
\end{proposition}

\begin{proof}
Following \cite{FKS97} and \cite[Exercises 15.1 Q.1]{Fo10}, the starting point is to note that
the joint element PDF of $ \sqrt{1 +  \tau}  H_1$ and $ \sqrt{1 +  \tau}  H_2$ is proportional to
$\exp ( - {1 \over 1 + \tau} {\rm Tr} (H_1^2 + H_2^2) )$. The definition of $J$ gives
$$
{\rm Tr} \, H_1^2 = {1 \over 2} 
\Big ( {\rm Tr}( J  J^\dagger)
+ {\rm Re} \, {\rm Tr} (J^2) \Big ), \qquad
 {\rm Tr} \, H_2^2 = {1 \over 2v^2} \Big ( {\rm Tr}( J 
J^\dagger) - {\rm Re} \, {\rm Tr} ( J^2) \Big ).
$$
As a consequence, it follows that the joint element PDF of $J$ is proportional to
\begin{equation}\label{2.25a}
\exp \Big ( - {1 \over 1 - \tau^2} {\rm Tr} (
 J   J^\dagger - \tau {\rm Re} \, J^2 ) \Big ).
\end{equation}
With this knowledge, the strategy of the proof of Proposition \ref{P1.1}
leads to (\ref{2.3}).
\end{proof}

The correlations for (\ref{2.3}) have, for an appropriate correlation kernel
$K_N(w,z;\tau)$, the
determinantal form (\ref{2.1f}). This involves the
scaled monic Hermite polynomials 
\begin{equation}\label{cH}
C_n(z) := \Big ( {\tau \over 2} \Big )^{n/2} H_n \Big ( {z \over \sqrt{2 \tau} } \Big ), \qquad z \in \mathbb C.
\end{equation}

\begin{proposition}\label{P2.6}
In the notation specified above, we have
\begin{multline}\label{cH1}
K_N(w,z;\tau) =  {1 \over \pi} {1 \over \sqrt{1 - \tau^2} }\\
\times
\exp \Big ( - {1 \over 2(1 - \tau^2)} \Big (
|w|^2+ | z|^2 -  \tau  ( {\rm Re} \, w^2  +  {\rm Re} \, z^2 ) \Big )\Big )
\sum_{l=0}^{N-1} {C_l(w) C_l(\bar{z}) \over l!}.
\end{multline}
\end{proposition}

\begin{proof}
Following the same procedure as used to begin the proof of Proposition \ref{P2.2},
we modify (\ref{2.1h}) so that it reads
\begin{equation}\label{2.1hm+} 
 \prod_{1 \le j < k \le N}  (z_k - z_j ) = \det  [ p_{k-1}(z_j) ]_{j,k=1}^N,
\end{equation} 
where $\{ p_l \}$ is a set of monic polynomials, each $p_l$ of degree $l$.
Choosing these polynomials according to (\ref{cH}), with this choice being
motivated by
the orthogonality \cite{vM90, DGIL94}
$$
\int_{-\infty}^\infty dx \int_{-\infty}^\infty dy \, e^{ - x^2/(1 + \tau) -  y^2/(1 - \tau)} 
 C_m(z) C_n(\bar{z}) = \pi m! \sqrt{1 - \tau^2}
\delta_{m,n},
$$
the remaining working of the proof of Proposition \ref{P2.2} establishes the result.
\end{proof}

Applying the Coulomb gas argument of Remark \ref{R2.4}.2, with the scaling
$z_l \mapsto \sqrt{N} z_l$, we conclude that within some  domain $\Omega$ the density is constant, taking the value $\rho_{(1)}(w) =
1/(\pi (1 - \tau^2))$. 
This domain, or equivalently
droplet, can be determined to be an ellipse with semi-axes $A = 1 + \tau$,
$B = 1 - \tau$ and area equal to $\pi(1 - \tau^2)$. The shape can be verified directly, by showing that with $\Omega$ so specified, and $|z|^2$ in (\ref{PT}) replaced by
\begin{equation}\label{EF}
 {1 \over 1 - \tau^2} ( | z|^2 - \tau {\rm Re} \, z^2 ), 
 \end{equation}
 the required minimisation equation is indeed satisfied \cite{CPR87}, \cite[Exercises 15.2 q.4]{Fo10}; 
see also  \cite{DGIL94, FJ96,LR16,By23}. 
This droplet shape
explains the terminology elliptic GinUE in relation to (\ref{2.3}). For complex non-Hermitian
matrices (\ref{Ja}), now with the Hermitian random matrices $H_1,H_2$ constructed
from (\ref{Jam}) with general zero mean, unit standard deviation identically distributed entries
that are not required to be Gaussian, a constant density in an ellipse was first deduced by Girko
\cite{Gi86} upon additional assumptions, and without qualification by Nguyen and O'Rourke \cite{NO15}.

As for the GinUE, as the boundary of the leading support of the elliptic GinUE is approached, there
is a transition from a constant density to a density which decays to zero. Beginning with
(\ref{cH1}) in the case $w=z$, the analysis of the eigenvalue density in this edge regime
is more complicated than for the deduction of (\ref{2.2e}). It was carried out by
Lee and Riser \cite{LR16}; see also \cite{ADM22,Mo22}.

\begin{proposition} \label{P2.7}
Consider the elliptic GinUE, with the eigenvalues scaled $z_j \mapsto \sqrt{N} z_j$ so that for
$N$ large the leading order support is the ellipse $\Omega$. Let $z_0$ be a point on the boundary
of $ \Omega$. 
Denote the unit vector corresponding to the outer normal at this point by $\mathbf n$,
and the corresponding curvature by $\kappa$. We have
\begin{equation}\label{2.25}
\rho_{(1),N}\Big ( z_0 + {(\alpha + i \beta ) \mathbf n \over \sqrt{N} } \Big ) =
{1 \over 2 \pi} \Big ( 1 - {\rm erf}(\sqrt{2} \alpha) \Big ) +
 {\kappa \over  \pi \sqrt{2 \pi N}} e^{-2 \alpha^2} \Big ( {\alpha^2 - 1 \over 3} - \beta^2 \Big ) + {\rm O} \Big ( {1 \over N} \Big ).
\end{equation}
\end{proposition}

\begin{remark}\label{R2.8} $ $ \\
1.~The leading term in (\ref{2.25}) is identical to that in (\ref{2.2f}) with the identification $y = - \alpha$. The universality of this
functional form has been established for a wide class of normal matrix models
(see \S~\ref{S4.2b} in relation to this class), and
moreover extended to the edge scaled correlation kernel (\ref{2.2c})
\cite{AKM20,HW21}. This has recently been proved too for non-Hermitian random matrices
constructed according to (\ref{Jam}) with the elements of $G$ identically distributed mean
zero, finite variance random variables \cite{CES21}. \\
2.~Consider the elliptic shaped domain
$$
\{ z \in \mathbb C: \, 1 - {1 \over 1 - \tau^2}  ( | z|^2 - \tau {\rm Re} \, z^2 ) > 0 \}.
$$
Let the function on the LHS of the inequality be denoted by $1 - f(z)$. We see that $f(z)$ coincides with (\ref{EF}). It has been
shown in \cite{ANPV21,NAKP20} that the polynomials $\{p_n(z)\}$ satisfying the orthogonality
$$
\int_{\Omega} p_m(z) p_n(\overline{z}) (1 - f(z))^\alpha \, d^2z \propto \delta_{m,n},
$$
are simply related to the Gegenbauer polynomials $\{ C_n^{(1+\alpha)}(z) \}$. After 
scaling, the result (\ref{cH}) can be reclaimed by taking the limit $\alpha \to \infty$. \\
3.~We see from (\ref{Ja}) that as $\tau \to 1^-$, $J$ is proportional to a GUE matrix, and in
particular its eigenvalues are then all real. It was found by Fyodorov, 
Khoruzhenko and Sommers \cite{FKS98}
 that setting $\tau = 1 -\pi^2 \alpha^2/2N$ and scaling the
eigenvalues $z_j \mapsto \pi z_j/N = \pi (x_j + i y_j)/N$ that
\begin{multline}\label{Wk}
\lim_{N \to \infty} {\pi^2 \over N^2} K_N(\pi w/N,\pi z/N;\tau) \Big |_{\tau = 1 - \pi^2 \alpha^2/2N}
 \\  =
\sqrt{2 \over \pi \alpha^2} \exp \Big ( - {y_1^2 + y_2^2 \over \alpha^2} \Big )
{1 \over 2 \pi} \int_{-\pi}^\pi \exp \Big ( - {\alpha^2 u^2 \over 2} + i u (w - \bar{z}) \Big ) \, du.
\end{multline}
(For a direct comparison between (\ref{Wk}) and what has been reported in the literature, see
\cite[$K_{\rm weak}(z_1,z_2)$ in Theorem 3]{ACV18} and \cite[\S 1]{AB21}.) This is referred to as the weakly non-Hermitian limit. Scaling of the correlation kernel at the edge in this limit has been considered in \cite{GNV02,Be10}.

\end{remark}

\subsection{Induced GinUE}\label{S.GI}
The fact that a complex Gaussian matrix $G$ is bi-invariant with respect to multiplication
by unitary matrices allows for the distribution of the singular values to be related to the
eigenvalue distribution \cite{FBKSZ12,KK16}.

\begin{proposition}
Let $G$ be bi-unitary invariant, and let $U$ be a Haar distributed unitary matrix.
We have that $G$ and $(G^\dagger G)^{1/2} U$ have the same joint element distribution,
and so in particular have the same distribution of eigenvalues.
\end{proposition}

\begin{proof}
By the 
singular value decomposition, $G = U_1 \Sigma U_2$
for some unitary matrices $U_1, U_2$ and
where $\Sigma$ is the diagonal matrix of the singular values. 
We then have $(G^\dagger G)^{1/2} U = U_2^\dagger \Sigma U_2 U$. 
By the assumed  bi-unitary invariance of $G$, this matrix and $G$ have the same
distribution.
\end{proof}

The matrix 
\begin{equation}\label{dA}
A:= (G^\dagger G)^{1/2} U
\end{equation}
 is well defined for $G$ rectangular, and
moreover the property of $G$ being bi-unitary invariant can be generalised to this
setting. Of interest is the relation between the joint element distributions of $G$ and
$A$ \cite{FBKSZ12}.

\begin{proposition}\label{P2.8a}
Let $G$  be a bi-unitary invariant rectangular  $n \times N$ random matrix with
joint element distribution of the functional form $g(G^\dagger G)$. The joint element distribution of
the matrix $A$ (\ref{dA}) with $U$ a Haar distributed unitary matrix is proportional to
\begin{equation}\label{DA}
(\det A^\dagger A)^{ (n - N)} g(A^\dagger A).
\end{equation}
As a consequence, for $G$ a rectangular complex Ginibre matrix, the eigenvalue PDF of $A$  is 
\begin{equation}\label{1.1cI}
{1 \over C_{n,N}} \prod_{l=1}^N | z_l|^{2(n - N)} e^{-  | z_l |^2 } \prod_{1 \le j < k \le N} | z_k - z_j |^2,
\end{equation}
with normalisation $C_{n,N} = N! \pi^N \prod_{j=1}^N (n-N+j-1)!$.

\end{proposition}

\begin{proof}
Write $B = G^\dagger G$. With the joint element distribution of $G$
of the functional form $g(G^\dagger G)$, it is a standard result
(see e.g.~\cite[Eq.~(3.23)]{Fo10}) that the joint element distribution of
$B$ is proportional to $\det B^{n - N} g(B)$. But $B$ and $A^\dagger A$ have the
same joint element distribution from the bi-unitary invariance of $G$,
and the result just quoted with $n=N$ tells us that the Jacobian for the joint
element distribution of 
$A$ and $B$ is
a constant, which implies (\ref{DA}). In the particular case that $G$ is a rectangular 
complex Ginibre matrix, the function $g$ in (\ref{DA}) is the exponential
$g(X) = e^{-{\rm Tr} \, X}$. Furthermore, with the eigenvalues of $A$ denoted
$\{z_l\}$, we have
$(\det A^\dagger A)^{ (n - N)}  =
\prod_{l=1}^N  | z_l|^{2(n - N)}$. Taking these points into consideration, we see (\ref{1.1cI})
results by following the proof of Proposition \ref{P1.1}.  Moreover, with $\omega(z) = \omega(|z|^2) = |z|^{2(n - N)} e^{-|z|^2}$ the weight function, that the analogue of (\ref{2.1i}) in that proof have the properties (\ref{2.1j}) requires that the normalisation equal $N! \prod_{j=0}^{N-1} 2 \pi \int_0^\infty r^{2j+1} \omega(r) \, dr$. This implies the stated value of $C_{n,N}$.
\end{proof}

The correlations for (\ref{1.1cI}) are of the determinantal form (\ref{2.1f}). Denoting the
corresponding correlation by $K^{\rm i G}(w,z)$, the derivation of (\ref{2.1g}) shows  \cite{Ake01,FBKSZ12}
\begin{align}\label{2.1gI}
K_N^{\rm i G}(w,z) & = {1 \over \pi}  e^{- ( |w|^2 + |z|^2)/2}  \sum_{j=1}^N  {  (w \bar{z})^{n-N+j-1} \over (n - N+ j-1)! } 
\nonumber \\
& =
 {1 \over \pi}  e^{- ( |w|^2 + |z|^2)/2}  e^{w \bar{z}} \Big ( 
{\Gamma(n;w \bar{z}) \over \Gamma(n)} - 
 {\Gamma(n-N;w \bar{z}) \over \Gamma(n - N)}  \Big ). 
\end{align}

We know that the eigenvalue density, $\rho_{(1),N}^{\rm iG}(z)$ say, results by setting $w=z$ in $K_N^{\rm i G}(w,z)$.
Knowing from (\ref{2.2e}) that for  $N$ large $\Gamma(N;xN)/\Gamma(N)$ exhibits a transition from the value $1$
for $0 < x < 1$ to the value $0$ for $x > 1$, we see from (\ref{2.1gI}) that
\begin{align}\label{2.1gm}
\lim_{N \to \infty} \rho_{(1),N}^{\rm iG}(\sqrt{N} z) \Big |_{n / N = \alpha + 1} = {1 \over \pi} \Big ( \chi_{|z| < \sqrt{\alpha + 1}} -
\chi_{|z| < \sqrt{\alpha} }\Big ).
\end{align}
Here $ ( \chi_{|z| < \sqrt{\alpha + 1}} -
\chi_{|z| < \sqrt{\alpha}} )$ corresponds to an annulus of inner radius $\sqrt{\alpha}$ and outer radius $\sqrt{\alpha+1}$.  
Generally, for random matrices of the form $A=UTV$, where $U,V$ are Haar distributed and the singular values of $T$
converge to a compactly supported probability measure, it is known that the eigenvalue PDF in the complex plane is
either a disk or an annulus. This is referred to as the single ring theorem \cite{FZ97,GKZ17}.
Scaling of (\ref{2.1gI}) in the bulk of the annulus, or at the boundary, is straightforward and leads to the
functional forms exhibited in Proposition \ref{P2.3} for the GinUE.
 Furthermore, in the double scaling regime that the spectrum tends to form a thin annulus of width $O(1/N)$, the scaling of (\ref{2.1gI}) in the bulk gives rise to the weakly non-Hermitian limit \eqref{Wk} \cite{BS21}.
We also stress that the induced Ginibre ensemble was introduced more generally in \cite{Ake01}, where the eigenvalue PDF is a mixture of \eqref{2.3} and \eqref{1.1cI}. 
(The induced GinUE model is then obtained by setting the masses to zero $m_f=0$ and $\tau=0$.) 
In \cite{Ake01}, the finite-$N$ expression of the correlation functions as well as their scaling limits both at strong and weak non-Hermiticity were obtained. 

\subsection{Complex spherical ensemble}\label{S2.5}
For $G_1, G_2$ matrices from the GinUE, matrices of the form $G=G_1^{-1}G_2$ are said to form the complex
spherical ensemble \cite{Kr06}. For with $\alpha, \beta \in \mathbb C$ with
$|\alpha|^2 + |\beta|^2 = 1$, introduce the transformed pair of matrices $(C,D)$,
$$
C := -\bar{\beta} G_2 + \bar{\alpha} G_1, \qquad D = \alpha G_2 + \beta G_1.
$$
Since $G_1,G_2$ have independent standard complex elements, it is easy to check that the pair $(G_1,G_2)$ has
the same distribution as $(C,D)$. As a consequence, the eigenvalues $\{z_j\}$ of $G$ are unchanged by the
fractional linear transformation
$$
z \mapsto {z \alpha - \bar{\beta} \over z \beta - \bar{\alpha}}.
$$
This implies that upon a stereographic projection from the complex plane to the
Riemann sphere, the eigenvalue distribution of $\{ G \}$ is invariant under rotation of the sphere, giving rise to the
name of the ensemble and telling us that on this surface the eigenvalue density is constant.

The explicit eigenvalue PDF for the complex spherical ensemble was calculated by Krishnapur \cite{Kr06}.

\begin{proposition}\label{P2.11}
For the complex spherical ensemble, the eigenvalue PDF is proportional to
\begin{equation}\label{CV}
\prod_{l=1}^N {1 \over (1 + |z_l|^2)^{N+1}} \prod_{1 \le j < k \le N} |z_k - z_j|^2, \qquad z_l \in \mathbb C.
\end{equation}
\end{proposition}

\begin{proof}
The joint element distribution of $(G_1,G_2)$ is proportional to $\exp(-{\rm Tr} \, G_1^\dagger G_1 - {\rm Tr} \, G_2^\dagger G_2)$.
Substituting  $G_2 = G G_1$, (this gives a Jacobian factor $| \det G_1 |^N$) then integrating over $G_1$ gives that the element
distribution of $G$ is proportional to
\begin{equation}\label{2.45a}
\det ( \mathbb I + G^\dagger G)^{-2N}.
\end{equation}
Introducing now the Schur decomposition (\ref{GS}), and changing variables as in the proof of Proposition
\ref{P1.1}, we have from this that the element distribution of the upper triangular
matrix $Z$ therein is proportional to
\begin{equation}\label{CV1}
\det ( \mathbb I + Z^\dagger Z)^{-2N} \prod_{1 \le j < k \le N} | z_k - z_j |^2.
\end{equation}
Here $\{z_j\}$ are the eigenvalues of $G$, and also the diagonal entries of $Z$.

It remains to integrate over the strictly upper triangular entries of $Z$. For this, 
denote the leading $n \times n$ sub-block of $Z$ by $Z_n$ and let the
product of differentials for the strictly upper entries of $Z_n$ be
denoted $(d \tilde{Z}_n)$. For $p \ge n$, define
$$
I_{n,p}(z_1,\dots,z_n) := \int {1 \over \det ( \mathbb 1_n + Z_n  Z_n^\dagger)^{p} } \, (d \tilde{T}_n).
$$
Writing
\begin{equation}\label{Z}
Z_n = \begin{bmatrix} Z_{n-1} & \mathbf u \\
\mathbf 0_{n-1}^T & z_n \end{bmatrix}, \qquad \mathbf u = [Z_{jn}]_{j=1}^{n-1},
\end{equation}
it's possible to integrate out over the elements of $\mathbf u$ to obtain the recurrence \cite{HKPV08,FK09}
$$
I_{n,p}(z_1,\dots,z_n) = {C_{n-1,p} \over (1 + |z_n|^2 )^{p-n+1} }
I_{n-1,p-1}(z_1,\dots,z_{n-1}), \qquad  C_{n-1,p} = \int { (d \mathbf v ) \over (1 + \mathbf v^\dagger \mathbf v )^p }.
$$
Iterating this with $n=2N$, $p=N$ shows
$$
\int  \det ( \mathbb I + Z^\dagger Z)^{-2N} \,  (d \tilde{Z}_N) \propto \prod_{l=1}^N {1 \over (1 + |z_l|^2)^{N+1} },
$$
which when used in (\ref{CV1}) implies (\ref{CV}).

\end{proof}

The stereographic projection from the south pole of a sphere with radius $1/2$, spherical coordinates $(\theta, \phi)$, to a
plane tangent to the north pole is specified by the equation
$$
z = e^{i \phi} \tan {\theta \over 2}, \qquad z = x + i y.
$$
Making this change of variables in (\ref{CV1}) gives the PDF on the sphere proportional to
\begin{equation}\label{CV2}
\prod_{1 \le j < k \le N} | u_k v_j  - u_j v_k |^2, \qquad u = \cos (\theta/2) e^{i \phi/2}, \: v = -i \sin (\theta/2) e^{-i \phi/2}.
\end{equation}
Here $u,v$ are the Cayley-Klein parameters, and moreover $ | u_k v_j  - u_j v_k | = | \mathbf r_k - \mathbf r_j|$ where
$\mathbf r_j,  \mathbf r_k$ are the vector coordinates on the sphere. It is furthermore the case that minus the logarithm of this distance solves the Poisson equation on the sphere, and so (\ref{CV2}) has the interpretation of the Boltzmann factor at coupling $\beta =2$ of the corresponding Coulomb gas; see \cite[\S 15.6]{Fo10} for more details.

After first removing $u_k, u_j$ from the product of differences in (\ref{CV2}), the Vandermonde determinant identity
(\ref{2.1h}) can be used to compute the correlation functions following a strategy analogous to that used in the
proof of Proposition \ref{P2.2}. This calculation was first done in the context of the corresponding two-dimensional
one-component plasma, as obtained by the rewrite of (\ref{CV2}) analogous to (\ref{1.1i}) \cite{Ca81}.

\begin{proposition}
The $n$-point correlations for (\ref{CV2}) are given by 
\begin{equation}\label{11.53}
\rho_{(n,N)}\left((\theta_1,\phi_1),\dots,(\theta_n,\phi_n)\right)=\Big ( {N \over  \pi} \Big )^n
{\rm{det}}\left[{\left( u_j\bar{u}_k+v_j\bar{v}_k \right)}^{N-1}
\right]_{j,k=1,\dots,n}.
\end{equation}
\end{proposition}

\begin{remark} $ $ \\
1.~Changing variables $\theta_j = 2  r_j \sqrt{\pi/ N}$ in (\ref{11.53}) and taking $N \to \infty$, the
bulk scaled correlation kernel (\ref{2.2b}) results. Here there is no edge regime. \\
2.~Denote by $a$ ($X$) an $n \times N$ ($N \times M$), $n \ge N$ ($M \ge N$) standard complex Gaussian
matrix, and set $A = a^\dagger a$, $Y = A^{-1/2} X$. In terms of $Y$ define $Z = U (Y Y^\dagger)^{1/2}$,
which corresponds to the induced ensemble construction of \S~\ref{S.GI}.
It was shown in \cite{FF11} that the element PDF of $Z$ is proportional to 
$$
(\det Z^\dagger Z)^{M - N} {1 \over \det (\mathbb I + Z^\dagger Z)^{n+M}}.
$$
The proof of Proposition \ref{P2.11} now gives that the corresponding eigenvalue PDF is proportional to
\begin{equation} \label{PDF Ispherical}
\prod_{l=1}^N {| z_l |^{2 (M - N)}  \over (1 + |z_l|^2)^{n+M - N +1}} \prod_{1 \le j < k \le N} |z_k - z_j|^2.
\end{equation}
With $M, n$ scaled with $N$, $M/N \to \alpha_1\ge 1 , n/N \to \alpha_2 \ge 1$  the leading order eigenvalue 
support now occurs in an annulus with inner and outer radii 
\begin{equation} \label{r1 r2 Ispherical}
r_1 = \sqrt{(\alpha_1 -1)/\alpha_2}, \qquad r_2 = \sqrt{\alpha_1/ (\alpha_2 - 1)}. 
\end{equation}
The associated correlation kernel $K_N$ can be expressed in terms of the incomplete beta function $I_x (a,b) \propto \int_0^x t^{a-1}(1 - t)^{b-1} \, dt$, normalised to equal unity for
$x=1$ as  
\begin{align}  \label{K Ispherical}
K_N(w,z)&= \frac{1}{\pi} \frac{ |zw|^{ M-N  } }{  ( (1+|z|^2)(1+|w|^2) )^{ (n+M-N+1)/2  }  }
\\
&\quad \times (M+n-N)(1+w\bar{z})^{n+M-N-1} \Big( I_\zeta (M-N,n)-I_\zeta (M,n-N) \Big),  \nonumber 
\end{align}
where $ \zeta=( w\bar{z} )/( 1+w\bar{z} )$, see e.g. \cite{BF22}. 
Scaling of the correlation functions at either of these boundaries
gives the edge kernel (\ref{2.2c}). Setting $w=z$ and with $\alpha_1, \alpha_2$ as in (\ref{r1 r2 Ispherical})
allows for the computation of the global density \cite{FF11}
\begin{equation} \label{sp.g}
\lim_{N \to \infty}
{1 \over N} \rho_{(1),N}(z) =
{ \alpha_1+\alpha_2-1 \over \pi (1 + r^2)^2}
\chi_{r_1 \le r \le r_2}.
\end{equation}
\end{remark} 

\subsection{Sub-block of a unitary matrix}
Let $U$ be an $(n + N) \times (n + N)$ unitary matrix, and let $A$ be the top $N \times N$ sub-block. 
The non-zero eigenvalues of $A$ are then the nonzero eigenvalues of $DUD$, where $D$ is the diagonal
matrix with the first $N$ diagonal entries $1$, and the last $n$ diagonal entries $0$. The eigenvalues of this
matrix are the same as that for $UD$, since in general for square matrices the eigenvalues of $AB$ and $BA$
coincide, and furthermore $D^2 = D$. For $\mathbf u$ a normalised eigenvector of $UD$ with eigenvalue $\lambda$,
computing the length squared of $UD \mathbf u$ gives $\mathbf u^\dagger D \mathbf u = | \lambda|^2$, where use
has been made of the fact $U^\dagger U = \mathbb I$. But with the entries of $\mathbf u$ all nonzero, the action of
$D$ reduces the length, implying $|\lambda| < 1$. With $U$ chosen with Haar measure, 
it was shown by Zyczkowski and Sommers \cite{ZS99}
 that the exact eigenvalue PDF of $A$ can be calculated.

 \begin{proposition}\label{P2.12}
 For $A$ the top $N \times N$ sub-block of an $(n + N) \times (n + N)$ unitary matrix chosen with Haar measure,
  the eigenvalue PDF is proportional to
\begin{equation}\label{CV+}
\prod_{l=1}^N (1 - | z_l |^2)^{n-1} \prod_{1 \le j < k \le N} |z_k - z_j|^2, \qquad |z_l| < 1.
\end{equation}
\end{proposition}

\begin{proof}
Let $C$ be the block of $U$ in the first $N$ columns directly below $A$. Then by the unitarity of $U$,
$A^\dagger A + C^\dagger C = \mathbb I_N$. This implies a joint distribution in  the space of general $N \times N$
and $n \times N$ complex matrices $A,C$ proportional to the matrix delta function
\begin{equation}\label{3.x}
\delta(A^\dagger A + C^\dagger C - \mathbb I_N) \propto \int e^{{\rm Tr} ((i H - \mu \mathbb I_N) ( A^\dagger A + C^\dagger C - \mathbb I_N)} \, (d H), \quad \mu > 0.
\end{equation}
Here in the integral form of the matrix delta function $H$ is Hermitian; see e.g.~\cite[Eq.~(3.27)]{Fo10}. Beginning with
this integral form the integration of the complex matrix $C$ can be carried out according to 
\cite[displayed equation below (3.27)]{Fo10} to give that the element PDF of $A=: A_N$ is proportional to
\begin{equation}\label{3.y}
F_n(A_N) := \int (\det(i H_N - \mathbb I_N))^{-n} e^{{\rm Tr} ((i H_N - \mathbb I_N)( A_N^\dagger A_N - \mathbb I_N))} \, (dH_N).
\end{equation}
This matrix integral is the starting point for the derivation of (\ref{3.x}) given in \cite[Appendix B]{ABKN13}, which we follow below.

In (\ref{3.y}) the integral is unchanged by conjugating $H$ with  a unitary matrix $V$ say. Choosing $V$ 
to be the unitary matrix in the
Schur decomposition (\ref{GS}) allows us to effectively replace $A_N$ by $Z_N$ throughout (\ref{3.y}).
Doing this, and integrating too over $\tilde{Z}_N$ (i.e.~the strictly upper triangular entries of $Z_N$), we denote
the matrix integral by $\tilde{F}_n(z_1,\dots,z_N)$.

Substituting the decomposition (\ref{Z}) for
$Z_N$ we see that the vector $\mathbf u$ occurs as a quadratic form, and can be integrated over. To progress
further, $H_N$ too is decomposed 
 by a bordering procedure of the leading $(N-1) \times (N-1)$ sub-block $H_{N-1}$,
 with the $N-1$ column (row) vector $\mathbf w$
($\mathbf w^\dagger$) and entry $h_N$ in the bottom right corner. 
Key now is a determinant identity for
a block matrix 
$$
\det \begin{bmatrix} A & B \\ C & D \end{bmatrix} = \det (D) \det(A - B D^{-1} C),
$$
familiar
from the theory of the Schur complement (see e.g.~\cite{Ou81}),
applied to the first term in the integrand  of (\ref{3.y}) with $D = h_N$ (a scalar). A shift of the integration
domain according to the additive rank 1 perturbation $H_{N-1} \mapsto H_{N-1} - i \mathbf w \mathbf  w^\dagger$ 
gives a quadratic form in $\mathbf w$, which after integration reduces (\ref{3.y}) to
\begin{multline}\label{3.y1}
\tilde{F}_n(z_1,\dots,z_N) \propto \int (\det(i H_{N-1}-  \mathbb I_{N-1}))^{-n} (i h_N - 1)^{-n} e^{(|z_N|^2 - 1) (i h_N - 1)}  \\
\times e^{{\rm Tr} (Z_{N-1}^\dagger Z_{N-1} (i H_{N-1}  - \mathbb I_{N-1}))} \, (dH_{N-1}) (d Z_{N-1}) d h_N.
\end{multline}
Here, using the residue theorem, the integral over $h_N$ can be performed, yielding the recurrence
$$
\tilde{F}_n(z_1,\dots,z_N) \propto  (1 - |z_N|^2)^{n-1} \chi_{|z_N|^2 < 1} \tilde{F}_n(z_1,\dots,z_{N-1}).
$$
Iterating gives the first factor in (\ref{CV+}), while the product of differences is due to the Jacobian (\ref{2.1cZ}) for the
change of variables to the Schur decomposition as computed in the proof of Proposition \ref{P1.1}.
\end{proof}

%

\begin{remark}
The matrix integral (\ref{3.y})
can be evaluated to deduce that the element distribution
of $A$ is proportional to \cite{FS03,Co05}
$
 \det (\mathbb I - A^\dagger A)^{n - N}.
$
Only for $n \ge N$ is the normalisable. This is in keeping with $A^\dagger A$ having $N-n$ eigenvalues equal to
$1$ for $n < N$. Starting with this expression, and thus the corresponding restriction on $n,N$, a derivation of 
(\ref{CV}) can be given which is analogous to the proof of Proposition \ref{P2.11} \cite{FK09}.
\end{remark}

We have seen that the eigenvalue PDF (\ref{CV}) can, after a stereographic projection, be identified as the Boltzmann factor for a Coulomb gas model on the sphere. The concept of stereographic projection can also be applied to a pseudosphere, which refers to the two-dimensional hyperbolic space with constant negative Gaussian curvature. Doing so gives the metric specifying the Poincar\'e disk. Consideration of the solution of the Poisson equation on the latter allows (\ref{CV+}) to be interpreted as the Boltzmann factor for a Coulomb gas model on the Poincar\'e disk \cite{FK09}.

The eigenvalue correlations for (\ref{CV+}) are determinantal with correlation kernel, $K_{N,n}$ say, expressible in terms of
the incomplete beta function $I_x (a,b)$.  Thus \cite[\S~3.2.3]{Fi12}
\begin{equation}\label{CV4}
K_{N,n}(w,z) = {n \over \pi} {(1 - |w|^2)^{(n-1)/2}  (1 - |z|^2)^{(n-1)/2}  \over (1 - w \bar{z})^{n+1}} \Big (
1 - I_{w \bar{z}}(N,n+1) \Big ).
\end{equation}
In the limit $n,N \to \infty$ with $(N+n)/N = \alpha < 1$ it follows from this that
\begin{equation}\label{CV5} 
\lim_{N \to \infty} {1 \over N} \rho_{(1),N}(z)
= { (1 - \alpha) \over \alpha \pi } {1 \over (1 - |z|^2)^2} \chi_{|z|^2 < 1/ (1 + \alpha)}.
\end{equation}

Hence there is a bulk regime, and an edge regime. The neighbourhood of the origin is
typical of the bulk regime, and it follows from (\ref{CV4}) that
\begin{equation}\label{CV5a} 
\lim_{n,N \to \infty \atop n/N = \alpha} {1 \over n} K_{N,n}(w/\sqrt{n},w/\sqrt{n}) = K_\infty^{\rm b}(w,z),
\end{equation}
where $K_\infty^{\rm b}$ is specified by (\ref{2.2b}).
And after appropriate scaling about $|z| = \sqrt{\alpha)}$, the universal edge density as given
by the first term in (\ref{2.25}) is obtained \cite{Fi12},
\begin{equation}\label{CV6} 
\lim_{n,N \to \infty \atop (n+N)/n = \alpha} \alpha (1 - \alpha)  {1 \over N}\rho_{(1),N} \Big ( \sqrt{\alpha} +
{\xi \over \sqrt{N}} \sqrt{\alpha(1 - \alpha)}  \Big ) = {1 \over 2 \pi} \Big ( 1 - {\rm erf}({\sqrt  2} \xi) \Big ).
\end{equation}

A distinct scaling regime, referred to as close to unitary \cite{FS03}, or more specifically a multiplicative
rank $n$ contraction of a random unitary matrix, is obtained by taking $N \to \infty$ with $n$ fixed in (\ref{CV4}).
Scaling the eigenvalues $z_k = (1 - y_k/N) e^{i \phi_k/N}$ gives \cite{KS11} \cite[Eq.~(3.2.140)]{Fi12}
\begin{equation}\label{CV7} 
\lim_{N \to \infty} {1 \over N^2} K_{N,n}(z_1,z_2)  = {1 \over \pi}
{(2 \sqrt{y_1 y_2})^{n - 1} \over (n-1)!} \int_0^1 s^n e^{-(y_1+y_2+i (\phi_1 - \phi_2))s} \, ds.
\end{equation}
This functional form was first obtained for the scaled correlation kernel
 in the setting of a rank $n$ additive anti-Hermitian perturbation of a GUE matrix \cite[Eq.~(15) after
 setting $\nu(x) = 1/\pi$, identifying  $\tilde{z}_k  = (\phi_k + i y_k)/2$ and taking the limit
 $g_m \to 1$ ($m=1,\dots,M$) with $M$ identified as $n$]{FK99}.
 
 A generalisation of the setting of Proposition \ref{P2.12} is to consider the top $(N + L) \times N$
 sub-block of an $(n + N + L) \times (n + N + L) $ unitary matrix chosen with Haar measure. Denote such
 a rectangular matrix by $A_{N,L}$, and define from this $\tilde{A}_{N,L} = V (A_{N,L}^\dagger A_{N,L})^{1/2}$,
 where $V$ is an $N \times N$ random Haar distributed unitary matrix. It was shown in \cite{Fi12} that the
 eigenvalue PDF of $\tilde{A}_{N,L}$ is proportional to (\ref{CV}) with an additional factor of $\prod_{l=1}^N | z_l|^{2L}$.
 Denote the corresponding correlation kernel by $K_{N,n,L}$.
 Scaling $L$ by requiring that $L/N = \alpha > 0$ as $N \to \infty$, scaling the eigenvalues $z_k$ as in
 the above paragraph, and keeping $n$ fixed, it was shown in \cite{KS11}, \cite[Eq.~(3.2.122)]{Fi12}
that
\begin{equation}\label{CV8} 
\lim_{N \to \infty} {1 \over N^2} K_{N,n,L}(z_1,z_2) \Big |_{L/N = \alpha } = {1 \over \pi}
{(2 \sqrt{y_1 y_2})^{n - 1} \over (n-1)!} \int_\alpha^{\alpha +1} s^n e^{-(y_1+y_2+i (\phi_1 - \phi_2))s} \, ds.
\end{equation}
Note the consistency with (\ref{CV7}) in the limit $\alpha \to 0^+$.

Another class of generalisation is, for $U$ an $(n+N) \times (n+N)$ Haar distributed unitary matrix, and $A$ a fixed diagonal matrix with first $n$ diagonal entries $a_1,\dots,a_n$ and the remaining entries unity, to consider the product $UA$ \cite{Fy01,FS03}.
With each $a_i=0$, this corresponds to the truncated Haar unitary model of this subsection. The simplest case is when $n=1$, and furthermore we take $a_1 = a$ with $|a|<1$.
Notice that this corresponds to a multiplicative rank $1$ perturbation of $U$ --- see the recent review \cite{Fo21c} for context from this viewpoint.
The eigenvalue PDF is then proportional to
\begin{equation}\label{q0}
(1 - |a|^2)^{N}  \delta \Big ( |a|^2 - \prod_{l=1}^{N+1} | z_l|^2 \Big )
\prod_{1 \le j < k \le N+1} | z_j - z_k|^2, \quad |z_l|<1,
\end{equation}
(cf.~(\ref{CV+}))
and an explicit formula for the general $k$-point correlation function is known
\cite{Fy01,FS03}. The latter involves determinants but technically the state is not (unless $a=0$) a determinantal point process as no correlation kernel can be identified. Specifically, the functional form of the density is
\begin{equation}\label{q1}
\rho_{(1),N+1}^{UA}(z) = {1 \over \pi} {1 \over (1 - |a|^2)^N} {d \over dx} \bigg (
{x^N - 1 \over x - 1} + 
\Big ( 1 - {|a|^2 \over x} \Big )^N \bigg ) \bigg |_{x = |z|^2}.
\end{equation}
For recent works relating to the $UA$ model see \cite{FI19,DR22}.

\subsection{Products of complex Ginibre matrices} \label{S2.7}
Consideration of the eigenvalue PDF for the product $G_1 G_2$, where $G_1$ ($G_2$) are independent
$N \times p$ $(p \times N)$ rectangular complex Ginibre matrices was first undertaken by Osborn \cite{Os04}, in the context of a study relating to quantum chromodynamics. (We also mention that the GinUE has been directly used in the QCD bulk spectrum \cite{MPW99}.)
Later \cite{IK13} it was realised that this eigenvalue PDF must be the same as that for the product $\tilde{G}_1 \tilde{G}_2$ (or $\tilde{G}_2 \tilde{G}_1)$, where each $\tilde{G}_i$ is an $N \times
N$ complex random matrix with element distribution proportional to
\begin{equation}\label{7.1}
| \det \tilde{G}_i \tilde{G}_i^\dagger |^{\nu_i} e^{- {\rm Tr} \, \tilde{G}_i \tilde{G}_i^\dagger}, \quad \nu_1 =0, \: \nu_2 = p - N;
\end{equation}
note the construction (\ref{dA}) for random matrices with this element PDF. Hence it suffices to consider the square
case. Key for this is the generalised Schur decomposition (equivalent to the so-called QZ decomposition in numerical linear
algebra)
\begin{equation}\label{7.2}
 \tilde{G}_1 = U Z_{1} V, \qquad \tilde{G}_2 = V^\dagger  Z_{2} U^\dagger,
\end{equation} 
where $ Z_{1},  Z_{2} $ are upper triangular matrices with diagonal entries
$\{z^{(1)}_j \}, \{z^{(2)}_j \}$ such that 
\begin{equation}\label{7.3}
z^{(1)}_j z^{(2)}_j = z_j
\end{equation} 
 with $\{z_j \}$
 the eigenvalues of $\tilde{G}_1 \tilde{G}_2$.
 
 For the change of variables (\ref{7.2}) the wedge product strategy of the proof of Proposition
 \ref{P1.1} can again be implemented \cite{Os04}, \cite[proof of Proposition 15.11.2]{Fo10} to give
 for the Jacobian (\ref{2.1cZ}). Substituting (\ref{7.2}) in (\ref{7.1}) and recalling (\ref{7.3}),
 it follows that the eigenvalue PDF is proportional to
\begin{equation}\label{7.3a}
\prod_{l=1}^N w^{(2)}(z_l) \prod_{j<k} | z_j - z_k |^2, \quad w^{(2)}(z):=  \int_{\mathbb C} d^2 z_1\, |z_1|^{2 \nu_1} \int_{\mathbb C} d^2  z_2 \,  |z_2|^{2 \nu_2}   \delta (z - z_1   z_2)
 e^{- |z_1|^2 - |z_2|^2}.
\end{equation}
Moreover, it was shown in \cite{Os04} that $w^{(2)}(z)$ can be expressed in terms of the modified
Bessel function $K_{\nu_2 - \nu_1}(2 | z|)$, assuming $\nu_2 \ge \nu_1$.

Akemann and Burda \cite{AB12} (in the case of all matrices square), and soon after Adhikari et al.~\cite{ARRS13}
(the general rectangular case), generalised  (\ref{7.3a}) to hold for the product of $M$ complex
Gaussian matrices. As already indicated in the case $M=2$, following the work of Ipsen and Kieburg \cite{IK13}, it is now known that the rectangular case
can be reduced to the square case, where the square matrices have distribution as in (\ref{7.1}), with the $\nu_i$
equal to the difference (assumed non-negative) between the number of rows in $G_i$ and the number of rows in $G_1$ ($=N$). 
The role of the modified Bessel function, as the special function evaluating the weight $w^{(2)}(z)$ in (\ref{7.3a}),
is now played by the Meijer $G$-function
 \begin{equation}\label{GGa}
G_{p,q}^{m,n} \Big ( z \Big | {a_1,\dots, a_p \atop b_1,\dots, b_q} \Big )
=
{1 \over 2 \pi i}
\int_{\mathcal C} {\prod_{j=1}^m \Gamma ( b_j - s) \prod_{j=1}^n \Gamma(1 - a_j + s) \over
\prod_{j=m+1}^q \Gamma (1-  b_j + s) \prod_{j=n+1}^p \Gamma(a_j - s) } z^s \, ds,
\end{equation}
where $\mathcal C$ is an appropriate contour as occurs in the inversion formula for the corresponding  Mellin transform;
see e.g.~\cite{Lu69} for an extended account.

 \begin{proposition}\label{P3.1}
 Consider the product $\tilde{G}_1 \cdots \tilde{G}_M$ (in any order) of complex Gaussian matrices
 with element PDF as given in (\ref{7.1}), with each $\nu_i \ge 0$ but otherwise unrestricted. The eigenvalue
 PDF is proportional to
 \begin{equation}\label{7.4}
 \prod_{l=1}^N w^{(M)}(z_l) \prod_{j<k} | z_j - z_k |^2, \quad w^{(M)}(z) =  G^{M,0}_{0,M} \Big ( |z|^2 \Big | { \underline{\hspace{0.5cm}}  \atop \nu_1,\dots, \nu_M} \Big ).
 \end{equation}
 The functional form (\ref{7.4}) remains valid for the product of rectangular complex Ginibre matrices
 $G_1 \cdots G_M$, with $\nu_i$ specified as in the second sentence above (\ref{GGa}).
\end{proposition} 

\begin{proof} (Sketch) The generalised Schur decomposition (\ref{7.2}) can be extended to a general number $M$ of
square matrices to give what is referred to as the periodic Schur form \cite{BGV92},  \cite[Corollary 3.2]{Lu01}
 \begin{equation}\label{7.4a}
 \tilde{G}_i = U_i Z_{i} V_i, \qquad V_i = U_{i+1}^\dagger \: (i=1,\dots,M-1), \: V_M = U_1^\dagger;
 \end{equation} 
in \cite{AB12} this was deduced independently. The key features with respect to computing the
eigenvalue PDF of the product $\tilde{G}_1 \cdots \tilde{G}_M$,
as already discussed in the case $M=2$, again hold true. In particular the Jacobian factor for the change
of variables is given by (\ref{2.1cZ}), and the $k$-th diagonal entry of each $Z_{i}$  multiply together to give
the eigenvalue $z_k$ of the product. It follows that the eigenvalue PDF is given by (\ref{7.4}) with
 \begin{equation}\label{7.4a1}
w^{(M)}(z) \propto  \int_{\mathbb C} d^2 z_1\, |z_1|^{2 \nu_1} \cdots  \int_{\mathbb C} d^2  z_M \,  |z_M|^{2 \nu_M}   \delta (z - z_1  \cdots  z_M)
 e^{-\sum_{j=1}^M  |z_j|^2}.
  \end{equation} 
 As noted in \cite{AB12,ARRS13}, taking the Mellin transform of this expression leads to the Meijer $G$-function form in
 (\ref{7.4}).
 \end{proof}
 
 An easy consequence of (\ref{7.4}) is that the eigenvalue correlations are determinantal
 with kernel, $K_{N,M}$ say, given by \cite{AB12,ARRS13},
 \begin{equation}\label{7.4a2} 
K_{N,M}(w,z) =   \Big ( w^{(M)}(w)   w^{(M)}({z}) \Big )^{1/2} \sum_{j=1}^{N} { (w \bar{z})^{j-1} \over 
\prod_{m=1}^M \Gamma(j + \nu_m)}.
\end{equation}
Rigorous asymptotic analysis of (\ref{7.4a2}) can be carried out \cite{AB12,LW16}.
For example, with $w=z$ and $z \mapsto N^{M/2} z$, one obtains
 \begin{equation}\label{7.4b}
 \lim_{N \to \infty} N^{M- 1} \rho_{(1),N}( N^{M/2} z) = { |z|^{-2+2/M} \over \pi M} \chi_{|z| < 1},
\end{equation} 
where it is assumed each $\nu_i$ is fixed. To interpret this result, form the $M$-th power of a single
GinUE matrix $G$. The eigenvalues are $\{\tilde{z}_j = z_j^M \}$ with $\{z_j\}$ the the eigenvalues of $G$.
Changing variables in the circular law (\ref{2.2a}) $\tilde{z} = z^M$ gives the density (\ref{7.4b}) in the variable
$\tilde{z}$. Hence, as anticipated in \cite{BNS12}, the global eigenvalue density for the product of
$M$ independent GinUE matrices is identical to that of the $M$-th power of a single
GinUE matrix. This same global scaling limit is also well defined if some or all
of the $\nu_i$ are proportional to $N$.  An explicit limit formula has been obtained by Liu and
Wang \cite{LW16}. It is found that the density has the support of an annulus if each $\mu_i:=\lim_{N \to \infty} \nu_i/N$ is positive, but otherwise
is again supported in a disk, albeit with a density function distinct to that in (\ref{7.4b}). 

One observes from (\ref{7.4b}) a singularity at $z=0$, not present in the circular law global density
(\ref{2.2a}) for GinUE. A consequence is that the correlation kernel about this point, simply obtained
by setting the upper terminal of the sum in (\ref{7.4a2})  equal to infinity, is distinct from the corresponding
kernel (\ref{2.2b})  for the bulk scaled GinUE. On the other hand, it is shown in \cite{LW16} that the bulk scaled GinUE
kernel is reclaimed by choosing the origin a distance $\alpha \sqrt{N}$ from $z=0$, for any $0 < \alpha < 1$.
Also, from \cite{AB12,LW16} we know that the scaling of (\ref{7.4a2})  about $|z|=1$ reclaims the kernel
(\ref{2.2c}) for the edge scaled GinUE.

\begin{remark}\label{R2.17} $ $ \\
1.~Required in the derivation of (\ref{7.4b}) from (\ref{7.4a2}) is knowledge of the $z \to \infty$ asymptotic expansion
\cite{Lu69}
 \begin{equation}
G_{0,M}^{M,0} \Big ( z \Big | { \underline{\hspace{0.5cm}}  \atop \nu_1,\dots, \nu_M} \Big ) \sim
{1 \over \sqrt{M}} \Big ( {2 \pi \over z} \Big )^{(M-1)/2} e^{-M z^{1/M}} z^{(\nu_1 + \cdots + \nu_M)/M} \Big ( 1 + {\rm O}(z^{-1/M})
\Big ).
\end{equation}
Of particular interest is the exponential term herein, which with $z$ replaced by $|z|^2$ as required in (\ref{7.4a2})
reads $e^{-M |z|^{2/M}} $.  This suggests a modification of $U$ in (\ref{1.1i}) to read
\begin{equation}\label{1.1iM} 
U_M := {M \over 2} \sum_{j=1}^N |z_j|^{2/M} - \sum_{1 \le j < k \le N} \log | z_k - z_j |.
\end{equation} 
Indeed starting from $U_M$ and repeating the working of Remark \ref{R2.4} reclaims (\ref{7.4b}).
The correlation kernel appearing in such a calculation is called the Mittag-Leffler kernel. This terminology applies too in the more general case that a term $-c \sum_{j=1}^N \log |z_j|$ is included in (\ref{1.1iM}) \cite{AK13,AKS21}. This very case, for $c=(k-M)/M$, $(k=1,\dots,M)$, appears in the calculation of the joint distribution of the eigenvalues of the $M$-th power ($M \le N)$ of a GinUE matrix \cite[Th.~1.5]{Du18}. We also refer to \cite{Kat19,KatS22} for the appearance of particular Mittag-Leffler point process as a degenerate limit of so-called elliptic determinantal point processes, characterised by the product of differences in (\ref{1.1f}) being replaced by a product over Jacobi theta functions.
\\
2.~The eigenvalue PDF for a product of GinUE matrices and inverses of GinUE matrices has been
shown to be of the form (\ref{7.4}), but with the Meijer $G$-function therein replace by a different Meijer $G$-function
\cite{ARRS13}. The corresponding limiting eigenvalue density, in the case of equal numbers of matrices
and inverses, is \cite{GKT15,Ze16}
 \begin{equation}\label{7.4c} 
 \lim_{N \to \infty} \rho_{(1),N}(z) = {1 \over \pi M} {|z|^{-2+2/M} \over (1 + |z|^{2/M})^2}.
 \end{equation}
 This functional form relates to the $M=1$ case in the same way as (\ref{7.4b}) relates to (\ref{2.2a}), being the $M$-th power of the so-called spherical law.  
 \\
 3.~The case $M=2$ of Proposition \ref{P3.1} permits a generalisation. Thus for the matrices
 in the product $\tilde{G}_1 \tilde{G}_2$ choose
 $$
 \tilde{G}_1 = \sqrt{1+\tau} X_1 + \sqrt{1-\tau} X_2, \quad  \tilde{G}_2 = \sqrt{1+\tau} X_1^\dagger + \sqrt{1-\tau} X_2^\dagger, \qquad 0 < \tau < 1,
 $$
 where $X_1, X_2$ are $N \times (N+p)$ rectangular complex Gaussian matrices. The joint element distribution is then proportional to
 $$
 \exp \bigg ( - {1 \over 1 - \tau^2 } {\rm Tr} \Big ( \tilde{G}_1^\dagger   \tilde{G}_1 +  \tilde{G}_2^\dagger   \tilde{G}_2 - 2 \tau {\rm Re} \, \tilde{G}_1 \tilde{G}_2 \Big ) \bigg );
 $$
 cf.~(\ref{2.25a}). Results in  \cite{Os04} give that the eigenvalue PDF of $\tilde{G}_1 \tilde{G}_2$ is of the form in (\ref{7.3a}) but with $w^{(2)}(z)$ now dependent on $\tau$. Due to the shape of the resulting droplet, this gives rise to the so-called shifted elliptic law \cite[Th.~1]{ABK21}. \\
 4.~The product of $M$ random matrices in the limit $M \to \infty$ is of interest from a dynamical systems viewpoint, as the scaled logarithm of the singular values gives the Lyapunov spectrum. Closely related for product matrices themselves are the stability exponents, defined as ${1 \over M} \log |z_k|$, ($k=1,\dots,N$). In fact for bi-unitary invariant random matrices at the least, the Lyapunov and stability exponents are the same \cite{Re19}. Works relating to the computation of these exponents and their related statistical properties for GinUE include \cite{Fo13,Ka14,ABK14,Fo15,Ip15,LWW18,GS18,LW19,AV21}.
 \end{remark}
 
 \subsection{Products of truncated unitary matrices}
 The eigenvalue PDF for a truncation of a Haar distributed unitary matrix has been given in Proposition
  \ref{P2.12}. It turns out that knowledge of the evaluation of the matrix integral (\ref{3.y1}) as implied in the proof of
  Proposition \ref{P2.12}, used in conjunction with the periodic Schur form (\ref{7.4a}), is sufficient to allow
  for the determination of the eigenvalue PDF of the product of $M$ truncated Haar distributed unitary matrices
  \cite{ABKN13}.
  
  \begin{proposition}\label{P2.16}
  Consider $M$ independent Haar distributed $(n_j + N) \times (n_j + N)$ unitary matrices, with the top $N \times N$
  sub-block of each denoted $A_j$ ($j=1,\dots,M$). The eigenvalue PDF of the product $A_1 \cdots A_M$ (in
  any order) is given by (\ref{7.4}) with 
   \begin{equation}\label{7.4d}
   w^{(M)}(z) =    G_{M,M}^{M,0} \Big ( |z|^2 \Big | { n_1, \dots, n_M  \atop 0,\dots, 0} \Big ) \chi_{|z|<1}.
   \end{equation}  
   \end{proposition} 
\begin{proof} (Sketch) The method of proof of Proposition \ref{P3.1}, which begins with the periodic Schur
decomposition (\ref{7.4a}), and then integrates out over the upper triangular entries of the matrices $Z_j$
(here this latter step requires knowledge of the evaluation of the matrix integral (\ref{3.y1})) gives (\ref{7.4}) with 
$$
w^{(M)}(z) \propto   \int_{\mathbb C} d^2 \, z_1  \cdots  \int_{\mathbb C} d^2  \, z_M \,  \delta (z - z_1  \cdots  z_M)
  \prod_{j=1}^M (1 - | z_j |^2)^{n_j - 1}
$$
 Taking the Mellin transform of this expression leads to the Meijer $G$-function form in
 (\ref{7.4d}).

\end{proof}   

Analogous to (\ref{7.4a2}), it follows from the result of Proposition \ref{P2.16} that the eigenvalue correlations
are determinantal with kernel (to be denoted $K_{N,M}^{\rm U}$) given by \cite{ABKN13}
  \begin{equation}\label{7.4e}
K_{N,M}^{\rm U}(z_1,z_2) =   \Big ( w^{(M)}(z_1)   w^{(M)}({z_2}) \Big )^{1/2} \sum_{j=1}^{N}  (z_1 \bar{z}_2)^{j-1} 
\prod_{m=1}^M \prod_{m=1}^M \frac{(n_m+j-1)!}{(j-1)!}.
\end{equation}
Note here, that in distinction to  (\ref{7.4a2}), the $w^{(M)}$ are given by (\ref{7.4d}). Suppose all the $n_i$ are
equal to $n$, and that for large $N$, $n/N = \alpha > 0$. Using (\ref{7.4e}) in the case $z_1=z_2$, it
is derived in \cite{ABKN13} that
 \begin{equation}\label{7.4f}
 \rho_{(1),N}(z) \sim {\alpha N \over \pi M}  {|z|^{-2+2/M} \over (1 - |z|^{2/M})^2} \chi_{|z| < (1 + \alpha)^{-M/2}};
\end{equation} 
cf.~(\ref{7.4c}). 

In the same limit, with the scaling $z \mapsto z/N^{M/2}$  (and similarly for $w$), the kernel (\ref{7.4e}) multiplied by
$N^{-M}$ has the same limit as the kernel (\ref{7.4a2}) for the product of $M$ GinUE matrices about the origin in
the case of each matrix square (all $\nu_i = 0$). If instead a point $z_0$ with $0 < z_0 < (1 + \alpha)^{-M/2}$
is chosen before this scaling, one obtains instead the GinUE result (\ref{2.2b}), while at the boundary of
support, generalising (\ref{CV6}) it is found that  \cite{ABKN13}
\begin{equation}\label{CV6a} 
\lim_{n,N \to \infty \atop n/N = \alpha} \frac{M}{N}{ \alpha \over (1 + \alpha)^{M+1} } \rho_{(1),N} \bigg ( {1 \over (1 + \alpha)^{M/2}} +
{\xi \over \sqrt{N}} {\sqrt{M}\alpha^{1/2} \over (1 + \alpha)^{(M+1)/2}} \bigg ) = {1 \over 2 \pi} \Big ( 1 - {\rm erf}({\sqrt  2} \xi) \Big ),
\end{equation}
thus again exhibiting the universal functional form seen at the edge scaling of the GinUE (\ref{2.2f+}).

Also considered in  \cite{ABKN13} is a close to unity scaling, with $n_i = n$ all fixed ($i=1,\dots,M$) as
$N \to \infty$. Scaling the eigenvalues $z_k = (1 - y_k/N) e^{i \phi_k/N}$, 
it is found that (\ref{CV7}) again holds but with each occurrence of $n$ replaced by $nM$.

\subsection{The distribution of the squared eigenvalue moduli} 
All the explicit eigenvalue PDFs obtained in the above subsections of \S 2, excluding the elliptic Ginibre ensemble, have the form 
\begin{equation}\label{2.9}
\prod_{l=1}^N w( |z_l|^2) \prod_{1 \le j < k \le N} | z_k - z_j|^2,
\end{equation}
for some weight function $w$. In particular, they are invariant under rotations about the origin. An observation of Kostlan \cite{Ko92} for the GinUE case (\ref{1.1f}) of (\ref{2.9}) is that
the set of squared eigenvalue moduli $\{ | z_j |^2 \}_{j=1}^N$, appropriately ordered, are independently distributed
(specifically as gamma random variables $\{ \Gamma[j;1]\}$); see too
 \cite[Theorem 1.2]{CP14}.

 \begin{proposition}\label{P2.18}
 Let $F_w(s_1,\dots,s_N)$ denote the PDF for the distribution of $\{s_j:=|z_j|^2 \}_{j=1}^N$ for the PDF (\ref{2.9}).
 We have
\begin{equation}\label{2.9a} 
F_w(s_1,\dots,s_N) = {1 \over N!} {\rm Sym} \, \prod_{j=1}^N { w(s_j)  s_l^{l-1} \over \int_0^\infty w(s) s^{j-1} \, ds},
\end{equation}
where Sym denotes symmetrisation with respect to $\{s_j\}$. 
 \end{proposition}
 
 \begin{proof}
Starting with (\ref{2.1hm}),  then substituting (\ref{2.1h}) and its complex conjugate, we see
$$
\prod_{l=1}^N w( |z_l|^2) \prod_{1 \le j < k \le N} | z_k - z_j|^2 = \prod_{l=1}^N w( |z_l|^2)  \det [ f_{k-1}(z_j) ]_{j,k=1}^N
 \det [ g_{k-1}(z_j) ]_{j,k=1}^N, 
 $$
where $f_l(z) = z^l$, $g_l(z) = \bar{z}^l$. According to Andr\'eief's identity (see e.g.~\cite{Fo18}), it follows from this that
the integral over $z_l \in \mathbb C$ $(l=1,\dots,N)$, $I_{w,N}$ say, is itself given by a determinant
$$
I_{w,N} = N! \det \Big [ \int_{\mathbb C} w( |z_l|^2) f_{j-1}(z)   g_{k-1}(z) \, d^2 z \Big ]_{j,k=1}^N.
$$
Substituting the explicit form of $f_l, g_l$ and changing to polar coordinates $z_l = r_l e^{i \theta_l}$ 
shows that only the diagonal terms are non-zero. This allows the determinant to be evaluated
\begin{equation}\label{2.85a}
I_{w,N} = N! \pi^N \prod_{j=1}^N \int_0^\infty w(s) s^{j-1} \, ds,
\end{equation}
where in each integration the change of variables $r^2 = s$ has been made. 
Forming now $I_{w\phi,N}/I_{w,N}$, where $\phi$ is a suitable test function, and
assuming that (\ref{2.9}) is normalisable, the result (\ref{2.9a}) follows.
 \end{proof}
 
 \section{Fluctuation formulas}
 \subsection{Counting function in general domains}
The eigenvalues of a non-Hermitian matrix are examples of point processes in the plane. Statistical quantities characterising
the point process are functions of the eigenvalues $\{z_j\}$ of the form $\sum_{j=1}^N f(z_j)$ --- referred to as linear
statistics --- for given $f$. Such statistics are closely related to the correlation functions. Thus
\begin{equation}\label{4.1}
\Big \langle \sum_{j=1}^N f(z_j) \Big \rangle = \int_{\mathbb C} f(z) 
\rho_{(1),N}(z)\, d^2 z,
\end{equation} 
while 
\begin{eqnarray}\label{4.1a}
{\rm Cov} \, \Big ( \sum_{l=1}^N f(z_l),  \sum_{l=1}^N g(z_l) \Big ) & =
  \int_{\mathbb C} d^2z   \int_{\mathbb C} d^2z' \, f(z) g(z') \Big ( \rho_{(2),N}^T(z,z') + \rho_{(1),N}(z) \delta(z - z') \Big ) \nonumber \\
& =   - {1\over 2}  \int_{\mathbb C} d^2z   \int_{\mathbb C} d^2z' \, (f(z) -f(z'))(g(z) - g(z')) \rho_{(2),N}^T(z,z'),
\end{eqnarray}
see e.g.~\cite[\S 2.1]{Fo22}. Here $ \rho_{(2),N}^T(z,z') :=  \rho_{(2),N}(z,z') -  \rho_{(1),N}(z)  \rho_{(1),N}(z')$. 

One of the most prominent examples of a linear statistic is the choice $f(z) = \chi_{z \in \mathcal D}$, where
$\mathcal D \in \mathbb C$. The linear statistic is then the counting function for the number of eigenvalues in
$\mathcal D$, $N(\mathcal D)$ say. Let $E_N(n; \mathcal D)$ denote the probability that there are exactly $n$ eigenvalues in $\mathcal D$,
so that $E_N(n; \mathcal D) = {\rm Pr} \, (\sum_{j=1}^N \chi_{z_j \in \mathcal D} = n)$.
Denote the corresponding generating function (in the variable $1 - \xi$) by $\tilde{E}_N(\xi; \mathcal D)$ so that 
\begin{equation}\label{4.1b}
\tilde{E}_N (\xi; \mathcal D) = \sum_{n=0}^N E_N(n; \mathcal D) (1 - \xi)^n.
\end{equation}
Note that with $1 - \xi = e^{it}$ this corresponds to the characteristic function for the probability mass functions $\{ E_N(n;\mathcal D) \}$.
A straightforward calculation (see e.g.~\cite[Prop.~9.1.1]{Fo10}) shows that $\tilde{E}_{N} (\xi;\mathcal D)$ can
be expressed in terms of the correlation functions according to
\begin{equation}\label{4.1c} 
\tilde{E}_{N} (\xi;\mathcal D ) =  1 + \sum_{n=1}^N {(-\xi)^n \over n!}
\int_{\mathcal D} d^2z_1 \cdots \int_{\mathcal D} d^2z_n \,  \rho_{(n),N}(z_1,\dots,z_n).
\end{equation}
It can readily be checked that 
 (\ref{4.1}) and (\ref{4.1a})  in the special
case $f(z) = g(z) =  \chi_{z \in \mathcal D}$ are consistent with 
(\ref{4.1c}).

Specialising now to the circumstance that the PDF for the eigenvalues is of the form (\ref{2.9}), a particular product
formula for $\tilde{E}_N (\xi; \mathcal D)$ can be deduced, which was known to Gaudin \cite{Ga61}.

 \begin{proposition}\label{P4.1}
 Let $\{  \rho_{(n),N} \}$ in (\ref{4.1c}) be given by (\ref{2.1f}) where the correlation kernel $K_N$ corresponds to (\ref{2.9}).
 Let $\mathbb K_{N,\mathcal D}$ denote the integral operator supported on $z_2 \in \mathcal D$ with kernel
 $K_N(z_1,z_2)$. This integral operator has at most $N$ non-zero eigenvalues $\{ \lambda_j(\mathcal D) \}_{j=1}^N$, where
 $ 0 \le  \lambda_j(\mathcal D) \le 1$, and
\begin{equation}\label{4.1d}  
\tilde{E}_{N} (\xi;\mathcal D ) =  \prod_{j=1}^N (1 - \xi  \lambda_j(\mathcal D)).
\end{equation}
\end{proposition}

 \begin{proof}
 Let $\{p_s(z)\}_{s=0}^\infty$ be a set of orthogonal polynomials with respect to the inner product
 $\langle f, g \rangle := \int_{\mathbb C}w(|z|^2) f(z) g(\bar{z}) \, d^2 z$ with corresponding normalisation
 denoted $h_s$.
 Taking as a basis $\{ (w(|z|^2)^{1/2} p_k(z) \}_{k=0}^\infty$ it is straightforward to check that 
\begin{equation}\label{4.1d+}  
K_{N}(z_1,z_2) = (w(|z_1|^2) w(|z_2|^2) )^{1/2} \sum_{s=0}^{N-1} {p_s(z_1) \overline{p_s(z_2)} \over h_s},
\end{equation}
and so the eigenfunctions of $\mathbb K_{N,\mathcal D}$
 are of the form $(w(|z|^2)^{1/2} \sum_{s=0}^{N-1} c_s p_s(z)$
(see e.g.~\cite[proof of Prop. 5.2.2]{Fo10}). Hence there are at most $N$ nonzero eigenvalues, which moreover can
be related to an Hermitian matrix and so must be real. In terms
of these eigenvalues, the determinantal form (\ref{2.1f}) substituted in (\ref{4.1c}) implies (\ref{4.1d}) --- this is a result from the
theory of Fredholm integral operators (see e.g.~\cite{WW65}). Since by definition, each $n$-point correlation is non-negative,
we see from the RHS of (\ref{4.1c}) that $\tilde{E}_{N} (\xi;\mathcal D ) > 0$ for $\xi < 0$, and so $ \lambda_j(\mathcal D) \ge 0$
($j=1,\dots,N$). Also, the definition (\ref{4.1b}) tells us that $\tilde{E}_N (\xi; \mathcal D) > 0$ for all $\xi < 1$. This would
contradict (\ref{4.1d}) if it was to be that any $ \lambda_j(\mathcal D) > 1$, since in this circumstance there would
be a $\xi$ is this range such that $\tilde{E}_{N} (\xi;\mathcal D ) $ vanishes.
\end{proof}

Consider $\sum_{j=1}^N x_j$ where $x_j \in \{0,1\}$ is a Bernoulli random variable with
${\rm Pr} \,(x_j=1) = \lambda_j(\mathcal D)$. The characteristic function is $\prod_{j=1}^N(1 - \lambda_j(\mathcal D)
+ e^{it}  \lambda_j(\mathcal D))$. With $e^{it} = 1 - \xi$ this gives the RHS of (\ref{4.1d}). 
But it has already been
noted that $\tilde{E}_{N} (\xi;\mathcal D )$ with $\xi$ related to $e^{it}$ in this way is the characteristic function for
the counting statistic $\mathcal N(\mathcal D)$, and hence the equality in distribution
$\mathcal N(\mathcal D) \mathop{=}\limits^{\rm d} \sum_{j=1}^N {\rm Bernoulli} (\lambda_j(\mathcal D))$
\cite{Ha67,HKPV08}. From this, it follows using the standard arguments (see  \cite[\S~XVI.5, Theorem 2]{Fe66} ) that a central limit theorem holds
for $\mathcal N(\mathcal D_N)$ (here the subscript $N$ on $\mathcal D_N$ is to indicate that the region $\mathcal D$
depends on $N$),
\begin{equation}\label{rK1a}
\lim_{N \to \infty} { \mathcal N(\mathcal D_N) - \langle \mathcal N(\mathcal D_N) \rangle  \over ({\rm Var} \, \mathcal N(\mathcal D_N))^{1/2}}
 \mathop{=}\limits^{\rm d} {\rm N}[0,1],
\end{equation}
valid provided $ {\rm Var} \, \mathcal N(\mathcal D_N) \to \infty$ as $N \to \infty$;
see also \cite{CL95,So00,ST03}.

A stronger result, extending the central limit theorem (\ref{rK1a}), follows from the fact that (\ref{4.1d}) in the variable $z=1 - \xi$
has all its zeros on the negative real axis \cite{FL14}.

\begin{proposition}\label{T1}
In the setting of the applicability of Proposition \ref{P4.1}, and with $\sigma_{\mathcal N_D} :=  ({\rm Var} \, {\mathcal N}(\mathcal D_N))^{1/2}$,
we then have that 
$\{E_N(k;\mathcal D_N) \}$ satisfy the local central limit theorem
\begin{equation}\label{rK1}
\lim_{N \to \infty} \, \mathop{\sup}\limits_{x \in (-\infty, \infty)}
\Big |  \sigma_{\mathcal N_D}  E_N([\sigma_{\mathcal N_D}  x +  \langle {\mathcal N}(\mathcal D_N) \rangle ]; \mathcal D_N) - {1 \over \sqrt{2 \pi}} e^{- x^2/2} \Big | = 0.
\end{equation}
\end{proposition}

\begin{proof}
The fact that the zeros of (\ref{4.1d}) with the variable $z = 1 - \xi$ are on the negative real axis implies, by Newton's theorem on log-cavity of the sequence
of elementary symmetric functions \cite{Ni00}, that $\{E_N(k;\mathcal D)\}$ is log concave. It is known that log-cavity is a sufficient condition for
extending a central limit theorem to a local limit theorem \cite{Be73}.
\end{proof}

For large $N$, inside the disk of radius $\sqrt{N}$, the eigenvalue density for GinUE is constant and the full distribution is rotationally invariant.
In such circumstances, for a two-dimensional point process in general, it is known \cite{Be87} that Var$\, \mathcal N(\mathcal D_N)$ cannot grow
slower than of order $| \partial \mathcal D_N|$, i.e.~the length of the boundary of $\mathcal D_N$. Thus both (\ref{4.1c}) and (\ref{4.1d}) are valid for any region $\mathcal D_N$ constrained
strictly inside the disk of radius $\sqrt{N}$ and with a boundary of length tending to infinity with $N$.
In fact for GinUE more precise asymptotic information is available \cite[Eq.~(11)]{CE20}, \cite[Eq.~(2.7)]{FL22}, which gives that for any $D_0 \subseteq \{z: |z| \le 1 \}$,
 \begin{equation}\label{KL}
 {\rm Var} \, \mathcal N(\sqrt{N}  D_0) = \sqrt{N} {| \partial  D_0| \over 2 \pi} \int_{-\infty}^\infty \Big ( {\rm Var} \, \chi_{U \le (1+{\rm erf}(t/\sqrt{2}))/2} \Big ) \, dt + {\rm O} \Big ( {1 \over N^{1/2}} \Big ).
\end{equation}
Here $U$ is a Bernoulli random variable.
A direct calculation  gives that the integral evaluates to $\sqrt{1 \over  \pi}$ --- the advantage of the form (\ref{KL}) is that it remains
true if the appearance of the variance throughout is replaced by any even cumulant \cite{CE20,FL22}. In particular this shows that the growth of the
variance with respect to the region is the smallest order possible --- the corresponding point process is then referred to as being hyperuniform  \cite{TS03, To16,  GL17}. 
A corollary of the property of being hyperuniform, together with the fast decay of the correlations, is that the bulk scaled GinUE exhibits
number rigidity \cite{GP17,GL17a}. This means that conditioning on the positions of the (infinite number of) points outside a region $\mathcal D$ fully determines the
number of (but not positions of) the points inside $\mathcal D$, and their centre of mass.

\subsection{Counting function in a disk}\label{S3.2}

In the special case that $\mathcal D_N$ is a disk of radius $R$ centred at the origin (we write this as $ D_R$), the polynomials
 in (\ref{4.1d+}) are simply the monomials $p_s(z) = z^s$. The eigenfunctions of $\mathbb K_{N,\mathcal D}$ are
 also given in terms of the monomials as $\{(w(|z|^2))^{1/2} z^{j-1} \}_{j=1,\dots,N}$, and hence for the
 corresponding eigenvalues we have
 $$
 \lambda_j( D_R) = \int_0^{R^2}  r^{j-1} w(r) \, dr \Big /  \int_0^{\infty}  r^{j-1} w(r) \, dr, \qquad j=1,\dots,N.
 $$
 Substituting in (\ref{4.1d}) and choosing $w(|z|^2) = \exp(-|z|^2)$ then shows that for the GinUE \cite{Fo92b}
 \begin{equation}\label{rK2+}
\tilde{E}_N(\xi;D_R) = \prod_{j=1}^N \Big ( 1 - \xi {\gamma(j;R^2) \over \Gamma(j)} \Big ),
\end{equation}
where $\gamma(j;x)$ denotes the (lower) incomplete gamma function.
Note that this remains valid for $N \to \infty$ in keeping with the discussion of the previous paragraph.
Setting $\xi = 1$,  asymptotic expansions for the incomplete gamma function can
be used to deduce the $N \to \infty$ asymptotic expansion of $E_N(0; D_{\alpha \sqrt{N}})$,
 \begin{equation}\label{rK2p}
E_N(0; D_{\alpha \sqrt{N}}) = \exp \Big ( C_1N^2 + C_2 N \log N + C_3 N + C_4 \sqrt{N} + {1 \over 3}
\log N + {\rm O}(1) \Big ), \quad 0 < \alpha < 1.
\end{equation}
Here the constants $C_1,\dots,C_4$ depend on $\alpha$ and are known explicitly (e.g.~$C_1=-\alpha^4/4$), 
being first given in \cite{Fo92b}. The first two of these can be deduced from the large $R$ expansion of the quantity $F_\infty(0;D_R)$, defined in Remark \ref{R3.3}.1 below, given in the still earlier work \cite{GHS88}.
The $\log N$ term was determined recently in \cite{Ch21}, as  too was the explicit form of the next order term, a 
constant with respect to $N$.

Let us also mention that for any $p \ge 2$, the $p$-th cumulant $\kappa_p(D_R)$ of the number of eigenvalues in $D_R$ can be written as  
\begin{equation}\label{cum fin N} 
\kappa_p(R) = (-1)^{p+1} \sum_{j = 0}^{N-1} {\rm{Li}}_{1 - p}\Big( 1 - \frac{1}{ \lambda_j( D_R) } \Big),
\end{equation}
where ${\rm{Li}}_s(x) = \sum_{k = 1}^{\infty} k^{-s} x^k$ is the polylogarithm function.
The formula \eqref{cum fin N} as well as its large $N$ behaviour both in the bulk and at the edge were obtained in \cite{LMS19}. 
 
We turn our attention now to the circumstance that a (possibly scaled) large $N$ limit has already been taken, and ask about
the fluctuations of the number of particles in a region $\mathcal N(\mathcal D)$ for large values of $|\mathcal D|$, i.e.~the
volume of $\mathcal D$. The first point to note is that if the coefficient of $\xi^n$ in  (\ref{4.1c}) tends to zero as $N \to \infty$, then the expansion
 remains valid in this limit \cite{Le73}. The decay is easy to establish in the determinantal case, since then (see e.g.~\cite[Eq.~(9.13)]{Fo10})
 $\rho_{(n),N}(z_1,\dots,z_n) \le \prod_{l=1}^n \rho_{(1),N}(z_l)$. Hence it is sufficient that $\int_\Omega  \rho_{(1),N}(z) \, d^2z$ be bounded
 for $N \to \infty$. With the limiting form of (\ref{4.1c}) valid, the theory of Fredholm integral operators \cite{WW65} tells us that the limit of
 (\ref{4.1d}) is well defined with  the RHS identified as the Fredholm determinant $\det(\mathbb I - \xi \mathbb K_{\infty,\mathcal D})$.
 In (\ref{rK2+}) the limit corresponds to simply replacing the upper terminal of the product by $\infty$.
 We stipulate the further structure that the correlation kernel be Hermitian, as holds for the appropriately scaled
 form of (\ref{4.1d+}). Then the argument of the proof of Proposition \ref{P4.1} tells us that the eigenvalues of $ \mathcal K_{\infty,\mathcal D}$
 are all between $0$ and $1$, which in turn allows the reasoning leading to (\ref{rK1a}) to be repeated.
 The conclusion is, assuming ${\rm Var} \, \mathcal N(\mathcal D) \to \infty$ as $|\mathcal D| \to \infty$ which
 as already remarked is guaranteed by the results of \cite{Be87}, that
 (\ref{rK1a}) remains valid
 with $\mathcal D_N$ replaced by $\mathcal D$, and the limit $N \to \infty$ replaced by the limit $|\mathcal D| \to \infty$ 
 \cite{CL95,So00,HKPV08}. 
 
 It is moreover the case that in the above setting and with these modifications the local
 central limit theorem of Proposition \ref{T1} remains valid \cite{FL14}. Another point of interest is that the expansion
 (\ref{rK2p}) is uniformly valid
in the variable $R=\alpha \sqrt{N}$, provided this quantity grows with $N$, and hence also provides the large $R$
expansion of $E_\infty(0; D_R)$. Finally we consider results of \cite{Le83} as they apply to number fluctuations in the infinite GinUE. The plane is to be divided into squares $\Gamma_j$ of area $L^2$ with centres at $L \mathbb Z^2$. Define $\Upsilon_j = \mathcal N(\Gamma_j)/\sqrt{{\rm Var} \,
\mathcal N(\Gamma_j)}$. For large $L$, in keeping with (\ref{KL}) we have ${\rm Var} \, (\Gamma_j) \sim 2L/\pi^{3/2}$. The question of interest is the joint distribution of $\{ \Upsilon_j \}$. It is established in \cite{Le83} that for $L \to \infty$ this distribution is Gaussian, with covariance ${1 \over 4} [-\Delta]_{j,k}$, where $\Delta$ is the discrete Laplacian on $\mathbb Z^2$. Consequently, fluctuations of $N(\Gamma_j)$ induce opposite fluctuations in the regions neighbouring $\Gamma_j$.

\begin{remark}\label{R3.3} $ $ \\
1.~Closely related to the probability $E_N(N;\mathcal D)$ is the conditioned quantity $F_N(n;\mathcal D) :={\rm Pr} (\sum_{j=1}^N \chi_{z_j \in \mathcal D} = n
| z_j = 0)$. Denote the corresponding generating function by $\tilde{F}_N(\xi;\mathcal D)$. Proceeding as in the derivation of (\ref{rK2+})  
shows $\tilde{F}_N(\xi; D_R) = \tilde{E}_N(\xi; D_R)/(1 - \xi(1 - e^{-R^2}))$. Thus in particular ${F}_N(0; D_R) = e^{R^2} E_N(0; D_R)$
\cite{GHS88}. Note that $-{d \over d R} F(0;D_R)$ gives the spacing distribution between an eigenvalue conditioned to be at the origin, and its nearest neighbour at a distance $R$. The work \cite{SJ12} gives results relating to the PDF for the minimum of all the nearest neighbour spacings with global scaling, obtaining a scale of $N^{-3/4}$ and a PDF proportional to $x^3 e^{-x^4}$.\\
2.~Let $\bar{ D}_{\alpha \sqrt{N}}$ denote the region $\{z: |z| > \alpha \sqrt{N} \}$, i.e.~the region outside the disk
of radius $\alpha \sqrt{N}$, where it is assumed $0 < \alpha < 1$. Note that then $E_N(0;\bar{ D}_{\alpha \sqrt{N}}) =
E_N(N; D_{\alpha \sqrt{N}})$.
The analogue of (\ref{rK2p}) has been calculated in
\cite{Ch21}, where in particular it is found that 
\begin{equation}\label{L1}
C_1=\alpha^4/4-\alpha^2+(1/2)\log \alpha^2 + 3/4,
\end{equation}
the coefficient $C_4$ is unchanged, while the coefficient ${1 \over 3}$ for
$\log N$ seen in (\ref{rK2p}) is to be replaced by $- {1 \over 4}$. 
(We also refer to \cite{CMP16} for an earlier work for which the expansion (\ref{rK2p}) was obtained up to $C_3$.)  
One sees from (\ref{L1}) that $C_1=0$ for
$\alpha = 1$, and the result of \cite{Ch21} gives that $C_2,C_3$ similarly vanish, giving that $E_N(N; D_{ \sqrt{N}}) \sim
e^{C_4 \sqrt{N}}$. Extending $\alpha$ larger that $1$ according to the precise $N$ dependent value
\begin{equation}\label{3.13a}
\alpha = 1 + {1 \over 2 \sqrt{N}} \Big ( \sqrt{\gamma_N} + {x \over  \sqrt{\gamma_N}} \Big ), \qquad
\gamma_N = \log {N \over 2 \pi} - 2 \log \log N,
\end{equation}
gives the extreme value result $\lim_{N \to \infty} E_N(N; D_{ \alpha_N \sqrt{N}}) = \exp(-\exp(-x))$ \cite{Ri03a}
(see too \cite[Th.~1.3 with $\alpha = 2$]{CP14} for a generalisation to the case of (\ref{1.1iM}), considered also in
\cite{Ch21} for $\alpha<1$), which is
the Gumbel law. Other references on fluctuations of the spectral radius under various boundary conditions include
\cite{CGJ20,Seo20,AKS20,Gar22,BS21,BL22}.
Furthermore, an intermediate fluctuation regime which interpolates between the Gumbel law with the large deviation regime \eqref{L1} was investigated in \cite{L+18}. 
Another case considered in \cite{Ch21} is when $D_N$ is
specified as the outside of an annulus contained inside of the disk of radius $\sqrt{N}$. 
Two features of the corresponding asymptotic expansion (\ref{rK2p}) are: (1) the absence of a term proportional to $\log N$, and (2) the presence of oscillations of order 1 that are described in terms of the Jacobi theta function. 
We also refer to \cite{CL22,ACCL22a,ACCL22b} for further recent studies in this direction in the presence of hard edges. 
\\
3.~For general $\mathcal D_N$ with $| \mathcal D_N | \to \infty$ the coefficient $C_1$ in (\ref{rK2p}) relates to an energy
minimisation (electrostatics) problem, and similarly for the $| \mathcal D| \to \infty$ expansion of $E_\infty(0;\mathcal D)$ \cite{Dy62e,JLM93,CMP16,AR17,Ad18}.
Thus for $\mathcal D = D_{\alpha \sqrt{N}}$ the electrostatics problem is to compute the potential due to a uniform charge density $1/\pi$ inside a disk of radius $\alpha$, with a neutralising uniform surface charge density $-\alpha/2\pi$ on the boundary.
The applicability of electrostatics remains true for the asymptotic expansion of $E_N(k; D_{\alpha \sqrt{N}})$ (and $E_\infty(k;\mathcal D)$) in
the so-called large deviation regime, when $k \ll N \alpha^2$ or $k \gg N \alpha^2$. For a disk the electrostatics problem can be
solved explicitly to give \cite{ATW14}
\begin{equation}\label{I1}
E_N(k; D_{\alpha \sqrt{N}}) \sim e^{-N^2 \psi_0(\alpha;k/N)}, \quad \psi_0(\alpha;x) = {1 \over 4} \Big | (\alpha^2 - x)(\alpha^2 - 3 x) -2
x^2 \log(x/\alpha^2) \Big |.
\end{equation}
Note in particular that $\psi_0(\alpha;0) = \alpha^4/4$, which is the value of $-C_1$ in (\ref{rK2p}), while setting $x=1$
gives the value of $-C_1$ noted in the above paragraph. There is also a scaling regime, where $|k - N \alpha^2| = {\rm O}(N^{1/2})$, for the asymptotic value of
$E_N(k; D_{\alpha \sqrt{N}})$ which interpolates between (\ref{I1}) and the local central limit theorem result
(\ref{rK1}) \cite{L+19,FL22}. In the case of $E_\infty(k; D_R)$, it makes sense to consider $k$ proportional to not only $\alpha R^2$ but also to $\alpha R^\gamma$
with $\gamma > 2$. Then \cite{JLM93,FL22}
$$
E_\infty(\alpha R^\gamma; D_R) \sim e^{- {1 \over 2} (\gamma - 2) \alpha^2 R^{2 \gamma} \log R (1 + {\rm o}(1) )}.
$$
\\
4.~For the infinite GinUE the exact result in terms of modified Bessel functions
 \begin{equation}\label{J1}
 {\rm Var} \, \mathcal N( D_R) = R^2 e^{-2 R^2} \Big ( I_0(2 R^2) + I_1(2 R^2) \Big ) =  \sum_{j=1}^\infty \frac{ \gamma(j; R^2) }{ \Gamma(j) } \Big(  1- \frac{ \gamma(j; R^2) }{ \Gamma(j) } \Big)
 \end{equation}
 is known \cite[Th.~1.3]{Sh06}, \cite[Appendix B]{FL22}. The second expression in \eqref{J1} also appears in \cite{LMS19} as a large $N$ limit of the number variance of the finite Ginbire ensemble in the deep bulk regime.
 This exhibits the leading large $R$ form $R/\sqrt{\pi}$ which is consistent with identifying $\sqrt{N} | \partial  D_0|$
 as $2 \pi R$ on the RHS of (\ref{KL}); see also \cite{ABE23}.
\end{remark}

\subsection{Smooth linear statistics}
The theory of fluctuation formulas for GinUE in the case that $f(z)$ in (\ref{4.1}) is smooth has some different features to
the discontinuous case $f(z) = \chi_{z \in \mathcal D}$. This can be seen by considering the bulk scaled limit, and in particular the truncated two-point correlation (\ref{2.2t}).
From this we can compute the structure factor
 \begin{equation}\label{J2a}
 S_\infty^{\rm GinUE}(\mathbf k) := \int_{\mathbb R^2} \Big ( \rho_{(2), \infty}^{{\rm b} ,T}(\mathbf 0,  \mathbf r)  + {1 \over \pi} \delta(\mathbf r) \Big ) e^{i \mathbf k \cdot \mathbf r} \,
 d \mathbf r = {1 \over \pi} \Big ( 1 - e^{-|\mathbf k|^2/4} \Big ).
 \end{equation}
 Knowledge of the structure factor allows the limiting covariance (\ref{4.1a}) to be computed using the Fourier transform
  \begin{equation}\label{J2b}
 {\rm Cov}^{\rm GinUE_\infty} \, \Big ( \sum f(\mathbf r_l),  \sum  g(\mathbf r_l) \Big )   =
{1 \over (2 \pi)^2}  {1 \over  \pi}  \int_{\mathbb R^2}  \hat{f}(\mathbf k) \hat{g}(- \mathbf k)   \Big ( 1 - e^{-|\mathbf k|^2/4} \Big ) \, d \mathbf k,
 \end{equation} 
 valid provided the integral converges. Here, with $z = x + i y$, $\mathbf r = (x,y)$ and the Fourier transform $\hat{f}(\mathbf k)$ is defined by integrating $f(\mathbf r)$ times $e^{i\mathbf k \cdot \mathbf r}$ over $\mathbb R^2$ --- thus according to (\ref{J2a}) $S_\infty^{\rm GinUE}(\mathbf k)$ is a particular Fourier transform. Now introduce a scale $R$ so that $f(\mathbf r) \mapsto f(\mathbf r/R), \, g(\mathbf r) \mapsto g(\mathbf r/R)$. It follows
 from (\ref{J2b}) that
 \begin{equation}\label{J2b+}
\lim_{R \to \infty}  {\rm Cov}^{\rm GinUE_\infty} \, \Big ( \sum f(\mathbf r_l/R),  \sum  g(\mathbf r_l/R) \Big )   =
{1 \over (2 \pi)^2}  {1 \over 4 \pi}  \int_{\mathbb R^2}  \hat{f}(\mathbf k) \hat{g}(- \mathbf k) |\mathbf k|^2 \, d \mathbf k,
 \end{equation} 
 again provided the integral converges. Most noteworthy is that this limiting quantity is O$(1)$. In contrast, with $f(\mathbf r) = g(\mathbf r) = \chi_{|\mathbf r| < 1}$,
 and then introducing $R$ as prescribed above, we know  that (\ref{J2b}) has the evaluation (\ref{J1}). As previously commented, the large $R$ form
 of the latter is proportional to $R$.

 \begin{remark} 
 Consider the linear statistic $A(\mathbf x) = - \sum_{j=1}^N ( \log | \mathbf x - \mathbf r_j| - \log | \mathbf r_j| )$. 
 In the Coulomb gas picture relating to (\ref{1.1i}), this corresponds to the
 difference in the potential at $\mathbf x$ and the origin. For bulk scaled GinUE, one has from (\ref{J2b}) the exact result \cite{AJ84}
  \begin{equation}\label{J2c}
  {\rm Var}^{\rm GinUE_\infty} \, A(\mathbf x) = {1 \over 2} \Big ( 2 \log | \mathbf x| + ( |\mathbf x|^2 + 1 ) \int_{|\mathbf x|^2}^\infty {e^{-t} \over t} \, dt - e^{- | \mathbf x|^2} + C + 1 \Big ),
 \end{equation}   
 where here $C$ denotes Euler's constant. In particular, for large $|\mathbf x|$, $ {\rm Var}^{\rm GinUE_\infty} \, A(\mathbf  x)  \sim \log | \mathbf x |$. This last point shows
 that the introduction of a scale $R$ as in (\ref{J2b+}) would give rise to a divergence proportional to $\log R$. Such a log-correlated structure underlies a relationship between the logarithm of the absolute value of the characteristic polynomial for GinUE and Gaussian multiplicative chaos \cite{La20}.
 \end{remark}

 The covariance with test functions $f(\mathbf r) \mapsto f(\mathbf r/\sqrt{N})$, $g(\mathbf r) \mapsto g(\mathbf r/\sqrt{N})$  assumed to take on
 real or complex values is also an order one quantity for GinUE 
 in the $N \to \infty$ limit, upon the additional assumption that $f,g$ are differentiable and don't grow too fast at infinity \cite{Fo99,RV07,AHM11,AHM15}.

 \begin{proposition}\label{P3.5}
Require that $f,g$ have the properties as stated above. Let $f(\mathbf r) |_{\mathbf r = (\cos \theta, \sin \theta)} = \sum_{n=-\infty}^\infty f_n e^{i n \theta}$ and similarly for the Fourier expansion of
 $g(\mathbf r) $ for $\mathbf r = (\cos \theta, \sin \theta)$. We have
  \begin{equation}\label{5.2e} 
  \lim_{N \to \infty} {\rm Cov}^{\rm GinUE} \Big ( \sum_{j=1}^N f(\mathbf r_j/\sqrt{N}),  \sum_{j=1}^N \bar{g}(\mathbf r_j/\sqrt{N}) \Big ) =
  {1 \over 4 \pi  } \int_{ |\mathbf r  | < 1}   \nabla f \cdot  \nabla \bar{g} \, dx dy +
  {1 \over 2} \sum_{n=-\infty}^\infty |n| f_n  \bar{g}_{-n}.
   \end{equation}   
   \end{proposition}
   
   \begin{proof} (Sketch) In the method of \cite{RV07}, a direct calculation using (\ref{4.1a}), (\ref{2.2b}) and (\ref{2.1g}) allows (\ref{5.2e}) to be established
   for $f,g$ polynomials jointly in $z=x+iy$ and $\bar{z} = x - iy$. The required integrations can be computed exactly using polar coordinates. To go beyond the
   polynomial case, the so-called dbar (Cauchy-Pompeiu) representation is used. This gives that for any once continuously differentiable $f$ in the unit
   disk $ D_1$, and $z$ contained in the interior of the disk,
    \begin{equation}\label{DB}
   f(z) = - {1 \over \pi } \int_{ D_1} {\partial_{\bar{w}} f(w) \over w - z} \, d^2 w + {1 \over 2 \pi i} \int_{\partial  D_1} {f(w) \over
   w - z} \, dw,
     \end{equation} 
   where with $w = \alpha + i \gamma$, $\partial_{\bar w} = {1\over 2} ( {\partial \over \partial  \alpha} + i {\partial \over \partial  \gamma}  )$. The covariance
   problem is thus reduced to the particular class of linear statistics of the functional form $h(z) = 1/(w-z)$. The required analysis in this case is
   facilitated by the use of the corresponding Laurent expansion, with only a finite number of terms contributing after integration.
   \end{proof}


\begin{remark}\label{R3.6}  $ $ \\
1.~As predicted in \cite{Fo99}, upon multiplying the RHS by $2/\beta$, (\ref{5.2e}) remains valid for the Coulomb gas
model (\ref{1.1i}) \cite{LS18,Se18,BBNY19}. 
In the case of the elliptic GinUE, a simple modification of (\ref{5.2e}) holds true. Thus the
domain $| \mathbf r | < 1$ in the first term is to be replaced by the  appropriate ellipse, and the Fourier components
of the second term are now in the variable $\eta$, where  $(A \cos \eta, B \sin \eta)$,
$0 \le \eta \le 2 \pi$ parametrises the boundary of the ellipse. The results of \cite{Fo99,AHM15,LS18} also cover this case. \\
2.~In the case of an ellipse, major and minor axes $A,B$ say, there is particular interest in the linear statistic $P_x := \sum_{j=1}^N x_j$  \cite{CPR87}.
Linear response theory gives for the $xx$ component of the susceptibility tensor $\chi$ --- relating the polarisation density to the
applied electric field --- the formula $\chi_{xx} = (\beta/(\pi A B)) \lim_{N \to \infty} {\rm Var} \, P_x$ (and similarly for the $xy$ and $yy$ component).
This same quantity can be computed by considerations of macroscopic electrostatics, which gives $\chi_{xx} = (A+B)/(\pi B)$. Using (\ref{5.2e}) modified
as in the above paragraph, the consistency of these formulas can be verified. \\
3.~Considering further the case of elliptic GinUE, dividing by $N$ and taking the limit $\tau \to 1$ gives the GUE with eigenvalues supported on $(-2,2)$, and similarly
for the $\beta$ generalisation limiting to (\ref{1.1i}) restricted to this interval. For this model it is known (see e.g.~\cite[Eq.~(3.2) with the identification $x= 2 \cos \theta$]{Fo22}) 
$$
\lim_{N \to \infty} {\rm Cov} \Big ( \sum_{j=1}^N f( x_j),  \sum_{j=1}^N  {g}(x_j) \Big ) =
  {2 \over \beta} \sum_{n=1}^\infty  n f_n^{\rm c} g_n^{\rm c},
 $$
 where $f(x) |_{x = 2 \cos \theta} = f_0^{\rm c} + 2 \sum_{n=1}^\infty f_n^{\rm c} \cos n \theta$ and similarly for $g(x)$.
 We observe that this is identical to the final term in (\ref{5.2e}), modified according to the specifications of point 1.~above. \\
 4.~There has been a recent application of Proposition \ref{P3.5} in relation to the computation of the analogue of the Page curve for a density matrix constructed out of GinUE matrices \cite{CK22}
\end{remark}

We turn our attention now to the limiting distribution of a smooth linear statistic. By way of introduction,
consider the particular linear statistic ${1 \over N} \sum_{j=1}^N | \mathbf r_j |^2$ for GinUE. An elementary calculation
gives that the corresponding characteristic function, $\hat{P}_N(k)$ say, has the exact functional form 
\begin{equation}\label{SL1}
\hat{P}_N(k) =
(1 - ik/N)^{-N(N+1)/2}.    
\end{equation}
It follows from this that after centring by the mean, the limiting distribution is a
Gaussian with variance given by (\ref{5.2e}) (which is this specific case evaluates to one). A limiting Gaussian
form holds in the general case of the applicability of (\ref{5.2e}), as first proved by Rider and Vir\'ag \cite{RV07}.

 \begin{proposition}\label{P3.7}
 Let $f$ be subject to the same conditions as in Proposition \ref{P3.5}, and denote the case $f = g$ of (\ref{5.2e}) by $\sigma_f^2$.
  For the GinUE, if $f$ takes on complex values then as $N \to \infty$
 $$
  \sum_{j=1}^N f(\mathbf r_j/\sqrt{N}) - \Big \langle  \sum_{j=1}^N f(\mathbf r_j/\sqrt{N}) \Big \rangle \mathop{\to}\limits^{\rm d} {\rm N}[0,\sigma_f]
  + i {\rm N}[0,\sigma_f],
  $$
  while if $f$ is real valued the RHS of this expression is to be replaced by $ {\rm N}[0,\sigma_f]$. Moreover this same limit formula holds for
  the elliptic GinUE \cite{AHM15} and its $\beta$ generalisation \cite{LS18} (both subject to further technical restrictions on $f$), with the variance modified according to Remark \ref{R3.6}.1.
  \end{proposition}

 \begin{proof} (Comments only) The proof of  \cite{RV07} proceeds by establishing that the higher order cumulants beyond the
 variance tend to zero as $N \to \infty$. Essential use is made of the rotation invariance of GinUE. The method of \cite{AHM15} uses a
 loop equation strategy, while \cite{LS18} involves energy minimisers and transport maps. For GinUE with $f$ a function of $| \mathbf r|$,
 a simple derivation based on the proof of Proposition \ref{P2.18} together with a Laplace approximation of the integrals \cite{Fo99} (see also
 \cite[Appendix B]{BF22}).
 \end{proof}

 \begin{remark} Other settings in which Proposition \ref{P3.7} has proved to be valid include products of GinUE matrices \cite{CO20,KOV20} (with the additional assumption that the test function have support strictly inside the unit circle), and for the complex spherical ensemble of subsection \ref{S2.5} after stereographic projection onto the sphere \cite{RV07a,Be12}.
 \end{remark}
 
 \subsection{Spatial modelling and the thinned GinUE}
 The GinuE viewed as a point process in the plane has been used to model geographical
 regions by way of the corresponding Voronoi tessellation \cite{LD93}, the positions of
 objects, for example trees in a plantation \cite{Le90} or the nests of birds of prey \cite{A+21},
 and the spatial distribution of base stations in modern wireless networks \cite{MS14,DZH15}, amongst
 other examples. The wireless network application has made use of the thinned GinUE, whereby
 each eigenvalue is independently deleted with probability $(1 - \zeta)$, $0 < \zeta \le 1$. The effect
 of this is simple to describe in terms of the correlation functions, according to the replacement
  \begin{equation}\label{5.3a} 
  \rho_{(n),N}(z_1,\dots,z_n) \mapsto \zeta^N  \rho_{(n),N}(z_1,\dots,z_n).
  \end{equation}
  With the bulk density of GinUE uniform and is equal to $1/\pi$, we can also rescale the
  position $z_j \mapsto z_j/\zeta$ so that this remains true in the thinned ensemble. For this
  (\ref{5.3a}) is to be updated to read 
  \begin{equation}\label{5.3b} 
  \rho_{(n),N}(z_1,\dots,z_n) \mapsto \rho_{(n),N}(z_1/\sqrt{\zeta},\dots,z_n/\sqrt{\zeta}).
  \end{equation} 
  
  Recalling now (\ref{2.1f}) and (\ref{2.2b}), for the bulk scaled limit of the thinned GinUE
  we have in particular
  $$
  \rho_{(1),\infty}^{\rm tGinUE}(z) = {1 \over \pi}, \qquad   \rho_{(2),\infty}^{{\rm tGinUE}, T}(w,z) = - {1 \over \pi^2} e^{-|w - z |^2/\zeta}.
  $$
  From these functional forms we see
  $$
  \int_{\mathbb C}  \rho_{(2),\infty}^{{\rm tGinUE}, T}(w,z)  \, d^2z = - {\zeta \over \pi} \ne -  \rho_{(1),\infty}^{{\rm tGinUE}}(w) \qquad {\rm unless} \: \zeta = 1,
  $$
  where the superscript ``tGinU'' denotes the thinned GinUE.
  Equivalently, in terms of the structure factor (\ref{J2a}),
  $$
  S_\infty^{{\rm tGinUE}}(\mathbf 0) = {1 - \zeta \over \pi} \ne 0  \qquad {\rm unless} \: \zeta = 1.
  $$
  Due to this last fact, the O$(1)$ scaled covariance for smooth linear statistics (\ref{J2b}) is no longer true,
  and now reads instead
   \begin{equation}\label{J2b+a}
  {\rm Cov}^{\rm tGinUE} \, \Big ( \sum f(\mathbf r_l/R),  \sum  g(\mathbf r_l/R) \Big )   \mathop{\sim}\limits_{R \to \infty}
{R^2 \over (2 \pi)^2}  {(1 - \zeta) \over  \pi}  \int_{\mathbb R^2}  \hat{f}(\mathbf k) \hat{g}(- \mathbf k) \, d \mathbf k.
 \end{equation} 
This leading dependence on $R^2$ holds too for the counting function $f(z) = \chi_{|z| < 1}$, since in distinction
to (\ref{J2b}) the integral now converges. Hence, in the terminology of the text introduced
below (\ref{KL}), the statistical state is no longer hyperuniform. There is an analogous change to the O$(1)$ scaled covariance
(\ref{5.2e}), which is now proportional to $N$ and reads
 \begin{equation}\label{5.2e+} 
  {\rm Cov}^{\rm tGinUE} \Big ( \sum_{j=1}^N f(\mathbf r_j/\sqrt{N}),  \sum_{j=1}^N \bar{g}(\mathbf r_j/\sqrt{N}) \Big )    \mathop{\sim}\limits_{N \to \infty}
 N  {(1 - \zeta)\over  \pi  } \int_{ |\mathbf r  | < 1}    f \bar{g} \, dx dy. 
    \end{equation} 
Notwithstanding this difference, the corresponding limiting distribution  function is still Gaussian \cite{La19}. A more
subtle limit, also considered in \cite{La19}, is when $N \to \infty$ and $\zeta \to 1^-$ simultaneously, with $N(1-\zeta)$ fixed.
The quantity (\ref{5.2e}) returns to being O$(1)$, but consists of a contribution of the form (\ref{5.2e}), and a term
characteristic of a Poisson process. 

We turn our attention now to the probabilities $\{ E_N^{\rm tGUE}(k;D_{\alpha \sqrt{N}}) \}$. Upon consideration of the thinning
prescription (\ref{5.3b}), the proof of Proposition \ref{P4.1}, and (\ref{rK2+}) shows that the corresponding generating function
is given by
 \begin{equation}\label{5.2f} 
 \tilde{E}_N^{\rm tGUE}(\xi;D_{\alpha \sqrt{\zeta N}})  = \prod_{j=1}^N \bigg ( 1 - \xi \zeta {\gamma(j;\alpha^2  N) \over \Gamma(j)} \bigg ).
  \end{equation} 
Setting $\xi = 1$ in this gives the probability $ E_N^{\rm tGUE}(0;D_{\alpha \sqrt{\zeta N}})$. 
Note that the implied formula shows $E_N^{\rm tGUE}(0;D_{\alpha \sqrt{\zeta N}}) = \tilde{E}_N^{\rm GUE}(\zeta;D_{\alpha \sqrt{ N}})$. The large $N$ asymptotics of $\tilde{E}_N^{\rm GUE}(\zeta;D_{\alpha \sqrt{ N}})$, and various generalisations, are available in the literature \cite{Ch22,BC22}. Here we present a self contained derivation of the first two terms (cf.~(\ref{rK2p})).

 \begin{proposition} 
For large $N$ and with $0 < \alpha, \zeta < 1$ we have
 \begin{equation}\label{5.2g} 
  \tilde{E}_N^{\rm tGUE}(0;D_{\alpha \sqrt{\zeta N}})
 =  \exp \Big ( {\alpha^2 N } \log (1 - \zeta) + \sqrt{\alpha^2 N } \, h(\zeta) + {\rm O}(1) \Big ),
  \end{equation}
  where
 \begin{equation}\label{5.2h}   
 h(\zeta) = \int_0^\infty \log \Big ( {1 - (\zeta/2) (1 + {\rm erf}(t/\sqrt{2})) \over 1 - \zeta} \Big ) \, dt +
  \int_0^\infty \log \Big ( 1 - (\zeta/2) (1 - {\rm erf}(t/\sqrt{2})) \Big ) \, dt.
   \end{equation}  
\end{proposition}

\begin{proof}
Our main tool is the uniform asymptotic expansion \cite{Tr50}
 \begin{equation}\label{Tr}
 {\gamma(M-j+1;M) \over \Gamma(M-j+1)} \mathop{\sim}\limits_{M \to \infty} {1 \over 2} \Big ( 1 + {\rm erf} \Big ( {j \over \sqrt{2M}} \Big ) \Big );
  \end{equation}  
  cf.~the leading term in (\ref{2.2e}).
  Here it is known that the error term has the structure $(1/\sqrt{M})g(j/\sqrt{2M})$ where $g(t)$ is integrable on $\mathbb R$ and decays rapidly at infinity.
 This expansion suggests we rewrite
 the product in (\ref{5.2f}) with $\xi = 1$ in the form 
 $$
 (1 - \zeta)^{[M^*]} \bigg ( \prod_{j=1}^{M^*} {1 - \zeta \gamma(j;M^*)/\Gamma(j) \over 1 - \zeta} \bigg )
  \bigg ( \prod_{j=M^*+1}^{N} (1 - \zeta \gamma(j;M^*)/\Gamma(j)) \bigg ),
  $$
  where $M^* = [\alpha^2 N]$. 
  
  We see that the first term in this expression gives the leading order term in (\ref{5.2g}). In the first product we change labels $j \mapsto M^* - j + 1$
  $(j=1,\dots,M^*)$. In the second we change labels $j \mapsto M^* + j + 1$ ($j=0,\dots,N - M^* - 1$). Now writing both these products as
  exponentials of sums and applying (\ref{Tr}) gives, upon recognising the sums as Riemann integrals, the O$( \sqrt{\alpha^2 N } )$ term in
  (\ref{5.2g}).
  \end{proof}
  
  The leading term in (\ref{5.2g}) is consistent with the general form expected for thinned log-gas systems, being of the form
  of the area of the rescaled excluded region, times the density, times $\log (1 - \zeta)$ \cite[Conj.~10]{Fo14a}.

  \begin{remark}
  The topic of spatial modelling using Ginibre eigenvalues naturally leads to questions on the efficient simulation of the point process confined to a compact subset of the $\mathbf C$. Practical algorithms for this task have been given in \cite{DFV15,DFPT16}.
  \end{remark}
  
 \section{Sum rules and asymptotic behaviours}
 Throughout this section we consider the Coulomb gas model (\ref{1.1i}) with general $\beta > 0$, for which we use the notation OCP, which stands for one-component plasma. The special case $\beta = 2$ coincides with GinUE.
 \subsection{Asymptotics associated with the configuration integral} \label{S4.1}
  The configuration integral for the Boltzmann factor (\ref{1.1i}) is specified by
 \begin{equation}\label{QN}
 Q_N^{\rm OCP}(\beta) = \int_{\mathbb C}d^2z_1 \cdots \int_{\mathbb C}d^2z_N \,
 e^{-(\beta/2) \sum_{j=1}^N |z_j|^2}
 \prod_{1 \le j < k \le N}|z_k - z_j |^\beta.
 \end{equation}
 From the OCP viewpoint, it is more natural to consider the renormalised quantity
 \begin{equation}\label{QN1}
 Z_N^{D_R,\rm OCP}(\beta) = {1 \over N!}
 A_{N,\beta} Q_N^{\rm OCP}(\beta), \qquad A_{N,\beta} = e^{-\beta N^2 ({1 \over 4} \log N - {3 \over 8})}
;
\end{equation}
see \cite[Eq.~(1.12)]{Fo10}. The use of the symbol $D_R$ denoting a disk of radius $R$ (specifically $R = \sqrt{N}$) as a subscript follows from the derivation of  $A_{N,\beta}$.
Thus this quantity  has the interpretation as the Boltzmann factor of the electrostatic self energy of a smeared out uniform background, charge density $-{1 \over \pi}$, confined to the disk $D_R$, and the constant terms of its electrostatic energy when coupled to a particle inside of this disk. The quantity $Z_N^{D_R,\rm OCP}(\beta)$ is then the partition function of the charge neutral OCP. Generally in statistical mechanics for a stable system  the dimensionless free energy, $\beta F_{N}(\beta) = - \log Z_N(\beta)$, is an extensive quantity, meaning that for large $N$ it is proportional to $N$. For the closely related model when the particles are strictly restricted to the disk, the validity of this statement was established in \cite{LN75,SM76}, and has been reconsidered recently using more far reaching techniques in \cite{LS17} (which for example form a platform for the study of fluctuation formulas in \cite{LS18}). Making use of Proposition \ref{P1.1} the large $N$ form of $\beta F_{N}^{D_R,\rm OCP}(\beta)$ can be computed for $\beta = 2$ \cite[Eq.~(3.14)]{TF99}.

\begin{proposition}
We have
\begin{equation}\label{QN2}
\beta F_{N}^{D_R,\rm OCP}(\beta) \Big |_{\beta = 2} = N \beta f(\beta) |_{\beta = 2} + {1 \over 12} \log N - \zeta'(-1) - {1 \over 720 N^2} + {\rm O} \Big ( {1 \over N^4} \Big ), 
\end{equation}
where
\begin{equation}\label{QN2+}
\beta f(\beta) |_{\beta = 2} = {1 \over 2} \log \Big ( {1 \over 2 \pi^3} \Big ).
\end{equation}
\end{proposition}

\begin{proof}
This relies on (\ref{2.1a}), identifying $\prod_{j=1}^{N-1} j! = G(N+1)$, where $G(x)$ denotes the Barnes $G$-function, and knowledge of the known asymptotic expansion of $G(N+1)$ (see e.g.~\cite[Th.~1]{FL01}).
\end{proof}

For the OCP on a sphere of radius $R$, $S_R^2$ say, and with $R={1 \over 2} \sqrt{N}$ so that the particle density is $1/\pi$, the partition function of the charge neutral system is (see e.g.~\cite[Eq.~(2.1)]{TF99})
$$
Z_N^{S_R^2,\rm OCP}(\beta) = {1 \over N!} N^{-N\beta/2} e^{\beta N^2/4}
\int_{S^2_R} d \theta_1 d \phi_1 \cdots 
\int_{S^2_R} d \theta_N d \phi_N \,
\prod_{1 \le j < k \le N} | u_k v_j  - u_j v_k |^\beta.
$$
Here the variables $\{u_j,v_k\}$ are the Cayley-Klein parameters as in (\ref{CV2}). From the analogue of Proposition \ref{P1.1} this can be evaluated exactly at $\beta =2$ \cite{Ca81}, implying that for large $N$ \cite{JMP94}, \cite[Eq.~(3.6)]{TF99}
\begin{equation}\label{QN3}
\beta F_{N}^{S_R^2,\rm OCP}(\beta) \Big |_{\beta = 2} = N \beta f(\beta) |_{\beta = 2} + {1 \over 6} \log N +{1 \over 12} - 2 \zeta'(-1) + {1 \over 180 N^2} + {\rm O} \Big ( {1 \over N^4} \Big ), 
\end{equation}
where $\beta f(\beta) |_{\beta = 2}$ is as in (\ref{QN2}).

Both expansions (\ref{QN2}) and (\ref{QN3}) illustrate a conjectured universal property of the large $N$ expansion of $\beta F_N^{\rm OCP}$ in the case that the droplet forms a shape with Euler index $\chi$ \cite{JMP94}
\begin{equation}\label{QN4}
\beta F_N^{\rm OCP}(\beta) = N \beta f(\beta) + a_\beta \sqrt{N} + {\chi \over 12} \log N + \cdots,
\end{equation}
valid for general $\beta > 0$. To compare against the exact results for $\beta = 2$ it should be recalled that $\chi = 1$ for a disk and $\chi = 2$ for a sphere. 
An analogous calculation in annulus geometry at $\beta = 2$ gives an expansion with a term proportional to $\log N$ absent, in keeping with $\chi = 0$ \cite{FF11}, \cite{BKS22}. A further point of interest is that the large $N$ expansions for $E_N(0;\mathcal D_N)$ from \cite{Ch21} as reviewed in \S~\ref{S3.2}, also exhibit simple fractions for the coefficient of $\log N$, and moreover this term is not present in the case the eigenvalues are constrained to a single annulus.

We note that the term proportional to $\sqrt{N}$ in (\ref{QN4}) has the interpretation as a surface tension, and so is expected not to be present in the case of a sphere, as seen in (\ref{QN3}) for $\beta =2$. Also, for the disk geometry, it has been conjectured (see \cite[Eq.~(3.2)]{CFTW15}) that
\begin{equation}\label{QN5}
a_\beta = {4 \log (\beta/2) \over 3 \pi^{1/2}},
\end{equation}
which in particular vanishes for $\beta=2$, as is consistent with (\ref{QN3}).

\begin{remark} $ $ \\
1.~Multiplication of the configuration integral $Q_N^{\rm OCP}(\beta)$ by $A_{N,\beta}$ as in (\ref{QN1}) effectively shifts the microscopic energy $U$ in (\ref{1.1i}) by a function of $N$. We denote this shifted, charge neutral, energy by $U'$, and similarly in relation to the sphere. For $\beta = 2$ direct calculation of the mean is possible. Thus in the case of the OCP on the plane, one finds \cite[equivalent to Eqns.~(25) and (29)]{Sh11} 
\begin{align*}
\langle U' \rangle^{D_R} \Big |_{\beta = 2} = &
-{1 \over 2} \bigg ( 
{N^2 \over 2}(\Psi(N) - \log N) 
 + {N + 1 \over 4} +{N C \over 2}  - {\Gamma(N + 3/2) \over \Gamma(N+2) \Gamma(3/2)} \\
 & \quad \times
{}_3 F_2 \Big ( {1,N-1,N+3/2 \atop N+2,N+1}
\Big | 1 \Big ) \bigg ) 
 = - {C N \over 4}
+ {2 \sqrt{N} \over 3 \sqrt{\pi}}   - {5 \over 48} + {\rm O}(N^{-1/2}) ,
\end{align*}
where here $C$ denotes Euler's constant and $\Psi(N)$ denotes the digamma function. The corresponding result for the sphere geometry at $\beta =2$ gives the simpler formula
\cite{CLWH82, AZ15, ST16}
$$
\langle U' \rangle^{S_R^2} \Big |_{\beta = 2} =
{N \over 4} \Big (  - H_N + \log N \Big )= -{C N \over 4}-{1 \over 8 } +{\rm O}(N^{-1}),
$$
where here $H_N$ denotes the harmonic numbers. The common leading order value of the charge neutral energy per particle, $-C/4$, was known to Jancovici \cite{Ja81} through the formula
$$
\lim_{N \to \infty} {1 \over N} \langle U' \rangle^{\rm OCP} \Big |_{\beta = 2} = - {\pi \over 2}
\int_{\mathbb R^2} \log | \mathbf r| \, \rho_{(2), \infty}^{\rm b, \rm GinUE}(\mathbf 0, \mathbf r) \, d \mathbf r = - {C \over 4},
$$
where the integral is evaluated from the explicit formula for $\rho_{(2), \infty}^{\rm b, \rm GinUE}$ given by (\ref{2.2t}). Note as an expansion about $\beta=2$, $\beta F_N(\beta) =\beta F_N(\beta)|_{\beta =2} + (\beta - 2) \langle U' \rangle + {\rm O}((\beta/2 - 1)^2) $, so the above results allow (\ref{QN2}) and (\ref{QN3}) to be extended to first order in $(\beta - 2)$. In particular consistency is obtained with (\ref{QN4}). \\
2.~In the low temperature limit $\beta \to \infty$ the OCP particles are expected to form a triangular lattice. Recent works relating to this include \cite{SS12,RS15,LS18,BS18,Am18,CSA21,AS21,AR22}.  
The conjectured exact value of twice the charge neutral energy per particle in this limit, with bulk density $1/(4\pi)$ (not $1/\pi$ as is natural for GinUE in the plane, rather the geometry used was a sphere of unit radius) is \cite{BHS12}
$$
2 \log 2 + {1 \over 2}\log {2 \over 3} + 3 \log {\sqrt{\pi} \over \Gamma(1/3)} = -0.0556053\dots.
$$
We also refer to  \cite{CHM18,AB19,La21} and references therein for recent works on the opposite, high temperature limit $\beta \to 0$ of the two dimensional Coulomb particles. 
\end{remark}

 \subsection{Sum rules and asymptotics for the edge density}\label{S4.2}
  First we note that for large $N$ the leading order support of the density is a disk of radius $\sqrt{N}$, independent of $\beta$, as follows from the potential theoretic argument of Remark \ref{R2.4}.2. Using the vector coordinate $\mathbf r = (x,y)$, as in the second equation in (\ref{2.2c}) we introduce edge scaling coordinates by writing
 \begin{equation}\label{15.1a}
 \rho_{(1),N}^{\rm OCP}((x,
 \sqrt{N}-y)) = \rho_{(1),\infty}^{\rm e, OCP}(y) + {1 \over \sqrt{N}} \mu^{\rm e, OCP}(y)+
 {\rm O} \Big ( {1 \over N} \Big ).
\end{equation} 
Here the form of the correction terms, known to be valid at $\beta = 2$ according to (\ref{2.2f}), are at this stage presented as an ansatz. We seek some integral identities that must be satisfied by $\rho_{(1),\infty}^{\rm e, OCP}(y)$ and
$\mu_{(1),\infty}^{\rm e, OCP}(y)$. Identities of this type are referred to as sum rules.

\begin{proposition}\label{P4.1x}
    We have
\begin{equation}\label{15.1b} 
\int_{-\infty}^\infty \Big ( \rho_{(1),\infty}^{\rm e, OCP}(y) 
- {1 \over \pi} \chi_{y>0} \Big ) \, dy = 0.
\end{equation}
and
\begin{equation}\label{15.1c} 
\int_{-\infty}^\infty y \Big ( \rho_{(1),\infty}^{\rm e, OCP}(y) 
- {1 \over \pi} \chi_{y>0} \Big ) \, dy =  \int_{-\infty}^\infty   \mu_{(1),\infty}^{\rm e, OCP}(y) 
 \, dy.
 \end{equation}
\end{proposition}

\begin{proof}
Because of the large $N$ form of the density, we have that $\rho_{(1),N}^{\rm e, OCP}(\mathbf r) 
- {1 \over \pi} \chi_{|\mathbf r| < \sqrt{N}}$ will be concentrated near $|\mathbf r| = \sqrt{N}$. Now, the normalisation condition for the density gives, with the use of polar coordinates
$$
\int_0^\infty r \Big ( \rho_{(1),N}^{\rm e, OCP}(\mathbf r) 
- {1 \over \pi} \chi_{|\mathbf r| < \sqrt{N}} \Big ) \, dr =
\int_{-\infty}^{\sqrt{N} }
(\sqrt{N} - y) \Big ( \rho_{(1),N}^{\rm e, OCP}(\sqrt{N}-y) 
- {1 \over \pi} \chi_{y>0} \Big ) \, dr = 0.
$$
In the second term
we now substitute (\ref{15.1a}) and equate terms of order $\sqrt{N}$ and of order unity to zero to obtain (\ref{15.1b}) and (\ref{15.1c}).
\end{proof} 

From the functional forms for
$\rho_{(1),\infty}^{\rm e, GinUE}(y)$ and
$\mu_{(1),\infty}^{\rm e, GinUE}(y)$ as read off from (\ref{2.2f}), we verify both of the above sum rules in this special case. Note that the integral on the LHS of (\ref{15.1c}) has the interpretation of the dipole moment of the excess charge in the edge boundary layer, using the Coulomb gas picture. In fact it is possible to derive a further sum rule which evaluates this dipole moment explicitly \cite[Eq.~(5.13)]{TF12}.

\begin{proposition}
    We have
\begin{equation}\label{15.1d} 
\int_{-\infty}^\infty y \Big ( \rho_{(1),\infty}^{\rm e, OCP}(y) 
- {1 \over \pi} \chi_{y>0} \Big ) \, dy = - {1 \over 2 \pi \beta} \Big ( 1 - {\beta \over 4} \Big ).
\end{equation}
\end{proposition}

\begin{proof}
    From the definition, we observe
\begin{equation}\label{15.1e} 
\Big \langle \sum_{j=1}^N | \mathbf r_j|^2 \Big \rangle = -{\partial \over \partial c} \log \bigg (
\int_{\mathbb R^2} d \mathbf r_1 \cdots \int_{\mathbb R^2} d \mathbf r_N \, e^{-c \sum_{j=1}^N |\mathbf r_j|^2} \prod_{1 \le j < k \le N}
|\mathbf r_k - \mathbf r_j|^\beta \,
\bigg ) \bigg |_{c = \beta /2}.
\end{equation} 
The $c$ dependence of the integral can be factored by a simple change of variables, allowing the integral to be replaced by $c^{-N - \beta N(N-1)/4}$, and so giving
\begin{equation}\label{15.1f}
\Big \langle \sum_{j=1}^N | \mathbf r_j|^2 \Big \rangle = {2N \over \beta}+{1\over 2} N(N-1).
\end{equation}
By writing the LHS as an average over the density and the use of polar coordinates, we see that this is equivalent to the
sum rule
$$
2 \pi \int_0^\infty r^3 \Big ( \rho_{(1),N}^{ \rm OCP}( r) 
- {1 \over \pi} \chi_{r < \sqrt{N}} \Big ) \, dr = {N \over 2} \Big ( {4 \over \beta} - 1 \Big ).
$$
Proceeding now as in the proof of Proposition \ref{P4.1x}, by substituting (\ref{15.1a}), equating terms proportional to $N$ on both sides, and making use too of (\ref{15.1c}), we deduce (\ref{15.1d}).
\end{proof}

One observes that the RHS of (\ref{15.1d}) changes sign as $\beta$ increase beyond $4$.
In fact in the work \cite{CFTW14} it is predicted that for general $\beta > 2$, the edge density profile of the OCP exhibits an overshoot effect where it rises before tailing off to zero. The exact evaluation of the edge density to leading order in $\beta - 2$ in \cite{CFTW15} lends analytic evidence to this claim. This edge density overshoot effect has been observed in the random matrix ensemble of even dimensional random matrices $Z_N W$, where $Z_N={\small \mathbb I_N \otimes \begin{bmatrix} 0 & 1\\ -1 & 0 \end{bmatrix}}$ and $W$ is a complex anti-symmetric Gaussian random matrix \cite{Ha01,Fo15}. Moreover, in \cite{Ha01}, upon the assumption of large eigenvalue separation, an analytic calculation of the joint eigenvalue PDF gives the functional form of the OCP with $\beta = 4$.

We now turn our attention to the large $N$ form of the global scaled density, $\rho_{(1),N}^{\rm g, OCP}(\mathbf r) := \rho_{(1),N}^{\rm OCP}(\sqrt{N} \mathbf r)$. From the theory noted at the beginning of the subsection, this will limit to the circular law (\ref{2.2a}). We ask about the leading correction term. A hint is given by (\ref{15.1f}), which after dividing both sides by $N^2$ to correspond to global coordinates tells us $\langle {1 \over N} \sum_{j=1}^N | \mathbf r_j |^2 \rangle_{\rm g, OCP} = {1 \over 2} - {1 \over 2} \Big ( 1 - {4 \over \beta} \Big ) {1 \over N}.$ The first term ${1 \over 2}$ is the average of the function $g(\mathbf r) = |\mathbf r|^2$ over the disk $D_{\sqrt{N}}$ with density ${1 \over \pi}$. The second term is a ${1 \over N}$ correction, so we might expect that the leading correction to the circular law is O$(1/N)$. In fact knowledge of (\ref{15.1d}) is sufficient to compute the leading correction term of the large $N$ expansion of all the moments $\langle {1 \over N} \sum_{j=1}^N | \mathbf r_j |^p \rangle_{\rm g, OCP}$ ($p=1,2,\dots)$, allowing us to conclude \cite[Eq.~(5.18)]{TF12}
\begin{equation}\label{15.1f+}
\rho_{(1),N}^{\rm g, OCP}(\mathbf r) =
{1 \over \pi} \chi_{|\mathbf r | < 1} +
{1 \over N} \kappa(r)+ 
{\rm O}(N^{-3/2}), \quad
\kappa(r)={1 \over 2 \pi \beta} \Big (1 - {\beta \over 4} \Big ) {1 \over r} \delta'(r-1).
\end{equation}
Thus the correction term is concentrated on the boundary. For GinUE, it can be deduced from (\ref{2.1g}) that inside the droplet the corrections are exponentially small \cite[Th.~1.2]{LR16},\cite[Lemma 3.1]{Ja20}. This same formula can also be read off from a more general formula relating to $\beta$ generalised normal matrix models (see \S~\ref{S5.3}) obtained by Zabrodin and Wiegmann \cite[Eq.~(5.16)]{ZW06}. Associated with (\ref{15.1f+}) is the expansion \cite[Eq.~(1.14)]{LS17}
\begin{equation}\label{15.1f++}
\Big \langle \sum_{j=}^N g(\mathbf r_j) \Big \rangle^{\rm g, OCP} = {1 \over \pi} \int_{|\mathbf r| < 1} g(\mathbf r) \, dx dy + {1 \over N} {1 \over 2 \pi \beta} \Big (1 - {\beta \over 4} \Big )
\int_{|\mathbf r| < 1} \nabla^2 g(\mathbf r) \, dx dy + {\rm o}(N^{-1}),
\end{equation}
valid for sufficiently smooth test functions $g$ (note that rotational invariance is not assumed). In the case $\beta =2$ this expansion can be found in \cite[Th.~2.1]{AHM15}.

As our final topic under this heading, we consider the $y \to - \infty$ asymptotic form of $\rho_{(1),N}^{\rm e, OCP}(y)$. According to  (\ref{2.2f+}), at $\beta = 2$ we have
\begin{equation}\label{Tr+}
\rho_{(1),N}^{\rm e, OCP}(y) \mathop{\sim}\limits_{y \to - \infty}
{e^{-2 y^2} \over ( 2\pi )^{3/2} |y|}.
\end{equation}
As a first step to extend this to general $\beta > 0$, a large deviation formula for $\rho_{(1),N}^{\rm OCP}(\mathbf r)$ can be computed, which asks for the asymptotic form of $\rho_{(1),N}^{\rm OCP}(\sqrt{N} r)$ (here polar coordinates are being used), $r > 1$
\cite[Prop.~1]{CFTW15}.

\begin{proposition}
  With  $\beta f (\beta)$ the dimensionless free energy per particle as in (\ref{QN4}), for $r > 1$ we have
  \begin{equation}
    \label{eq:17.4}
    \rho_{(1),N}^{\rm OCP}(\sqrt{N}r) = \frac{e^{\beta f (\beta)}}{N^{\beta /4}}e^{-(N\beta /2)(r^2-1)}  \exp \Big( N\beta \log r - \frac{\beta}{2} \log (r^2-1) + {\rm o}(1) \Big).
  \end{equation}
\end{proposition}

\begin{proof} (Sketch)
The starting point is to manipulate the definition of $\rho_{(1),N}^{\rm OCP}$ to obtain its form written as an average
\begin{equation}\label{eq:Av1}
  \rho_{(1),N+1}^{\rm OCP}(\sqrt{N+1}\vec{r}) = (N+1)N^{\beta N/2} e^{-(N+1)\beta r^2/2} \frac{Q_{N}^{\rm OCP}(\beta)}{Q_{N+1}^{\rm OCP}(\beta)}  \Big \langle \prod_{l=1}^{N}\Big |\sqrt{\frac{N+1}{N}}\vec{r}- \vec{r_l} \Big |^\beta  \Big \rangle_{\rm OCP^{\rm g}}.
\end{equation}
Here $Q_{N}^{\rm OCP}(\beta)$ is the configuration integral (\ref{QN}), and the average is with respect to the global scaled GUE, specified by the Boltzmann factor (\ref{1.1i}) but with the factor of ${1 \over 2}$ multiplying the first sum in $U$ replace by ${N \over 2}$. The significance of this is that upon exponentiating the product, the average can be recognised as the  characteristic function for a particular linear statistics. For this, with  $| \mathbf r| = r> 1$ Proposition \ref{P3.7} applies, telling us the leading two terms in its large $N$ asymptotic  expansion, once the corresponding mean and variance have been computed. After calculating these, (\ref{eq:17.4}) results.
\end{proof}

From (\ref{eq:17.4}) 
we compute the limit formula
\begin{equation}\label{Lb}
\lim_{N \to \infty} {\rho}_{(1),N}^{\rm OCP} (\sqrt{N}r)|_{r = 1-y/\sqrt N} 
 = e^{\beta f (\beta)} \frac{e^{-\beta y^2}}{(2 |y|)^{\beta /2}},
 \end{equation}
 which under the assumption that the large deviation formula connects to the $y \to - \infty$ tail of $\rho_{(1),N}^{\rm e, OCP}(y)$ is the sought $\beta > 0$ generalisation of the $\beta = 2$ result (\ref{Tr+}). The consistency between (\ref{Lb}) and the latter is immediate upon substituting (\ref{QN2+}).

\subsection{Sum rules and asymptotics for the two and higher point correlations}\label{S4.2c}
Setting $\mathbf k = \mathbf 0$ in (\ref{J2a}) gives
\begin{equation}\label{PSb}
\int_{\mathbb R^2} \Big ( \rho_{(2), \infty}^{{\rm b},T}(\mathbf 0,  \mathbf r)  + {1 \over \pi} \delta(\mathbf r) \Big )  \,
 d \mathbf r = 0.
\end{equation}
This constraint on $\rho_{(2), \infty}^T$ is an example of a sum rule. In physical terms, using the Coulomb gas picture, it says that the response of the system by the introduction of a charge (corresponding to the delta function) is to create a screening cloud (corresponding to $\rho_{(2), \infty}^T$) of opposite total charge. Consequently (\ref{PSb}) is referred to as the perfect screening sum rule. It relates to the point process  being hyperuniform or equivalently incompressible --- for a state that is compressible the RHS of (\ref{PSb}) is not zero but rather is given in terms of the second derivative of the pressure with respect to the fugacity; see e.g.~\cite[Eq.~(3.7)]{Fo21b}. For states with an underlying long range potential (as in (\ref{1.1i})) the perfect screening sum rule is expected to be a necessary condition for thermodynamic stability \cite{Ma88}. Of similar general validity is the sum rule which results when the integrand of (\ref{PSb}) is replaced by
\begin{multline}\label{M.1}
q(\mathbf r_1,\dots,\mathbf r_k,\mathbf r):=
\rho_{(k+1),\infty}(\mathbf r_1,\dots,\mathbf r_k,\mathbf r) \\ - {1 \over \pi}
\rho_{(k),\infty}(\mathbf r_1,\dots,\mathbf r_k) +
 \sum_{j=1}^k \delta(\mathbf r - \mathbf r_j) \rho_{(k),\infty}(\mathbf r_1,\dots,\mathbf r_k).
\end{multline}
Furthermore the fast decay of the correlations upon truncation (i.e.~suitable subtraction by combinations of lower order correlation as in the definition of $\rho_{(2),\infty}^T$) implies that not only does the total charge associated with $q$ vanish, but in fact so too does all the multipole moments \cite{Ma88}.

\begin{proposition}
Let $q$ be as in (\ref{M.1}) and set $\mathbf r = (x,y)$.
    For bulk scaled GinUE, specified in terms of the correlation functions by (\ref{2.1f}) with $N \to \infty$ and the correlation kernel (\ref{2.2b}), we have 
    \begin{equation}\label{PSc}
\int_{\mathbb R^2}   (x-iy)^p
q(\mathbf r_1,\dots,\mathbf r_k,\mathbf r) \, d \mathbf r = 0, \qquad p=0,1,\dots
\end{equation}
    \end{proposition}

\begin{proof}
 For any $k \ge 1$ integrating over the term involving the sum of delta functions gives $\Big ( \sum_{j=1}^k (x_j-iy_j)^p \Big ) \rho_{(k),\infty}$. To then integrate over the first two terms, we expand the determinant specifying $\rho_{(k+1),\infty}$ by the final row. The term coming from the last entry is recognised as $(1/\pi) \rho_{(k),\infty}$ and so cancels. For each of the $k$ other terms, we multiply the term coming from the $j$-th entry of the final row $K_\infty^{\rm b}(\mathbf r,\mathbf r_j)$ times $(x-iy)^p$ down the final column containing $[K_\infty^{\rm b}(\mathbf r_m, \mathbf r) ]_{m=1}^k$, and integrate over $\mathbf r$. For the latter task, polar coordinates can be used to deduce that
 $$
 \int_{\mathbb R^2}(x-iy)^p 
 K_\infty^{\rm b}(\mathbf r_m, \mathbf r)
 K_\infty^{\rm b}(\mathbf r,\mathbf r_j)
  \, d \mathbf r = (x_j-iy_j)^p 
 K_\infty^{\rm b}(\mathbf r_m,\mathbf r_j);
 $$ 
 the case $p=0$ is (\ref{2.2r}) with bulk scaling. After rearranging the columns, the result of the integration in each case can be identified as $-(x_j-iy_j)^p \rho_{(k),\infty}$, and thus in total cancel out with the integration over the final term.
\end{proof}

Replacing $1/\pi$ in (\ref{M.1}) by $\rho_{(1),\infty}(\mathbf r)$ the sum rule (\ref{PSc}) with $p=0$ remains valid in the case of edge scaling as a consequence of the validity of (\ref{2.2r}).  For $p \ge 1$ the slow decay of the correlations along the direction of the boundary as seen in (\ref{2.2t2}) means that the integral in (\ref{PSc}) is not well defined. Specifically the $p=0$, $k=1$ sum rule (\ref{PSc}) at the edge reads
\begin{equation}\label{SA1}
\int_{\mathbb C}\rho_{(2),\infty}^{{\rm e, OCP},T}(z,z') \, d^2 z' = - \rho_{(1),\infty}^{{\rm e, OCP}}(z).
\end{equation}
This  is the edge counterpart of the bulk perfect screening sum rule (\ref{PSb}). 

\begin{remark} 
It turns out that in the GinUE case, the determinantal structure together with (\ref{SA1}) can be used to show that $\rho(z):=\rho_{(1),\infty}^{{\rm e, GinUE}}(z)$ 
(this is (\ref{2.2f+})) satisfies a non-linear equation of infinite order,
\begin{equation}\label{SA2}
\rho(z) = \sum_{n=0}^\infty {| \partial_z^{(n)} \rho(z) |^2 \over n!};
\end{equation}
see \cite[\S~3.6]{AKM20}, where (\ref{SA1}) is referred to as a mass-one equation.
The non-linear equation is equivalent to the special function
function identity
\begin{equation}\label{SA3}
{1 \over 4} \Big ( 1 + {\rm erf}(\sqrt{2}
x) \Big ) \Big ( 1 - {\rm erf}(\sqrt{2}
x) \Big ) = {e^{-4 x^2} \over \pi}
\sum_{n=1}^\infty {(H_{n-1}(\sqrt{2} x))^2 \over 2^n n!} \quad (x \in \mathbb R);
\end{equation}
see \cite[\S~4.5]{AKM20} for the proof of \eqref{SA3}.
Together with the loop equation, the identity (\ref{SA2}) is used in \cite{AKM20} to study the universality of normal matrix models; see \S \ref{S5.4}.
\end{remark}

To deduce (\ref{J2b+}) from (\ref{J2b}) requires that
\begin{equation}\label{SS1}
\lim_{R \to \infty} R^2
S_\infty^{\rm GinUE}(\mathbf k/R) = {|\mathbf k|^2 \over 4 \pi}.
\end{equation}
In the general $\beta > 0$ case of the OCP we 
denote the bulk scaled structure factor by $S_\infty^{\rm OCP}(\mathbf k)$.
 A perfect screening argument extending the viewpoint which implies (\ref{PSb}) (for this see e.g.~\cite[\S 15.4.1]{Fo10}, \cite[\S 3.2]{Fo98a}) predicts
\begin{equation}\label{SS1b}
\lim_{R \to \infty} R^2
S_\infty^{\rm OCP}(\mathbf k/R) = {|\mathbf k|^2 \over 2 \beta \pi}.
\end{equation}
This result, in the slightly different guise of a mesoscopic scaling limit, is implied by the recent work \cite[Th.~1 mesoscopic case, formula for the variance]{LS18}. In the cases of the OCP applied to the anomalous quantum Hall effect ($\beta = 2 M$ with $M$ odd), the sum rule (\ref{SS1b}) combined with the Feynman-Bijl formula quantifies the collective mode excitation energy from the ground state \cite{GMP86}. From the definition of the structure factor, this is equivalent to the moment formula
\begin{equation}\label{SS1c}
\int_{\mathbf R^2}| \mathbf r|^2
\rho_{(2),\infty}^{{\rm OCP},T}(\mathbf r, \mathbf 0)
\, d \mathbf r = - {2 \over \pi \beta},
\end{equation}
known as the Stillinger-Lovett sum rule \cite{SL68}; see \cite{MG83} for a derivation which makes use of (\ref{PSc}) for $p=1$ and
$k=1,2$. Higher order moment formulas can also be derived,
\begin{equation}\label{SS1d}
\int_{\mathbf R^2}| \mathbf r|^4
\rho_{(2),\infty}^{{\rm OCP},T}(\mathbf r, \mathbf 0)
\, d \mathbf r = - {16 \over \pi \beta^2}
\Big ( 1 - {\beta \over 4} \Big ), \quad
\int_{\mathbf R^2}| \mathbf r|^6
\rho_{(2),\infty}^{{\rm OCP},T}(\mathbf r, \mathbf 0)
\, d \mathbf r = - {18 \over \pi \beta^3}\Big ( \beta - 6 \Big ) \Big ( \beta - {8 \over 3} \Big );
\end{equation}
in distinction to (\ref{SS1b}) the precise statement of these results depends on the underlying bulk particle density, which in keeping with GinUE has been assumed to be ${1 \over \pi}$. A linear response argument can be used to derive the first of these; see e.g.~\cite[\S 14.1.1]{Fo10}. It involves the thermal pressure for the OCP, which by a simple scaling argument takes the value $\rho^{\rm b} (1 - \beta/4)$ (here $\rho^{\rm b}$ denotes the bulk density; for a discussion of the relation between the thermal and mechanic pressure in the OCP, see \cite{CFG80}) which explains the appearance of this factor. Mayer diagrammatic expansion methods were used to first derive the sixth moment condition in (\ref{SS1d}) \cite{KMST00}.
Later a response argument involving variations to the spatial geometry was used to give an alternative derivation \cite{CLW15}. In keeping with the relationship between (\ref{SS1b}) and (\ref{SS1c}), the moment formulas (\ref{SS1d}) can be related to the small $k$ expansion of $S_\infty^{\rm GinUE}(\mathbf k)$, which must therefore read
\begin{equation}\label{SS1e}
{2 \pi \beta \over | \mathbf k|^2} S_\infty^{\rm OCP}(\mathbf k) = 1 + \Big ( {\beta \over 4} - 1 \Big ) {|\mathbf k|^2 \over 2 \beta} +\Big ( {\beta \over 4} - {3 \over 2} \Big )\Big ( {\beta \over 4} - {2 \over 3} \Big )\Big ( {|\mathbf k|^2 \over 2 \beta}\Big )^2 + {\rm O}(|\mathbf k|^6).
\end{equation}
Evidence is given in \cite{KMST00} that the polynomial structure in $\beta/4$ of the coefficients in this expansion breaks down at ${\rm O}(|\mathbf k|^6)$. This is in contrast to the power series expansion of the bulk scaled (density $1/\pi$ for definiteness) structure factor for the log-gas on a one-dimensional domain, $S_\infty^{(\beta)}(k)$ say. Thus expanding $\pi \beta S_\infty^{(\beta)}(k)/k $ as a power series in $(k/\beta)$, the $j$-th coefficient is a monic polynomial of degree $j$ in $\beta/2$ \cite{FJM00,Fo21a}.

The use of linear response to quantify the change of charge density upon the introduction of a charge $q$ into the OCP, computing from this the change of charge by integrating, and requiring by screening that this must equal $-q$ can be used
\cite{Ja87} to deduce the Carnie-Chan sum rule \cite{CC81,Ca83}
\begin{equation}\label{TT1a}
-\beta \int_{\mathbb R^2}
d \mathbf r \bigg (
\int_{\mathbb R^2}d \mathbf r' \, \log | \mathbf r'| \Big ( \rho_{(2),\infty}^{{\rm OCP} , T}(\mathbf r, \mathbf r') +
{1 \over \pi} \delta(\mathbf r - \mathbf r') \Big ) \bigg ) = 1.
\end{equation}
Here the integral cannot be interchanged, as according to (\ref{PSb}) integrating over $\mathbf r$ first would give zero. The validity of (\ref{TT1a}) can be checked directly for $\beta = 2$ in the bulk (i.e.~for GinUE in the bulk).

\begin{proposition}
The Carnie-Chan sum rule (\ref{TT1a}) is valid for GinUE in the bulk.
    \end{proposition}

    \begin{proof}
    Denote the integral over $\mathbf r'$ in (\ref{TT1a}) by $g(\mathbf r)$. Assuming only rotation and translation invariance of $\rho_{(2),\infty}^{{\rm OCP},T}$ implies
    $$
    g(\mathbf r) = \log |\mathbf r| \Big (
    \int_{|\mathbf r'| <|\mathbf r| }
    \rho_{(2),\infty}^{{\rm OCP},T}(\mathbf r',\mathbf 0) \, d \mathbf r' + {1 \over \pi} \Big ) + \int_{|\mathbf r'| >|\mathbf r| } \log |\mathbf r'|
    \rho_{(2),\infty}^{{\rm OCP},T}(\mathbf r',\mathbf 0) \, d \mathbf r'.
    $$
    Specialising now to $\beta = 2$ by substituting (\ref{2.2t}), the use of polar coordinates and integration by parts allows the integrals to be simplified, with the result
    $$
    g(\mathbf r) = - {1 \over \pi} \int_{|\mathbf r| }^\infty {e^{-(r')^2}\over r'} \, d  r'
    $$
The integration over $\mathbf r \in \mathbb R^2$ of this can be computed by further use of polar coordinates and integration by parts to confirm (\ref{TT1a}).
\end{proof}

Formal use of the convolution theorem in (\ref{TT1a}) shows that it reduces to (\ref{SS1b}) \cite{Ou74}. In the case of the edge geometry, this approach can be used \cite{JLM85} to derive from (\ref{TT1a}) in the edge case the edge dipole moment sum rule
\begin{equation}\label{TT1b}
-2 \pi \beta \int_{-\infty}^\infty
dy  \, \bigg (
\int_{-\infty}^\infty
dy'  \,
(y' - y) 
\int_{-\infty}^\infty dx' \,   \rho_{(2),\infty}^{{\rm e},{\rm OCP} , T}((0,y), (x',y')) 
 \bigg ) = 1,
\end{equation}
first derived in \cite{BHLGM81}.
In keeping with the analogous property of (\ref{TT1a}), it is not possible to interchange the integrations over $y$ and $y'$ in this expression (if it was the LHS would vanish due to the sign change of $y'-y$). In the case $\beta =2$ this is readily verified using the exact result (\ref{2.2t1}).

\begin{proposition}
The edge dipole moment sum rule (\ref{TT1b}) is valid for edge scaled GinUE.
    \end{proposition}

    \begin{proof}
    Substituting (\ref{2.2t1}), integration by parts gives
    $$
\int_{-\infty}^\infty
dy'  \,
(y' - y) 
\int_{-\infty}^\infty dx' \,   \rho_{(2),\infty}^{{\rm e},{\rm GinUE} , T}((0,y), (x',y')) 
  = -{1 \over \sqrt{8 \pi^3}} e^{-2 y^2}    
    $$
(this formula can be found in \cite[Eq.~(2.41)]{Ja82a}).
Now integrating over $y$ verifies (\ref{TT1b}) for $\beta = 2$.
\end{proof}
 
The exact result (\ref{2.2t2}) exhibits a slow decay of $\rho_{(2),\infty}^T((x_1,y_1),(x_2,y_2))$ parallel to the edge. This was first predicted by Jancovici \cite{Ja82b,Ja95}, who on the basis of linear response argument relating to the screening an oscillatory external charge density applied at the edge obtained for general $\beta > 0$,
\begin{equation}\label{TT1c}
\rho_{(2),\infty}^T((x_1,y_1),(x_2,y_2)) \mathop{\sim}\limits_{|x_1 - x_2| \to \infty} {f(y_1,y_2) \over (x_1 - x_2)^2}, \qquad 
\int_{-\infty}^\infty dy_1
\int_{-\infty}^\infty dy_2 \,
f(y_1,y_2) = -{1 \over 2 \beta \pi^2}.
\end{equation}
A derivation of this using the Carnie-Chan sum rule for the OCP confined to a strip geometry is given in \cite{Ja87}. In keeping with this, in \cite{JS01} the amplitude $f(y,y')$ in (\ref{TT1c}) is related to the dipole moment of the screening cloud at the edge by deriving that
\begin{equation}\label{TT1d}
\int_{-\infty}^\infty
dy'  \,
(y' - y) 
\int_{-\infty}^\infty dx' \,   \rho_{(2),\infty}^{{\rm e},{\rm OCP} , T}((0,y), (x',y')) 
  = \pi \int_{-\infty}^\infty dy' \,
f(y,y').
\end{equation}
Then the result for the integral over $y$ on the RHS follows from the dipole moment sum rule (\ref{TT1b}).

We consider now a global scaling, so that the droplet support is the unit disk. For large $N$, define the charge-charge surface correlation by
\begin{equation}\label{TT1e}
\langle \sigma(\theta) \sigma(\theta') \rangle_N^T :=
\int d \mathbf n \int d \mathbf n' \,  
\rho_{(2),N}^{{\rm g},{\rm OCP} , T}(\mathbf r, \mathbf r').
\end{equation}
Here $\rho_{(2),N}^{{\rm g},{\rm OCP} , T}$ refers to $\rho_{(2),N}^{{\rm OCP} , T}$ computed using global scaling (as specified below (\ref{eq:Av1})), then considering the large $N$ form with edge coordinates (these can be taken as $(r,\theta)$ with $r \approx 1$. The integration is over the normal direction $\mathbf n$ (here the radial direction); the more general notation has been used as this quantity can be defined in other geometries e.g.~for the elliptic GinUE \cite{FJ96}. The reasoning leading to (\ref{TT1c}) has been extended \cite{Ja95} to lead to the prediction
\begin{equation}\label{TT1f}
\langle \sigma(\theta) \sigma(\theta') \rangle_\infty^T = - {1 \over 8 \pi^2 \beta \sin^2((\theta - \theta')/2)}.
\end{equation}
This has been checked at $\beta = 2$ using the exact result for GinUE in \cite{CPR87}.

\begin{remark} 
For the bulk scaled OCP with $\beta$ an even integer, there is a constraint on the small $r=|\mathbf r|$ form of $\rho_{(2),\infty}^{{\rm b}}(\mathbf r, \mathbf 0)$. Thus \cite{SP95}
$$
\rho_{(2),\infty}^{{\rm b}}(\mathbf r, \mathbf 0) = r^\beta e^{-\beta r^2/4} f(r^4),
$$
for $f(z)$ analytic. Note from (\ref{2.2t}) that for GinUE, $f(z) = {2 \over \pi^2} {\sinh (\sqrt{z}/2) \over \sqrt{z}}$.
Recently consequences have been found in \cite{Sa19}.
\end{remark}

\section{Normal matrix models}\label{S4.2b}

\subsection{Eigenvalue PDF}
A generalisation of the joint element PDF (\ref{2.25a}) is the functional form proportional to
\begin{equation}\label{2.25a+}
\exp \Big ( - {N \over t_0 } {\rm Tr} (
 J   J^\dagger - 2 {\rm Re} \,  \sum_{p=2}^M t_p {\rm Tr} \, J^p ) \Big ), \quad t_0 > 0;
\end{equation}
note the scaling so that there is  a factor of $N$ in the exponent. Use of the Schur decomposition (\ref{GS})
gives a separation of the eigenvalues and the strictly upper triangular elements in the
exponent analogous to (\ref{2.1b}). This allows the latter to be integrated over, showing that the
corresponding eigenvalue PDF is proportional to
\begin{equation}\label{2.25b+}
\exp \bigg ( - {N \over t_0} \sum_{j=1}^N \Big ( |z_j|^2 - 2 {\rm Re} \sum_{p=2}^M t_p z_j^p \Big ) \bigg ).
\end{equation}
However for $t_M \ne 0$ this is not normalisable for $M > 2$.

A simple remedy is to impose a cutoff by restricting the eigenvalues to a disk centred about the origin
of radius $R$, and restricting to  small $t_0$ and values of $t_2,\dots,t_M$  so that the support $\Omega$ (assumed
simply connected) of the density is
contained inside this disk \cite{EF05}. The mean field argument leading to (\ref{PT}) tells us that in
relation (\ref{2.25b+}), the normalised density has the constant value $1/\pi t_0$ in $\Omega$,
where the latter is such that
\begin{equation}\label{2.25c+}
W(z) = {1 \over \pi} \int_\Omega \log | z - w| \,  d^2w, \qquad W(z) = {1 \over 2}
  \Big ( |z|^2 - 2 {\rm Re} \sum_{p=2}^M t_p z^p \Big ), 
\end{equation} 
 for $z$ contained inside of $\Omega$. From this equation the coupling constants $\{t_p\}_{p=2}^M$
 can be related to moments associated with $\Omega$ \cite{ABWZ02,Za06}.
 
 \begin{proposition}
 In the above setting, for $p \ge 2$ we have
 \begin{equation}\label{2.25d+}
 t_p = {1 \over 2 \pi i p} \int_{\partial \Omega} \bar{z} z^{-p} \, dz = - {1 \over \pi p} \int_{\mathbb C \backslash \Omega} z^{-p} \, d^2 z.
\end{equation} 
\end{proposition} 

\begin{proof}
By applying $\partial_z$ to both sides of (\ref{2.25c+}), then introducing $W(w)$ in the integrand via the
identity $2 \partial_w \partial_{\bar{w}} W(w)=1$ shows
$$
\partial_z W(z) = {1 \over \pi} \int_{\Omega} {\partial_w \partial_{\bar{w}} W(w) \over z - w} \, d^2 w.
$$
Now using the Cauchy-Pompeiu formula (\ref{DB}) with $C_1$ replaced by $\Omega$, we see from this that
$$
\int_{\partial \Omega} {\partial_w W(w) \over w - z} \, dw = 0.
$$
Simple manipulation then shows
$$
{1 \over 2 \pi i} \int_{\partial \Omega} {\bar{w} \over w - z} \, dw =      \sum_{p=2}^M p t_p z^{p-1}.
$$
The first equality of (\ref{2.25d+}) now follows by power series expanding the LHS with respect to $z$.
The second equality can be deduced from the first by applying the version of (\ref{DB}) valid for
$z$ outside $C_1$ (chosen as $\mathbb C \backslash \Omega$), for which the LHS is to be replaced by $0$. 
\end{proof}

The simplest case of (\ref{2.25b+}) beyond that corresponding to  (\ref{2.25a}) is to take $M = 3$.
After scaling, the choice of $t_3$ can be fixed. Choosing $t_3 = {1 \over 3}$, results from \cite{KT15}
give that the support of $\Omega$ is contained in a disk for all $0 < t_0 \le 1/8$, with the critical value
$t_0=1/8$ involving three cusp singularities. Then, the boundary $\partial \Omega$ is
a 3-cusped hypercycloid. For general $M \ge 2$ and $\Omega$ contained in a disk and simply connected,
it is shown in \cite{EF05} that $\partial \Omega$ can be parametrised by a Laurent polynomial of the form
$\alpha_1 w + \alpha_0 + \cdots + \alpha_{M-1} w^{-M+1}$ with $|w|=1$.
Dropping the assumption that $\Omega$ be simply connected requires the theory of quadrature domains; see \cite{LM16} and references therein.

More general than the joint eigenvalue PDF (\ref{2.25a+}) is the functional form
\begin{equation}\label{10.1a}
\exp \Big ( - \sum_{j,k=1}^\infty c_{jk} {\rm Tr} \, (J^j (J^\dagger)^k) \Big ) =: \exp \Big ( - {\rm Tr} \, W(J,J^\dagger) \Big ).
\end{equation}
However unlike the former, substituting the Schur decomposition (\ref{GS}) does not in general lead to a separation
of the eigenvalues from the strictly upper triangular variables. To overcome this, attention can be restricted to the
subset of $N \times N$ complex matrices having the further structure $[J,J^\dagger]=0$, which specifies the matrices
as being normal. For normal matrices, the eigenvectors can be chosen to form an orthonormal set, and so $J = U D U^\dagger$,
for $U$ unitary and $D$ the diagonal matrix of eigenvalues. Moreover, the change of variables from $J$ to $\{U,D\}$ gives
a decomposition of measure into separate eigenvalue and eigenvector factors, with the Jacobian given by
(\ref{2.1cZ}) \cite{CZ98}. Hence the eigenvalue PDF corresponding to (\ref{10.1a}) in this setting is proportional to
\begin{equation}\label{10.1b}
\exp \Big ( - \sum_{j=1}^N W(z_j, \bar{z}_j) \Big ) \prod_{1 \le j < k \le N} | z_k - z_j|^2.
\end{equation}
Note that for the joint element PDF (\ref{2.25a+}), the same eigenvalue PDF (\ref{2.25b+}) results for $J$ specified
on the full space of $N \times N$ complex matrices, as it does on the restriction to normal matrices as implied by
(\ref{10.1b}).

\subsection{Equilibrium measure}
Analogous to (\ref{2.25a+}), to obtain a compact eigenvalue support for large $N$, one considers the case that the weight $e^{-W}$ is exponentially varying in a sense that $W$ is of order $N$, say $W(z,\bar{z})=NQ(z)$ for a fixed potential $Q: \C \to \R$. 
Thus the eigenvalue PDF \eqref{10.1b} is written as
\begin{equation}\label{10.1b2}
e^{-H(z_1,\dots,z_N)}, \quad H(z_1,\dots,z_N):= \sum_{1 \le j < l \le N} \log \frac{1}{|z_j-z_l|^2}+N \sum_{j=1}^N Q(z_j). 
\end{equation}
As in Remark~\ref{R2.4}, the macroscopic behaviour of the system \eqref{10.1b2} can be described using the two-dimensional Coulomb gas interpretation with the help of logarithmic potential theory.  
Namely, it is well known in the literature \cite{Jo98,HM13,CGZ14,CHM18,Am21} that for a general $Q$ under suitable potential theoretic assumptions, the empirical measure $\frac{1}{N} \sum \delta_{z_j}$ weakly converges to a unique probability measure $\mu_Q$ that minimises 
\begin{equation} \label{IQ[mu]}
I_Q[\mu]:= \int_{ \mathbb{C}^2 } \log \frac{1}{ |z-w| }\, d\mu(z)\, d\mu(w) +\int_{ \mathbb{C} } Q \,d\mu;
\end{equation}
cf.~(\ref{1.1i}).
One may notice that \eqref{IQ[mu]} can be interpreted as a continuum limit of the Hamiltonian $H_N$ in \eqref{10.1b2} after normalisation, generalising (\ref{PTm}).
The probability measure $\mu_Q$ is called the equilibrium measure and its support $S_Q:=\textup{supp}(\mu_Q)$ is called the droplet, as previously remarked.
Furthermore, if $Q$ is $C^2$-smooth in a neighbourhood of $S_Q$, by Frostman's theorem \cite{ST97}, $\mu_Q$ is absolutely continuous with respect to the Lebesgue measure and takes the form 
\begin{equation} \label{Frostman thm}
d\mu_Q(z)= \frac{ \partial_z \partial_{ \bar{z} } Q(z)}{\pi} \chi_{ z\in S_Q }  \,d^2 z. 
\end{equation}
In particular, for a rotationally symmetric potential $q(r)=Q(|z|=r)$, the droplet is of the form $S_Q=\{ r_1 \le |z| \le r_2 \}$, where $(r_1,r_2)$ are the unique pair of constants satisfying 
\begin{equation} \label{droplet annulus}
r_1 q'(r_1)=0, \quad r_2 q'(r_2)=2, 
\end{equation}
see  \cite[\S~IV.6]{ST97}. 
Here, we have assumed that $Q(z)$ is strictly subharmonic in $\C$, which is equivalent to the requirement that $r \mapsto rq'(r)$ is increasing in $(0,\infty).$ 
Let us mention that all the explicit macroscopic densities with rotation invariance in the above subsections of \S 2 can be obtained as special cases of the formulas \eqref{Frostman thm} and \eqref{droplet annulus}. 
For instance, the RHS of \eqref{2.1gm} can be realised as the RHS of \eqref{Frostman thm} with $Q(z)=|z|^2-2\alpha \log |z|$.
Beyond the case when $Q$ is radially symmetric, the determination of the droplet, also known as the two-dimensional equilibrium problem (e.g. the ellipse in \S~\ref{S2.3}), is far from being obvious even for some explicit potentials with simple form; see \cite{BBLM15,BM15,KL15,ABK21,CK22a} and references therein for recent works in this direction.

\subsection{Partition functions} \label{S5.3}
Continuing the discussion in \S \ref{S4.1}, we consider the global scaled, non charge neutral, partition function 
\begin{equation}
Z_{N}(\beta;Q)= \frac{1}{N!} \int_{ \C } d^2z_1 \cdots \int_{ \C } d^2 z_N\,   e^{-\frac{\beta}{2}H(z_1,\dots,z_N)},  
\end{equation}
where $H(z_1,\dots,z_N)$ is the Hamiltonian given in \eqref{10.1b2}.
An explicit formula for the large $N$ expansion of $Z_N(\beta;Q)$ was predicted in \cite{ZW06}. 
Fairly recently, it was shown by Lebl\'{e} and Serfaty \cite{LS17} that for general $\beta>0$ and $Q$, $Z_{N}(\beta;Q)$ admits the large $N$ asymptotic expansion of the form 
  \begin{equation} \label{LS formula}
  \log Z_N(\beta;Q) = -\frac{\beta}{2}N^2 I_Q[\mu_Q] + \Big(\frac{\beta}{4}-1\Big)N \log N - \bigg( C(\beta) + \Big(1-\frac{\beta}{4}\Big) E_Q[\mu_Q] \bigg) N+ o(N);
\end{equation} 
see also an earlier work \cite{SS15} on \eqref{LS formula} up to the $O(N \log N)$ term.
The term $I_Q[\mu_Q]$ appearing in the leading order asymptotic is the energy \eqref{IQ[mu]} evaluated at the equilibrium measure $\mu_Q$. In the case of the OCP, up to a sign it is the quantity appearing in the exponent of $A_{N,\beta}$ in (\ref{QN1}) at order $N^2$.
To see this, we note that for a radially symmetric potential $q(r)=Q(|z|=r)$ generally, it is evaluated as 
\begin{equation} \label{energy radially sym}
I_Q[\mu_Q]=q(r_1)-\log r_1 -\frac14 \int_{r_0}^{r_1} rq'(r)^2\,dr, 
\end{equation}
where $r_1$ and $r_2$ are the constants specified in \eqref{droplet annulus}. 
For the OCP $Q(z)=|z|^2$, which gives $I_Q[\mu_Q]=3/4$, this indeed being the coefficient of $N^2$ seen in $A_{N,\beta}$.
The appearance of the term  $(\beta/4 - 1)N\log N$ in (\ref{LS formula}), whereas there is a term ${1 \over 4} \beta N^2 \log N$ in $A_{N,\beta}$ of (\ref{QN1}) is due to the simple scaling $z_j \mapsto  z_j/ \sqrt{N}$ required to go from the OCP with global scaled coordinates (as assumed in (\ref{LS formula})) to the OCP itself (as assumed in (\ref{QN1})).
The terms appearing in the $O(N)$ term of \eqref{LS formula} are the entropy
\begin{equation}\label{entropy} 
  E_Q[\mu_Q] := \int_{\C} \mu_Q(z)\,\log \mu_Q(z) \, d^2z
\end{equation} 
associated with $\mu_Q$ and a constant $C(\beta)$ independent of the potential $Q$.  Their sum is the free energy per particle, $\beta f(\beta;Q)$.
The expansion \eqref{LS formula} with quantitative error bounds is also available in the literature \cite{BBNY19,Se20}.

For the random normal matrix model when $\beta=2$, the determinantal structure \eqref{2.1f} allows an explicit expression 
\begin{equation}
Z_N(2;Q)=  \prod_{j=0}^{N-1} h_j, 
\end{equation}
where $h_j$ is the orthogonal norm in \eqref{KN RNM}; cf. Proposition~\ref{P1.1}.
In particular, if $Q$ is rotationally symmetric, since $p_j(z)=z^j$, we have
\begin{equation}
h_j=2\pi \int_0^\infty r^{2j+1} e^{-N q(r)}\,dr;
\end{equation}
cf.~(\ref{2.85a}).
Based on this knowledge together with a Laplace approximation of the integrals, the precise asymptotic expansion of the free energy up to the $O(1)$ term was derived in a recent work \cite{BKS22}. 
This in particular shows that  
\begin{equation}\label{ZN2Q exp}
\log Z_N(2;Q)  =-N^2 I_Q[\mu_Q] - \frac{1}{2}N\log N + \Big( \frac{\log(2\pi^2)}{2}- \frac12 E_Q[\mu_Q] \Big) \, N - \frac{\chi}{12}\log N +O(1),
\end{equation}
where $\chi$ is the Euler index; cf. \eqref{QN4}. Let us recall here that $\chi=0$ for the annulus ($r_0>0$) and $\chi=1$ for the disk ($r_0=0$) geometry.
See \cite[\S4]{BKS22} and \cite[\S4]{FF11} for some concrete examples of \eqref{ZN2Q exp} associated with the matrix models discussed in \S 2.

\subsection{Correlation functions and universality} \label{S5.4}
Turning to the correlation functions $\rho_{(k),N}$ of \eqref{10.1b2}, the determinantal structure \eqref{2.1f} remains valid with the correlation kernel 
\begin{equation} \label{KN RNM}
K_N(z,w) =e^{-\frac{N}{2} (Q(z) + Q(w))  } \sum_{j=0}^{N-1} \frac{ p_{j}(z)\overline{p_{j}(w)} }{ h_j },
\end{equation}
where $p_j$ is the monic orthogonal polynomial of degree $j$ with respect to the weighted Lebesgue measure $e^{-N Q(z)}\,d^2z$ and $h_j$ is its squared orthogonal norm; cf.~(\ref{4.1d+}).

Let us first discuss the asymptotic behaviours of $\rho_{(k),N}$ in the micro-scale. 
Given a base point $p \in S_Q$ for which we zoom the point process, we denote by
\begin{equation}
\delta:= \frac{ \partial_z \partial_{ \bar{z} } Q(z)}{\pi}\Big|_{z=p}  
\end{equation}
the mean eigenvalue density at $p$; cf.\eqref{Frostman thm}. 
From universality principles (see e.g.~\cite{Ku11}), one expects that for a general $Q$ and $p$ such that $\delta \in (0,\infty)$ (i.e. the eigenvalue density does not vanish or diverge at $p$), the universal scaling limit in Proposition~\ref{P2.3} arises. 
For the bulk case when $p \in \textup{Int}(S_Q)$, such a universality was established by Ameur, Hedenmalm and Makarov in \cite{AHM11} where they showed that for a fairly general potential $Q$ under some mild assumptions,  
\begin{align} \label{RNM bulk limit}
\frac{1}{(N \delta)^k}\rho_{(k),N}\Big( p+\frac{z_1}{ \sqrt{ N \delta } }, \dots, p+\frac{z_k}{ \sqrt{ N \delta } } \Big) \to  \det \Big [ K_\infty^{\rm b}(z_j, z_l) \Big ]_{j,l=1}^k,
\end{align}
uniformly on compact subsets of $\C$ as $N \to \infty$. See also \cite{Be09,Be13}.
In a sequential work \cite{AHM15} the theory of loop equations (or Ward's identities) was also used to show the bulk scaling limit \eqref{RNM bulk limit}. 
It says that for a given test function $\psi$, 
\begin{equation} \label{Ward} 
\mathbb{E}_N W_N^+[ \psi ]=0, \quad W_N^+[\psi]:=\frac12 \sum_{ j \not =k } \frac{ \psi(z_j)-\psi(z_k) }{ z_j-z_k } - N \sum_{ j=1 }^N [ \partial Q \cdot \psi ] (z_j) + \sum_{j=1}^N \partial \psi(z_j).  
\end{equation}
The functional $W_N^+$ is also called the stress energy tensor in the context of conformal field theory \cite[Appendix 6]{KM13}.
The identity \eqref{Ward} easily follows from the integration by parts
\begin{equation}
\mathbb{E}_N[\partial \psi (z_j)] = \mathbb{E}_N [ \partial_{z_j} H(z_1,\dots,z_N) \cdot \psi(z_j) ]. 
\end{equation}
The approach using Ward's identities was further developed in \cite{AKM20} and several related works to study various scaling limits of the random normal matrix models (see e.g. \cite{AKMW20,AB21} for the bulk scaling limit at weak non-Hermiticity \eqref{Wk}; \cite{AKS20,Seo22} for the edge scaling limit with boundary confinements \eqref{CV8}; \cite{AK13,AKS21} for normal matrices with Mittag-Leffler type singularities as in \S \ref{S.GI} and \S \ref{S2.7}) and we refer to \cite[Remark 2.9]{ABK22} for an expository summary of this strategy.
In particular, in \cite{AKM20}, the rescaled version of Ward's identity was introduced and used to show the edge universality: for $p \in \partial S_Q,$
\begin{align} \label{RNM edge limit}
\frac{1}{(N \delta)^k}\rho_{(k),N}\Big( p+i\mathbf n\frac{z_1}{ \sqrt{ N \delta } }, \dots, p+i\mathbf n\frac{z_k}{ \sqrt{ N \delta } } \Big) \to  \det \Big [ K_\infty^{\rm e}(z_j, z_l) \Big ]_{j,l=1}^k,
\end{align} 
where $\mathbf n$ is the outer normal vector at $p$ as in Proposition~\ref{P2.7}.
However, the result in \cite{AKM20} has an additional assumption that the limiting correlation function is translation invariant along the real axis. 
This assumption is intuitively natural but hard to rigorously show in general (except e.g. for the rotationally symmetric potential $Q$).
The edge universality was later then shown by Hedenmalm and Wennman for a wide class of potentials $Q$ in \cite{HW21}, where they developed an asymptotic theory for general planar orthogonal polynomials; see also \cite{He21} for an alternative approach to derive the main result in \cite{HW21} and \cite{HW20} for a theory developed for the orthogonal polynomials associated with non-exponentially varying weight. 
The asymptotic results in \cite{HW21} together with \eqref{KN RNM} then leads to \eqref{RNM edge limit} by the Riemann sum approximation.

We now turn our attention to the asymptotic behaviours of $\rho_{(k),N}$ in the macro-scale. 
Let us begin with the $1$-point function $\rho_{(1),N}$. 
For the bulk case when $z$ is interior of the droplet, the asymptotic behaviour of $\rho_{(1),N}$ is well known in the literature, see \cite{Be09,Am13} and references therein. 
It says that under the suitable assumptions on $Q$, there are real-analytic functions $B_j$ such that 
\begin{equation} \label{rho 1 expansion}
\pi \rho_{(1),N}(z)= N\Delta Q(z) +\frac12 \Delta \log \Delta Q(z) + N^{-1} B_2(z) + \cdots + N^{ -k+1 } B_k(z) + O(N^{-k}), 
\end{equation}
where $\Delta := \partial_z \partial_{ \bar{z} } $ is the one quarter of the usual Laplacian.  
Notice here that the leading order asymptotic of \eqref{rho 1 expansion} is implied by the Laplacian growth \eqref{Frostman thm}. 
As a concrete example, we consider the induced spherical ensemble \eqref{PDF Ispherical} with $M=\alpha_1 N$ and $n=\alpha_1 N-1$, which can be realised as \eqref{10.1b2} with 
$$Q(z)=(\alpha_1+\alpha_2-1) \log(1+|z|^2)-2(\alpha_1-1)\log|z|.$$ 
On the other hand, it follows from \eqref{K Ispherical} that with $\zeta=|z|^2/(1+|z|^2)$,
\begin{equation}
\pi\rho_{(1),N}(z)= \frac{ |z|^{ 2(M-N)  } }{   (1+|z|^2)^{ 2  }  } (M+n-N) \Big( I_\zeta (M-N,n)-I_\zeta (M,n-N) \Big). 
\end{equation}
Then using the well-known asymptotic behaviours of the incomplete beta function, one can observe that for $r_1<|z|<r_2$, where $r_1,r_2$ are given in \eqref{r1 r2 Ispherical}, the asymptotic behaviour \eqref{rho 1 expansion} holds with
$$
\Delta Q(z)= \frac{\alpha_1+\alpha_2-1}{ (1+|z|^2)^2 }, \qquad \Delta \log \Delta Q(z)= -\frac{2}{ (1+|z|^2)^2 }.
$$

Next, we consider the off-diagonal asymptotic behaviour of the correlation kernel. 
For the Ginibre ensemble, equivalently, for the random normal matrix model \eqref{10.1b2} with $Q(z)=|z|^2$, the associated correlation kernel $K_N$ in \eqref{KN RNM} satisfies  
\begin{equation} \label{KN Ginibre Szego}
 K_{N}(z,w) = \sqrt{ \frac{2N}{\pi} }    (z\bar{w})^N e^{N -\frac{N}{2}(|z|^2+|w|^2) } S(z,w)  \Big( 1+O(\frac{1}{N}) \Big), \qquad (z \not =w)
\end{equation}
where $S(z,w)$ is the exterior Szeg\"o kernel 
\begin{equation}
S(z,w):=\frac{1}{2\pi}\frac{1}{z\bar{w}-1}.
\end{equation}
From a viewpoint of Proposition~\ref{P2.2}, the asymptotic behaviour \eqref{KN Ginibre Szego} can be realised as a uniform expansion of the incomplete gamma function $z \mapsto \Gamma(N;N z)$, which is available in the literature in some particular domains; see e.g. \cite{Tr50} for $|\arg (z-1)|<3\pi/4.$ 
In a recent work \cite{AC22}, generalising the classical results, it was shown that the asymptotic behaviour \eqref{KN Ginibre Szego} remains valid as long as $z \bar{w}$ is outside the Szeg\"o curve $
\{ z \in \C: |z| \le 1, |z \, e^{1-z}|=1 \}.$
Note in particular that if $|z|=|w|=1,$ then \eqref{KN Ginibre Szego} reads
\begin{equation} \label{KN Szego Ginibre bdy}
 K_{N}(z,w) \overset{c}{\sim} \sqrt{\frac{2N}{\pi}}  S(z,w)  \Big( 1+O(\frac{1}{N}) \Big),
\end{equation}
where $\overset{c}{\sim}$ means that the asymptotic expansion holds up to a sequence of cocycles (in this case $(z\bar{w})^N$), which cancel out when forming a determinant \eqref{2.1f}.
From a statistical physics point of view, the behaviour \eqref{KN Szego Ginibre bdy} indicates that there are strong correlations among the particles on the boundary of the droplet, which also shows the slow decay of correlations at the boundary. For GinUE this is explicit in (\ref{TT1f}).
For elliptic GinUE, the phenomenon was studied in \cite{FJ96}, as a test of the generalisation to more general shaped droplets, when the RHS of (\ref{TT1f}) is predicted to be given in terms of a certain Green's function for an electrostatics problem outside of the droplet, which acts as a macroscopic conductor \cite{Ja95}, \cite[Eq.~(3.29]{Fo98a}.
Furthermore, it was obtained by Ameur and Cronvall \cite{AC22} that for a general class of potentials $Q$, the associated correlation kernel $K_N$ satisfies 
\begin{equation} \label{KN Szego bdy}
 K_{N}(z,w) \sim \sqrt{2N} \Big(  \frac{ \partial_z \partial_{ \bar{z} } Q(z)}{\pi} \Big)^{1/4}  \Big(  \frac{ \partial_w \partial_{ \bar{w} } Q(w)}{\pi} \Big)^{1/4}   S(z,w)  ( 1+o(1) ).
\end{equation}
For this, the use of a general theory on the orthogonal polynomial due to Hedenmalm and Wennman \cite{HW21} was made. 
We also refer to \cite{ADM22,BY22,ACC22} for more recent studies in this direction. 

\begin{remark} 
Let $\{p_j^{(N,R)}(z) \}_{j=0,1,\dots}$ be the orthogonal polynomials with respect to the inner product
$\langle f | g \rangle := \int_{C_R} f(z) g(\bar{z}) e^{-2N W(z)/t_0}$, where $W(z)$ as in (\ref{2.25c+}) and
it is assumed that the limiting eigenvalue support $\Omega$ corresponding to (\ref{2.25b+}) is contained in $D_R$. 
Consider the probability distribution $\mathcal P_k^{(N,R)}$ specified by the eigenvalue PDF proportional to
\begin{equation}\label{WTa}
\prod_{l=1}^k e^{-2N W(z_l)/t_0} \prod_{1 \le j < l \le k} | z_l - z_j |^2.
\end{equation}
As a result of the underlying determinantal structure, one has the formula for $p_k^{(N,R)}(z)$ as an expectation
with respect to $\mathcal P_k^{(N,R)}$ (see e.g.~\cite[proof of Proposition 5.1.3]{Fo10}; in fact such formulae were known to Heine \cite{Sz75})
$$
p_k^{(N,R)}(z) = \Big \langle \prod_{l=1}^k ( z - z_l) \Big \rangle_{\mathcal P_k^{(N,R)}}.
$$
For $k/N=x$ as $k,N \to \infty$ it has been conjectured \cite{El07} that the zeros of $p_k^{(N,R)}(z)$ accumulate on certain arcs $\Sigma$
contained in $\Omega$, with corresponding measure $\mu_x^*$. Assuming this, it follows that for $z \in \mathbb C \backslash \Omega$
\begin{equation}\label{10.1c}
{1 \over \pi t_0} \int_{\Omega} \log | z - w| \, d^2 w = \int_{\Sigma} \log | z - s| \, d \mu_x^*.
\end{equation}
This identifies $\Sigma$ as the so-called mother body or potential theoretic skeleton of $\Omega$. A number of works give further developments
along these lines, especially in relation to the strong asymptotics of the planar orthogonal polynomials based on Riemann-Hilbert analysis. 
In the case of (\ref{10.1b}) with a cubic potential, references include \cite{BK12,KT15,MS19,BS20}, while for the induced Gaussian type weight $2W(z)/t_0=|z|^2-2c \log |z-a|$ or more generally $2W(z)/t_0=|z|^2-2\sum_{j=1}^M c_j \log |z-a_j|$ in (\ref{WTa}), see
\cite[$c=O(1)$]{BBLM15}, 
\cite[$c=O(1/N)$ and $|a| \not = 1$]{BGM17}, \cite[$c=O(1/N)$ and $|a| \approx 1$]{BEG18},
\cite[$c=O(1/N)$ and $|a| \not = 1$]{LY17},
\cite[$c_j=O(1/N)$]{LY20}. 
We also refer to \cite{WW19,DS22,La20} for applications of such strong asymptotics in the context of the characteristic polynomials of the Ginibre matrix. 
\end{remark}

\section{Further theory and applications}

\subsection{Fermi gas wave function interpretation}\label{S4.2a}
We have seen that the rewrite of  (\ref{1.1f}), written in the exponential 
 form (\ref{1.1i}), allows for the GinUE eigenvalue PDF to be interpreted as the Boltzmann
 factor for a particular Coulomb gas. If instead of an exponential form we use
 (\ref{2.1h}) to rewrite  (\ref{1.1f}) as
\begin{equation}\label{4.5a}
\Big | \prod_{j=1}^N e^{- | z_j|^2/2} \det [ z_j^{k-1} ]_{j,k=1,\dots,N} \Big |^2,
\end{equation} 
then we are lead to an interpretation as the absolute value squared of a ground state free Fermi
quantum many body wave function. Thus inside the absolute value of (\ref{4.5a}) is a Slater
determinant of single body wave functions $\{\phi_l(z) \}_{l=0,\dots,N-1}$ with $\phi_l(z) =
e^{- | z|^2/2} z^l$. What remains then is to identify the corresponding one body Hamiltonian
for quantum particles in the plane which have these single body wave functions for the lowest
energy states.

The appropriate setting for this task is a quantum particle confined to the $xy$-plane subject
to a perpendicular magnetic field, $ (0,0,B)$, $B > 0$. Fundamental to this setting is the
vector potential $\mathbf A$, related to the magnetic field by $\nabla \times \mathbf A =  (0,0,B)$.
The so-called symmetric gauge corresponds to the particular choice $\mathbf A = (-By,Bx,0)=:(A_x,A_y,0)$, which
henceforth will be assumed. Physical quantities in this setting are $m$ (the particle mass),
$e$ (particle charge), $\hbar$ (Planck's constant), $c$ (speed of light), which together with $B$ are
combined to give $\omega_c := e B/mc$ (cyclotron frequency) and $\ell := \sqrt{\hbar c/eB}$
(magnetic length).

Defining the generalised momenta and corresponding raising and lowering operators by
$$
\Pi_u = - i \hbar {\partial \over \partial u} + {e \over c} A_u \: \: (u = x,y), \qquad a^\dagger = {\ell \over \sqrt{2} \hbar} (\Pi_x + i \Pi_y), \qquad
a = (a^\dagger)^\dagger,
$$
allows the quantum Hamiltonian to be written in the harmonic oscillator like form
$H_B = \hbar \omega_c(a^\dagger a + {1 \over 2})$  \cite{CDL77}. Important too are the quantum centre of orbit operators
and associated raising and lowering operators
$$
U = u - {\ell^2 \over \hbar} \Pi_u \: \: (U = X,Y; \, u=x,y),  \qquad b^\dagger = {1 \over \sqrt{2} \ell} (X  - i Y), \qquad
b = (b^\dagger)^\dagger,
$$
for which $X^2 + Y^2 = 2 \ell^2 (b^\dagger b + {1 \over 2})$. The operators $\{a, a^\dagger\}$ commute with $\{b, b^\dagger \}$,
implying that $H$ and $X^2 + Y^2 $ permit simultaneous eigenstates. A complete orthogonal set can be constructed using the
raising operators according to 
\begin{equation}\label{bns}
|n,m \rangle = {(a^\dagger)^n (b^\dagger)^m \over \sqrt{n! m!}} |0,0\rangle,
\end{equation}
with eigenvalues of $H$ equal to $(n + {1 \over 2}) \hbar \omega_c$ and eigenvalue of $X^2 + Y^2$ equal to $(2m+1) \ell^2$.
The ground state $|0,0\rangle$ is characterised by $a |0,0\rangle = b|0,0\rangle = 0$, which can be checked to have the unique solution
$|0,0\rangle \propto e^{-(x^2+y^2)/4 \ell^2}$. From this, application of $(b^\dagger)^m$ gives $|0,m\rangle \propto
\bar{z}^m e^{-|z|^2/4 \ell^2}$, $z = x + i y$. Forming a Slater determinant with respect to the first $N$ eigenstates of this type
gives (\ref{4.5a}) with $\ell^2 = 1/2$. Generally states with quantum number $n=0$ and thus belonging to the
ground state are said to be in the lowest Landau level. One remarks that the largest eigenvalue of $X^2 + Y^2$ is then $N - 1/2$, which is
in keeping with the squared radius of the leading order support in the circular law.

The above theory of
a quantum particle in the plane subject to a perpendicular magnetic field can be recast to apply to a rotating quantum particle
in the plane \cite{HC00, LMS19}. It is further true that the elliptic GinUE PDF (\ref{2.3}) admits an interpretation as the absolute value
squared of state in the lowest Landau level, and furthermore the corresponding orthogonal polynomials (\ref{cH}) can be constructed
using a Bogolyubov transformation of $\{b, b^\dagger \}$ \cite{FJ96}.
Also, the PDF on the sphere (\ref{CV2}) permits an interpretation as the absolute value
squared of the ground state wave function for a free Fermi gas on the sphere subject to a perpendicular magnetic field
\cite{Ha83x}. Another point of interest relates to the $N$-body Fermi ground state corresponding to the quantum Pauli Hamiltonian
in the plane with a perpendicular inhomogeneous magnetic field $B(x,y)$. The Hamiltonian $H_B$ defined above then is to be
multiplied by the $2 \times 2$ identity matrix, and the spin coupling term $-(g\hbar/2m) B(x,y) {\rm diag}\,(1/2,-1/2)$ added. With
$B(x,y) = - {1 \over 2} \nabla^2 W(x,y)$ for some real valued $W$, and with the assumption $\Phi := \int B(x,y) \, dx dy < \infty$,
the ground state for this model  (which is spin polarised all spins up) permits an exact solution for $g=2$ \cite{AC79}. The ground state
of normalisable eigenfunctions has degeneracy $ [\Phi/2\pi \hbar] =: N$, with basis of eigenfunctions $\{ z^j e^{W(x,y)/2 \hbar} \}_{j=0}^{N-1}$.
This implies the Fermi many body ground state (\ref{4.5a}) with $e^{-|z_j|^2/2}$ replaced by $e^{W(x_j,y_j)/2 \hbar}$ \cite{ABWZ02}.

\begin{remark} $ $ \\
1.~Upon stereographic projection of the sphere to the plane, it is possible to write the quantum Hamiltonian for a charge particle in a constant perpendicular magnetic field in a form unified with the original planar case \cite{Du92a}. This involves the K\"ahler metric and potential, and permits a viewpoint which carries over to further generalise the space to higher dimensional complex manifolds in $\mathbb C^m$. A point of interest is that doing so gives, for the bulk scaling limit of the corresponding $N$ particle lowest Landau level state, the natural higher dimensional analogue of the kernel (\ref{2.2c}) \cite{Be13,Be14}. \\
2.~The squared wavefunction for higher Landau levels (say the $r$-th) has been shown to give rise to the determinantal point process with bulk scaled kernel 
$$
K_\infty^{r}(w,z) = L_r^0(|w-z|^2) e^{w \bar{z}} e^{-(|w|^2+|z|^2)/2};
$$
see e.g.~\cite[Prop.~2.5]{Sh15}. Allowing for mixing between Landau levels up to and including level $r$ leads to squared wave functions giving rise to the same determinantal point process except for the replacement of Laguerre polynomials $L_r^0 \mapsto L_r^1$  in the kernel \cite{HH13}. 
Extending \cite{LMS19}, the precise mapping between the rotating fermions in the higher Landau levels and the polyanalytic Ginibre ensemble was established in \cite{KMG21}. 
Furthermore, its full counting statistics and generalisations to finite temperature were obtained in \cite{SLMS22,KLMG22}. 
\\
3.~In the theory of the fractional quantum Hall effect, constructing an anti-symmetric state with filling fraction of the lowest Landau level $\nu = 1/m$,
for $m$ an odd integer,  plays a crucial role. To accomplish this, Laughlin
 \cite{La83} proposed the ground state wave function proportional to
 \begin{equation}\label{bns+}
 \prod_{l=1}^N e^{-|z_l|^2/4 \ell^2} \prod_{1 \le j < k \le N} (\bar{z}_k - \bar{z}_j)^m;
 \end{equation}
 note that with the assumption that $m$ is odd, this is anti-symmetric as required for fermions. Moreover it belongs to the lowest Landau level as follows from the theory in the text below (\ref{bns}). The absolute value squared of (\ref{bns+}) coincides with the Boltzmann factor (\ref{1.1i}) with $\beta = 2m$, and the scaling $z_l \mapsto z_l/\sqrt{2m\ell^2}$. From potential theoretic/ Coulomb gas reasoning, the bulk density is therefore
 $1/(2m\pi \ell^2) $. The factor of $m$ in the denominator is in precise agreement with the requirement that the filling fraction be equal to $1/m$. \\
 4.~The ground state $N$-body free spinless Fermi gas in the plane, without a magnetic field but confined by a radial harmonic potential, is also an example of a determinantal point process for which exact calculations are possible; see the recent review \cite{DDMS19}. However, its statistical state is distinct from that of GinUE. Thus with a global scaling so that the support is the unit disk, the density profile as the $d=2$ Thomas-Fermi functional form ${2 \over \pi} (1 - |z|^2) \chi_{|z|<1}$, in contrast to the circular law (\ref{2.2a}). The bulk scaled two-point correlation function (bulk density $1/4\pi$) is given in terms of the $J_1$ Bessel function
$$
\rho_{(2),\infty}^{\rm hF}(z_1,z_2) = \Big ( {1 \over  4 \pi} \Big )^2 \bigg ( 1 - \Big ( {2J_1(|z_1 - z_2|) \over|z_1 - z_2|}\Big )^2 \bigg ),
$$
in contrast to (\ref{2.2t}). This gives a decay proportional to $1/|z_1 - z_2|^3$ of $\rho_{(2),\infty}^{{\rm hF}, T}$.
Also, the edge scaled correlation kernel now involves Airy functions \cite{DDMS16}, rather than the error function seen in (\ref{2.2c}). Universality results relating to many body free Fermi ground states in dimension $d \ge 2$ have recently been obtained \cite{DL21}. We highlight in particular the macroscopic fluctuation theorem for the linear statistic $G = \sum_j g(\mathbf r_j/R)$, with $g$ assumed sufficiently smooth and absolutely integrable, in the $R \to \infty$ limit
\cite[Th.~III.2]{DL21}
\begin{equation}
  {G - {R^d \omega_d \over (2 \pi)^d} \int_{\mathbb R^d} g(\mathbf r) \, d^d \mathbf r \over \sigma_d R^{(d-1)/2}} \to 
  {\rm N}[0,\Sigma(g)], \quad
(\Sigma(g))^2 = \int_{\mathbb R^d} | \hat{g}(\mathbf r)|^2 |\mathbf r| \, d^d \mathbf r.
\end{equation}
Here $\omega_d = \pi^{d/2}/\Gamma(1+d/2)$ is the volume of the Euclidean ball in $\mathbb R^d$, $\omega_d/(2 \pi)^d$ is the bulk density, $\sigma_d^2 := \omega_{d-1}/(2 \pi)^d$ and the Fourier transform has the definition $\hat{g}(\xi) = {1 \over (2 \pi)^{d/2}} \int_{\mathbb R^d} e^{ -i \xi \cdot \mathbf r } g(\mathbf r) \, d^d \mathbf r$. 
Note in particular that in contrast to (\ref{J2b+}), the variance of $G$ now diverges with the scale $R$.

\end{remark}

\subsection{Quantum chaos applications}
The pioneering works of Wigner and Dyson relating to the Hermitian random matrix ensembles was, as noted in  \S \ref{S1}, motivated by seeking a model for the (highly excited) energy levels of a complex quantum system. Later, in the 1980's, as a fundamental contribution to the then emerging subject of quantum chaos, Bohigas et al.~\cite{BGS84} identified the correct meaning of a complex quantum system not by the number of particles but rather as one for which the underlying classical mechanics is chaotic. To test this prediction on say the numerically generated spectrum of a quantum billiard system, the energy levels (beyond some threshold to qualify as being highly excited) were first unfolded so that their local density became unity, and then their numerically determined statistical properties were compared against random matrix predictions for the appropriate symmetry class; see e.g.~\cite{Ha00}. Most popular among the statistical properties have been the variance for the number of eigenvalues in a large interval, and the distribution of the spacing between successive eigenvalues.

A natural extension of these advances is to inquire about the spectrum of a dissipative chaotic quantum system, which due to the loss of energy need not be real. This question was taken up by Grobe, Haake and Sommers \cite{GHS88} for the specific model of a damped periodic kicked top. The quantum dynamics are specified by a subunitary density operator. It is the spectrum of this operator, which after unfolding, and considering only those eigenvalues in the upper half plane away from the real axis (there is a symmetry which requires that the eigenvalues come in complex conjugate pairs --- see the recent paper \cite{AKMP19} for a discussion of this point in a random matrix context) that were compared in \cite{GHS88} a statistical sense to GinUE. Following from precedents in the Hermitian case, in the statistical quantity measured was the  distribution of the radial spacing between closest eigenvalues, to be denoted $P^{\rm s,GUE}(r)$ with the normalisation $\int_0^\infty
P^{\rm s,GUE}(r) \, dr = 1$. This is the quantity $F_\infty(0;D_r)$ of Remark \ref{R3.3}.1. Recalling  (\ref{rK2+}) we therefore have
\begin{equation}\label{Np}
P^{\rm s,GUE}(r) = -  
{d \over d r} e^{r^2} \prod_{j=1}^\infty \Big ( 1 - {\gamma(j;r^2) \over \Gamma(j)} \Big ).
\end{equation}
It follows that for small $r$, $P^{\rm s,GUE}(r) \sim 2 r^3 $, while it follows from (\ref{rK2p}) and the comment in the sentence immediately above Remark \ref{R3.3} that for large $r$, $\log P^{\rm s,GUE}(r) \sim   -r^4/4$. A numerical plot can be obtained from the functional form (\ref{Np}). For the moments the formula $\langle r^p \rangle = p \int_0^\infty r^{p-1} e^{r^2} E_\infty(0;r) \, dr$ holds true. 
In particular, for the mean we calculate $\langle r \rangle = 1.142929\dots$.

A variation of the closest neighbour spacing for an eigenvalue at $z$ is the complex ratio $(z^{\rm c} - z)/(z^{\rm nc} - z)$, where $z^{\rm c}$ is the closest neighbour to $z$, and $z^{\rm nc}$ is the next closest neighbour \cite{SRP20}. An approximation, with fast convergence properties to the large $N$ form, has been given recently in \cite{DW22}. 

Very recently a non-Hermitian Hamiltonian realisation of GinUE has been obtained in the context of a proposed non-Hermitian $q$-body Sachev-Ye-Kitaev (SYK) model, with $N$ Majorana fermions --- $N$ large and tuned ${\rm mod} \, 8$ --- and $q>2$ and tuned ${\rm mod} \, 4$ \cite{GSV22}. On another front, again very recently, the emergence of GinUE behaviours in certain model many body quantum chaotic systems in the space direction has been demonstrated \cite{SDH22}. Of interest in both these lines of study is the so-called dissipative (connected) spectral form factor
$$
{\rm K}_N^{\rm c}(t,s) = {1 \over N} {\rm Cov} \, \Big (
\sum_{j=1}^N e^{i (x_j t + y_j s)}, \sum_{j=1}^N e^{-i (x_j t + y_j s)} \Big ).
$$
Making use of the first formula in (\ref{4.1a}) and the finite $N$ form of (\ref{2.2s}), this can be evaluated in terms of the hypergeometric function ${}_1F_1$ \cite[Eq.~(3)]{LPC21} (corrected in \cite[Appendix A]{GSV22a}; note too that both those references use global scaled variables, whereas we do not).

\begin{proposition}
We have
\begin{multline}
{\rm K}_N^{\rm c}(t,s)  = 1 \\ - {1 \over N} \sum_{m,n=0}^{N-1}
{(t^2+s^2)^{ |m-n|/2 } \over n! m! 2^{|m-n|}}
\bigg ( {{\rm max}(m,n)! \over |m-n|!} \,{}_1F_1\Big ({\rm max}(m,n)+1,|m-n|+1;- {t^2 + s^2 \over 4 }\Big ) \bigg )^2. 
\end{multline}
In particular, 
$$
\lim_{N \to \infty} {\rm K}_N^{\rm c}( t, s) = 1 - e^{-(t^2+s^2)/4};
$$
cf.~(\ref{J2a}).
\end{proposition}
\noindent
Also of interest in the many body quantum chaos application  is the GinUE average of $|{\rm Tr} \, X^k|^2$ for positive integer $k$ \cite{SDH22}.

\begin{proposition}
We have
$$
\Big \langle |{\rm Tr} \, X^k|^2 \Big \rangle_{\rm GinUE} =
{1 \over (k+1) (N-1)!}
\Big ( (k+N)! - {N! (N-1)! \over (N - k - 1)!} \Big ).
$$
In particular
$$
\lim_{N,k \to \infty \atop k/N=x} {1 \over k N^k} \Big \langle |{\rm Tr} \, X^k|^2 \Big \rangle_{\rm GinUE}=
{2 \sinh (x^2/2) \over x^2}.
$$
\end{proposition}

\begin{proof} (Sketch)
In \cite{SDH22} the  average is reduced to 
$\int_{\mathbb C} dz_1 \int_{\mathbb C} dz_2 \, \rho_{(2),N}(z_1,z_2) z_1^k \bar{z}_1^k$. Earlier, the evaluation of a more general quantity was given in \cite[Corollary 4]{FR09}.
\end{proof}

We conclude this subsection with a brief account of the use of the GinUE in an ensemble theory of Lindblad dynamics \cite{Ca19,D+19,SRP20a}. This relates to the evolution of the density matrix $\rho_t$ for an $N$-level dissipative quantum system in the so-called Markovian regime, specified by the master equation $\dot{\rho}_t = \mathcal L (\rho_t)$. Here the operator $\mathcal L$ assumes a special structure identified by Lindblad \cite{Li76}, and by Gorini, Kossakowski, and  Sudarshan \cite{GKS76}. Specifically $\mathcal L$ consists of the sum of two terms, the first corresponding to the familiar unitary von Neumann evolution, and the second to a dissipative part, being the sum over operators $D_L$ (referred to as simple dissipators), represented as $N^2 \times N^2$ matrices according to
$$
D_L = 2 L \otimes_T L^\dagger - L^\dagger L \otimes_T  \mathbb I_N - \mathbb I_N \otimes_T L^\dagger L,
$$
for some $N \times N$ matrix $L$. Here $A \otimes_T B := A \otimes B^T$, where $\otimes$ is the usual Kronecker product.
In an ensemble theory, there is interest in $F_N(t):= {1 \over N^2} \langle {\rm Tr} \, e^{t D_L } \rangle_L$ \cite{Ca19}.

\begin{proposition}
Let $L$ be chosen from GinUE with global scaling. We have
$$
\lim_{N \to \infty} F_N(t)
= e^{-4t} \Big ( I_0(2t) + I_1(2t) \Big )^2.
$$
\end{proposition}

\begin{proof} (Sketch)
Following Can \cite{Ca19}, using a diagrammatic calculus, it is first demonstrated that
$$
\lim_{N \to \infty} {1 \over N^2} \langle D_L^k \rangle =
\lim_{N \to \infty} {(-1)^k \over N^2}
\Big \langle \Big ( L^\dagger L 
 \otimes_T  \mathbb I_N + \mathbb I_N \otimes_T L^\dagger L \Big )^k \Big \rangle_{L \in {\rm GinUE}}.
 $$
 The average on the RHS, in terms of the eigenvalues $\{x_j\}$  of $L^\dagger L$,  reads
 $\langle \sum_{j,l=1}^N (x_j + x_l)^k \rangle_{L^\dagger L}$ and consequently
 $$
 \lim_{N \to \infty} F_N(t)=
 \lim_{N \to \infty} 
 \Big \langle {1 \over N^2}
 \sum_{j,l=1}^N e^{-t(x_j + x_l)} \Big \rangle_{L^\dagger L} = \lim_{N \to \infty} 
\bigg ( \Big \langle  {1 \over N}
 \sum_{j=1}^N e^{-tx_j} \Big \rangle_{L^\dagger L} \bigg )^2.
 $$
The latter is the mean of a linear statistic in the ensemble $\{ L^\dagger L \}$ (complex Wishart matrices; see e.g.~\cite[\S 3.2]{Fo10}). Using the Marchenko-Pastur law for the global density of this ensemble 
(see e.g.~\cite[\S 3.4.1]{Fo10}), the stated result follows.
\end{proof}

\begin{remark} (Classification of non-Hermitian matrices)
It was commented in the Introduction that, in distinction to Dyson's viewpoint based on symmetry considerations, Ginibre's study \cite{Gi65} was no similarly motivated. Nowadays however, it is recognised that a symmetry viewpoint is fundamental to topological driven effects in non-Hermitian quantum physics \cite{AGU20}. Starting with \cite{BL02,Ma07} and continuing in \cite{KSUS19}, a classification scheme based on symmetries with respect to the involutions of transpose, complex conjugation and Hermitian conjugation, and in which the (anti-)commutation relation involves unitary matrices satisfying certain quadratic relations in terms of these involution, has been given. For example, defining the block unitary matrix $P={\rm diag} \, (\mathbb I_N, - \mathbb I_N)$, and requiring that the matrix ensemble $\{A \}$ have the (anti-)symmetry $A = - P A P$, gives that each $A$ has the form
\begin{equation}\label{S.1w}
A = \begin{bmatrix} 0_{N \times N} & X \\ Y & 0_{N \times N} \end{bmatrix}
\end{equation}
for some square matrices $X,Y$. Denoting the eigenvalues of the matrix product $XY$ as $\{ -z_j^2 \} $, one sees that the eigenvalues of $A$ are $\{ \pm i z_j \} $.

In keeping with the viewpoint of this subsection, a basic question are signatures of the symmetry in the eigenvalue spectrum. For example, in (\ref{S.1w}), with $X,Y$ GinUE matrices, are the bulk scaled eigenvalues of $A$ statistically distinct from individual GinUE matrices? We know from the results quoted in the paragraph above Remark \ref{R2.17} that the answer in this case is no. However the answer to this question is yes, if instead the symmetry is that $A = A^T$, for the independent entries of $A$ standard complex Gaussians. This was demonstrated in \cite{HKKU20} by a numerical study of the nearest neighbour spacing distribution, and the relevance to Lindblad dynamics discussed.
\end{remark}

\subsection{Singular values}\label{S6.3}
One recalls that for a complex square matrix $X$ the squared singular values are the eigenvalues of $X^\dagger X$. For a general ensemble of non-Hermitian matrices $\{X\}$, motivation to study the singular values comes from various viewpoints. For example, in Remark \ref{R2.17}.4, singular values (specifically of product matrices) appeared in the context of Lyapunov exponents. As other example, one recalls that plus/minus of the singular values are the eigenvalues of the $2N \times 2N$ Hermitian matrix
$$
H = \begin{bmatrix} 0_{N \times N}&X \\ X^\dagger & 0_{N \times N}
\end{bmatrix}.
$$
The importance of this in relation to the eigenvalues of $X$ is that resolvent associated with $H$ is fundamental to the study of the circular law for the spectral density beyond the Gaussian case; see e.g.~\cite[\S 4.1]{BC12}. Another piece of theory is that the condition number $\kappa_N$ associated with $X$ is equal to the ratio of the smallest to the largest singular value \cite{Ed88}.   And from the identity $| \det X| = | \det X^\dagger X|^{1/2}$ the distribution of the modulus of $\det X$ is determined by the singular values.

For the GinUE, the squared singular values $\{s_j\}_{j=1}^N$ say are known to have for their joint distribution a PDF proportional to
\begin{equation}\label{cs1}
\prod_{j=1}^N e^{-s_j} \prod_{1 \le j < k \le N} (s_k - s_j)^2, \quad s_j \in \mathbb R_+;
\end{equation}
see e.g.~\cite[Prop.~3.2.2 with $\beta = 2$, $n=m=N$]{Fo10}. After scaling by $N$, almost surely the largest squared singular value has the limiting value $4$  \cite{PS11}. However, after the same scaling, a simple change of variables in (\ref{cs1}) integrated from $(s,\infty)$ in each variable reveals that the smallest singular value is an exponential random variable with rate parameter $N^2$. Putting these facts together implies that for large $N$, $\kappa_N/N$ is distributed according to the heavy tailed distribution with PDF ${8 \over x^3} e^{-4/x^2}\chi_{x > 0}$ \cite{Ed88}. Also, for $n \times N$ ($n \ge N$) rectangular GinUE matrices, it is proved in \cite{CD05} that $\langle \log \kappa_N \rangle < {N \over |n-N|+1} +2.24$, for any $N \ge 2$.

Let $P_N(t)$ denote the PDF for the distribution of $|\det X|^2$ for GinUE matrices. Making use of knowledge of the PDF of squared singular values (\ref{cs1}) shows that the Mellin transform of $P_N(t)$ is equal to the multiple integral
\begin{equation}\label{cs2}
{1 \over C_N} \int_0^\infty ds_1 \cdots \int_0^\infty ds_N \, \prod_{j=1}^N s_j^{s-1} e^{-s_j} \prod_{1 \le j < k \le N} (s_k - s_j)^2 =
\prod_{j=0}^{N-1}{ \Gamma(s+j) \over \Gamma(1+j)}.
\end{equation}
Here the normalisation $C_N$ is such that the expression equals unity for $s=1$, while the evaluation of the multiple integral follows as a special case of the Laguerre weight Selberg integral; see e.g.~\cite[Prop.~4.7.3]{Fo10}. As noted in \cite[Eq.~(2.17)]{FZ18}
(see also \cite[Prop.~2.2]{Ro07}), it follows immediately from this that
\begin{equation}\label{cs3}
| \det X |^2 \mathop{=}\limits^{\rm d} \prod_{l=1}^N {1 \over 2} \chi_{2l}^2.
\end{equation}
In words this says that the absolute value squared of the determinant of GinUE matrices is equal in distribution to the product of $N$ independent chi-squared distributions, with degrees of freedom $2,4,\dots,2N$, each scaled by a factor of 2.
Starting from (\ref{cs3}), and defining the global scaled GinUE matrices $X^{\rm g}$ by $X^{\rm g} = {1 \over \sqrt{N}} X$, the distribution of $\log |\det X^{\rm g}|^2$ can be shown to have leading order mean $-N$, variance $\log N$, and after recentring and rescaling satisfy a central limit theorem \cite[Th.~3.5]{Ro07}. For a general linear statistic $\sum_{j=1}^N f(z_j)$ of global scaled GinUE matrices, the leading order mean is ${N \over \pi} \int_{|z|<1} f(z) \, d^2z$. For $f(z) = \log |z|^2$, this gives the stated value of $-N$. Also, we notice that substituting this choice of $f(z)$ in the variance formula implied by (\ref{5.2e}) gives ${1 \over \pi} \int_{|\mathbf r|<1} {1 \over x^2+y^2} \, dx dy$, which is not integrable at the origin, in keeping with the variance actually diverging as $\log N$.

There is an alternative viewpoint on the result (\ref{cs3}) which does not require knowledge of the joint distribution of the singular values (\ref{cs1}), nor the evaluation of the multiple integral (\ref{cs2}). The idea, used in both \cite{Ro07,FZ18} and which goes back to Bartlett \cite{Ba34} in the case of real Gaussian matrices, is to decompose $X$ in terms of its QR (Gram-Schmidt) decomposition. The matrix of orthonormal vectors $Q$ constructed from the columns  of $X$ will for $X \in {\rm GinUE}$, be a Haar distributed unitary matrix, which we denote by $U$. The matrix $R=[r_{jk}]_{j,k=1}^N$ is upper triangular with diagonal elements real and positive. One notes
\begin{equation}\label{cs4b}
\det X^\dagger X = \prod_{j=1}^N r_{jj}^2, 
\end{equation}
and so it suffices to have knowledge on the distribution of $\{r_{jj}\}_{j=1}^N$ for $X$. 

\begin{proposition}\label{P6R}
Let $\{r_{jj}\}_{j=1}^N$ denote the diagonal elements in the QR decomposition of a GinUE matrix $X$. We have
\begin{equation}\label{cs4a}
r_{jj}^2 \mathop{=}\limits^{\rm d}  {1 \over 2} \chi_{2j}^2.
\end{equation}
\end{proposition}

\begin{proof}
The QR decomposition $X = UR$  gives the corresponding decomposition of measure
(see e.g.~\cite[Prop.~3.2.5]{Fo10})
$$
(dX) = \prod_{j=1}^N r_{jj}^{2(N-j)+1} (dR) (U^\dagger U),
$$
where as anticipated $(U^\dagger U)$ is recognised as Haar measure on the space of complex unitary matrices. The element distribution of GinUE matrices is 
proportional to $e^{-{\rm Tr} \, X^\dagger X} = e^{- \sum_{1 \le j \le k \le N} |r_{jk}|^2}$. The various factorisations implies that integrating over $U$ and the off diagonal elements of $R$ only changes the normalisation. We then read off that each $r_{jj}$ has a distribution with PDF proportional to $r^{2(N-j)+1}e^{-r^2}$, which implies (\ref{cs4a}).
\end{proof}

Using (\ref{cs4a}) in (\ref{cs4b}) reclaims (\ref{cs3}).

\begin{remark}
1.~Since with $\{z_j\}$ the eigenvalues of $X$, $| \det X |^2 = \prod_{j=1}^N | z_j|^2$, the fact that the Mellin transform of the distribution of this quantity is given by the product of gamma functions in (\ref{cs2}) implies
\begin{equation}\label{L.1}
    \Big \langle \prod_{l=1}^N |z_l|^{2(s-1)} \Big \rangle_{\rm GinUE}^{\rm g} =
    N^{N(s-1)}\prod_{j=0}^{N-1} {\Gamma(s+j) \over \Gamma(1+j)}.
    \end{equation}
    Here the superscript "g" indicates the use of global scaling coordinates $z_l \mapsto \sqrt{N} z_l$. 
 We observe that knowledge of the induced GinUE normalisation $C_{n,N}$ in Proposition \ref{P2.8a} provides a direct derivation of  (\ref{L.1}). For large $N$ this ratio of gamma functions can be written in terms of the Barnes $G$-function according to ${G(N+s) \over G(N+1) G(s)}$; see \cite[Eq.~(4.183)]{Fo10}. Known asympotics for ratios of the Barnes $G$-function (see e.g.~\cite[Eq.~(4.185)]{Fo10} then gives that for large $N$, and with $s = \gamma/2+1$ for convenience,
 \begin{equation}\label{L.1a}
    \Big \langle \prod_{l=1}^N |z_l|^{\gamma} \Big \rangle_{\rm GinUE}^{\rm g} \sim N^{\gamma^2/8} e^{-(\gamma/2)N} {(2 \pi)^{\gamma/4} \over G(1 + \gamma/2)}.
    \end{equation}
    This is the special case $z=0$ of an asymptotic formula for $ \langle \prod_{l=1}^N |z-z_l|^{\gamma}  \rangle_{\rm GinUE}^{\rm g}$ given by Webb and Wong \cite[Th.~1.1]{WW19}.
    \\
    2.~It is a standard result in random matrix theory (see e.g.~\cite{PS11}) that the density of singular values in (\ref{cs1}), after the global scaling $s_j \mapsto s_j N$, as the particular Marchenko-Pastur form $\rho^{|rm MP}_{(1),\infty}(x) = {1 \over 2 \pi} ({4 - x \over x} )^{1/2} \chi_{0 < x < 4}$. The $k$-th moment of the density is given in terms of the particular mixed moment of a global scaled Ginibre matrix $\tilde{G}$, $\langle {\rm Tr} (\tilde{G}^\dagger G)^k \rangle$. The calculation of these moments for large $N$ relates to free probability --- see the recent introductory text \cite{PB20} for the main ideas --- and to combinatorics as is seen from the fact that $\int_0^4 x^k \rho^{\rm MP}_{(1),\infty}(x) \, dx= C_k$, where $C_k$ denotes the $k$-th Catalan number. Works on mixed moments of Ginibre matrices include \cite{WS15,DF20,HL20,De22}. \\
    3.~The squared singular values as specified by the PDF (\ref{cs1}) form a determinantal points process, being a special case of the classical Laguerre unitary ensemble; see e.g.~\cite[Chapters.~3 and 5]{Fo10}. This is similarly true of the squared singular values of the various extensions of GinUE considered above: for example in the case of the spherical model and truncated unitary matrices, it is the classical Jacobi unitary ensemble which arises, while the singular values of products of GinUE matrices, or of truncated unitary matrices, gives rise to a class of determinantal point processes called P\'olya ensembles \cite{KS14,KKS16,KK19,FKK21}.  A notable exception is the singular values of elliptic GinUE matrices, which form a Pfaffian point process \cite{KK18}.
    
 \end{remark}
 
\subsection{Eigenvectors}
Associated with the set of eigenvalues $\{\lambda_j\}$ of a Ginibre matrix $G$ are two sets of
eigenvectors --- the left eigenvectors $\{ \boldsymbol \ell_j \}$ such that $ \boldsymbol \ell_j^T G = \lambda_j  \boldsymbol \ell_j^T$,
and the right eigenvectors $\{ \mathbf r_j \}$ such that $G \mathbf r_j = \lambda_j  \mathbf r_j$.
These are not independent, but rather (upon suitable normalisation), form a biorthogonal set
\begin{equation}\label{3.1a}
 \boldsymbol \ell_i^T \mathbf r_j = \delta_{i,j}. 
\end{equation}
This property follows from the diagonalisation formula
$G = X D X^{-1}$, 
where $X$ is the matrix of right eigenvectors, $D$ the diagonal matrix of eigenvectors,
and $X^{-1}$ identified as the matrix of left eigenvectors.
For nonzero scalars $\{c_i\}$ we see that (\ref{3.1a}) is unchanged by the rescalings $ \mathbf r_j  \mapsto c_j  \mathbf r_j $
and $ \boldsymbol \ell_j   \mapsto (1/c_j)  \boldsymbol \ell_j $. 

For $N \times 1$ column vectors $\mathbf u, \mathbf v$, define the inner product $\langle \mathbf u,  \mathbf v \rangle :=
\bar{ \mathbf u}^T  \mathbf v $.  The so called overlap matrix has its elements $\mathcal O_{ij}$ expressed in terms
of this inner product according to
\begin{equation}\label{3.1b}
\mathcal O_{ij} := \langle \boldsymbol \ell_i,  \boldsymbol \ell_j \rangle  \langle \mathbf r_i,  \mathbf r_j \rangle.
\end{equation}
Note that this is invariant under the mappings noted in the final sentence of the above paragraph, and for fixed
$i$ and summing over $j$ gives $1$. Also, it follows from (\ref{3.1b}) that the diagonal entries relate to the lengths
\begin{equation}\label{3.1c}
\mathcal O_{jj} = ||  \boldsymbol \ell_j ||^2  ||  \mathbf r_j ||^2.
\end{equation}
The square root of this quantity is known as the eigenvalue condition number; see the introduction to \cite{BD20}
and \cite[\S 1.1]{CR22}
for further context and references. Significant too is the fact that the overlaps (\ref{3.1c}) appear in the specification
of a Dyson Brownian motion extension of GinUE \cite{GW18,BD20}.

Statistical properties of $\{ O_{ij} \}$ for GinUE were first considered by Chalker and Mehlig \cite{CM98,CM00}.
By the Schur decomposition (\ref{GS}), instead of a GinUE matrix $G$, we may consider an upper triangular matrix
$Z$ with the eigenvalues $\{z_j\}$ of $G$ on the diagonal, and off diagonal entries standard complex Gaussians.
For the eigenvalue $\lambda_1$, the triangular structure shows that $\boldsymbol \ell_1 =(1,b_2,\dots,b_N)^T$
and $\mathbf r_1=(1,0,\dots,0)^T$ where for $p > 1$ and $b_1 = 1$, $b_p = {1 \over z_1 - z_p} \sum_{q=1}^{p-1}
b_q Z_{pq}$. From this last relation, it follows that with $\boldsymbol \ell_1^{(n)}  =(1,b_2,\dots,b_n)^T$ for
$n < N$ we have
\begin{equation}\label{3.1d}
|| \boldsymbol \ell_1^{(n+1)} ||^2 = || \boldsymbol \ell_1^{(n)} ||^2 \Big ( 1 + {1 \over | z_1 - z_{n+1}|^2} \Big |
\sum_{q=1}^n \tilde{b}_q Z_{(n+1)q} \Big |^2 \Big ), \qquad \tilde{b}_q := {b_q \over \sqrt{\sum_{q=1}^n |b_q|^2} }.
\end{equation}
This has immediate consequence in relation to $\mathcal O_{11} $  as shown by 
Bourgade and Dubach \cite{BD20}.

 \begin{proposition}\label{P3.1+}
 Let the eigenvalues $\{z_j \}$ be given. We have
 \begin{equation}\label{3.1e}
 \mathcal O_{11} \mathop{=}\limits^{\rm d} \prod_{n=2}^N \Big ( 1 + { |X_n|^2 \over | z_1 - z_n|^2} \Big ),
 \end{equation}
 where each $X_n$ is an independent complex standard Gaussian. Furthermore, it follows from this that after averaging over $\{z_2,\dots,z_N\}$
 \begin{equation}\label{3.1f}
  \mathcal O_{11} \Big |_{z_1=0}  \mathop{=}\limits^{\rm d}  {1 \over {\rm B}[2,N-1]},
 \end{equation}  
 where ${\rm B}[\alpha,\beta]$
   refers to the beta distribution. 
 \end{proposition}
 
 \begin{proof}
 A product formula for $\mathcal O_{11}$ follows from (\ref{3.1c}), the fact that $||  \mathbf r_1 || = 1$, and by iterating
 (\ref{3.1d}). This product formula is identified with the RHS of (\ref{3.1e})  upon noting that a vector of independent standard
 complex Gaussians dotted with any unit vector (here $(\tilde{b}_1,\dots, \tilde{b}_n)$) has distribution equal to a standard complex Gaussian.
 
 In relation to (\ref{3.1f}) a minor modification of the proof of Proposition \ref{P2.18} shows that conditioned on $z_1 = 0$, the
 ordered squared moduli $\{ | z_j |^2 \}_{j=2}^N$ are independently distributed as $\{ \Gamma[j;1] \}_{j=2}^N$. Noting too
 that with each $X_j$ a standard complex Gaussian, $|X_j|^2  \mathop{=}\limits^{\rm d} \Gamma[1;1]$ it follows
 $$
   \mathcal O_{11} \Big |_{z_1=0}  \mathop{=}\limits^{\rm d}  \prod_{n=2}^N \Big ( 1 + {\tilde{X}_n \over Y_n} \Big ), \qquad
   \tilde{X}_n  \mathop{=}\limits^{\rm d} \Gamma[1;1], \: Y_n  \mathop{=}\limits^{\rm d} \Gamma[n;1].
   $$
   Next, we require knowledge of the standard fact that $Y_n/(Y_n + \tilde{X}_n)  \mathop{=}\limits^{\rm d}  {\rm B}[n,1]$.
   Furthermore (see e.g.~\cite[Exercises 4.3 q.1]{Fo10}), for $x  \mathop{=}\limits^{\rm d} {\rm B}[\alpha+\beta,\gamma]$,
  $y  \mathop{=}\limits^{\rm d} {\rm B}[\alpha ,\beta]$, we have that $xy   \mathop{=}\limits^{\rm d} {\rm B}[\alpha ,\beta+\gamma]$, which tells us
  that with $b_n  \mathop{=}\limits^{\rm d}  {\rm B}[n,1]$ we have $\prod_{n=2}^N b_n  \mathop{=}\limits^{\rm d}  {\rm B}[2,N-1]$.
   
 \end{proof}

 Dividing both sides of (\ref{3.1f}) 
 by $N$ we see that the $N \to \infty$ is well defined since $N {\rm B}[2,N-1] \to \Gamma[2,1]$. 
 After scaling the GinUE matrix $G \mapsto G/\sqrt{N}$ so that the leading eigenvalue support is the unit disk,
an analogous
 limit formula has been extended from $z_1 = 0$ to any $z_1 = w$, $|w| < 1$ in \cite{BD20}. Thus
\begin{equation}\label{3.1g}
 { \mathcal O_{11} \Big |_{z_1=w_1}  \over N (1 - |w_1|^2) }
 \mathop{\to }\limits^{\rm d}  {1 \over  \Gamma [2,1]}.
 \end{equation} 
 In words, with $ \mathcal O_{11} $ corresponding to the condition number, one has that the instability of the spectrum is of order $N$
 and is more stable towards the edge. Another point of interest is that the PDF for $1/    \Gamma [2,1]$ is $\chi_{t > 0} e^{-1/t}/t^3$, which
 is heavy tailed, telling us that only the zeroth and first integer moments are well defined. We remark that limit theorems of the universal form (\ref{3.1g}) have been proved in the case of the complex spherical ensemble
 of \S 2.5, and for a sub-block of a Haar distributed unitary matrix \cite{Du21}; see also \cite{NT18}.
 
 The $1/t^3$ tail implied by (\ref{3.1g}) has been exhibited from another viewpoint in the work of Fyodorov \cite{Fy18}. There the joint PDF for the overlap non-orthogonality $\mathcal O_{jj} -1$, and the eigenvalue position $z_j$, was computed for finite $N$. The global scaled limit of this quantity, $\mathcal P^{\rm g}(t,w)$ say, was evaluated as
 \cite[Eq.~(2.24)]{Fy18}
 \begin{equation}\label{3.1g+}
 \mathcal P^{\rm g}(t,w) =
 {(1 - |w|^2)^2 \over \pi t^3}
 e^{-(1-|w|^2)/t}, \quad |w|<1.
 \end{equation}
 Note that for the first moment in $t$ this gives
 \begin{equation}\label{3.1gE} 
 \int_0^\infty t \mathcal  P^{\rm g}(t,w) \, dt = {1 \over \pi} (1 - |w|^2),
 \end{equation}
 in keeping with a prediction from 
 \cite{CM98,CM00}. This was first proved in  \cite{WS15}.

 An explicit formula for the large $N$ form of the average of the overlap (\ref{3.1b}) in the off diagonal case (say $(i,j)=(1,2)$),
 with  the GinUE matrix scaled $G \mapsto G/\sqrt{N}$, and 
 conditioned on $z_1 = w_1$, $z_2 = w_2$ with $|w_1|, |w_2| < 1$  is also known \cite{CM98,CM00,BD20,ATTZ20,CR22}
\begin{equation}\label{3.1h} 
\left \langle \mathcal O_{12} \Big |_{z_1 = w_1, z_2 = w_2} \right \rangle \mathop{\sim}\limits_{N \to \infty} -
{1 \over N} {1 - w_1 \bar{w}_2 \over | w_1 - w_2|^4} \bigg (
{1 - (1 + N |w_1 - w_2|^2) e^{- N | w_1 - w_2|^2} \over 1 - e^{- N | w_1 - w_2|^2} } \bigg ).
\end{equation}
This formula is uniformly valid down to the scale $N | w_1 - w_2| = {\rm O}(1)$.
The large $N$ form of the average value of the diagonal overlap for products of $M$ global scaled GinUE matrices has been considered in \cite{BSV17,BNST17}, with the result
\begin{equation}\label{3.1k}
\lim_{N \to \infty}{1 \over N} \left \langle  \mathcal O_{11} \Big |_{z_1 = w} \right \rangle
= {1 \over \pi} |z|^{-2+2/M} (1 - |z|^{2/M}) \chi_{|z|<1};
\end{equation}
cf.~(\ref{7.4b}).
We remark that the average values of 
 $\mathcal{O}_{11}$ and $\mathcal{O}_{12}$, conditioned on multiple eigenvalues are shown to have a determinantal form in \cite[Th.~1]{ATTZ20}, thus exhibiting 
  an integrable structure for these eigenvector statistics.

For general complex non-Hermitian matrices with independently distributed entries of the form $\xi_{jk} + i {\zeta}_{jk}$, 
where each $\xi_{jk}, \tilde{\zeta}$ is an identically distributed zero mean real random  variable of unit variance (this class is sometimes referred to as complex non-Hermitian Wigner matrices; see e.g.~\cite{AGN21}), a line of research in relation to the normalised eigenvectors is to quantify the similarity with a complex vector drawn from the sphere embedded in $\mathbb C^N$ with uniform distribution. Recent references on this include \cite{LO20,LT20}. For random vectors on the sphere, there are bounds on the size of the components which rule out gaps is the spread of the size of the components, referred to in \cite{RV16} as no-gaps localisation. As noted in \cite[\S 1.1]{LT20}, the bi-unitary invariance of GinUE matrices implies individual eigenvectors are distributed uniformly at random from the complex sphere, and thus with probability close to one have that the $j$-th largest modulus of the entries is bounded above and below by a positive constant times $\sqrt{N-j}/N$, for $j$ in the range from $N/2$ up to $N$ minus a constant time $\log N$.

\subsection*{Acknowledgements}
This research is part of the program of study supported by the Australian Research Council Discovery Project grant DP210102887.
SB was partially supported by the National Research Foundation of Korea grant NRF-2019R1A5A1028324, Samsung Science and Technology Foundation grant SSTF-BA1401-51, and KIAS Individual via the Center for Mathematical Challenges at Korea Institute for Advanced Study grant SP083201.	
The authors gratefully acknowledge Gernot Akemann, Christophe Charlier, Yan Fyodorov, Grégory Schehr and Aron Wennman for helpful feedback on the first draft of this work.

\nopagebreak

\providecommand{\bysame}{\leavevmode\hbox to3em{\hrulefill}\thinspace}
\providecommand{\MR}{\relax\ifhmode\unskip\space\fi MR }
\providecommand{\MRhref}[2]{%
  \href{http://www.ams.org/mathscinet-getitem?mr=#1}{#2}
}
\providecommand{\href}[2]{#2}

 \end{document}